\documentclass[12pt]{article}
\usepackage{eurosym}
\usepackage{amssymb}
\usepackage{amsfonts}
\usepackage{amsmath}
\usepackage{graphicx}
\usepackage{enumerate}
\usepackage{ragged2e}
\usepackage{url}
\usepackage[toc,title,page]{appendix}
\usepackage{multirow}

\RequirePackage[colorlinks,citecolor=blue,urlcolor=blue]{hyperref}

\newtheorem{theo}{Theorem}
\newtheorem{lemma}{Lemma}
\newtheorem{proposition}{Proposition}

\newtheorem{rem}{Remark}

\def\refhg{\hangindent=20pt\hangafter=1}

\def\refhg{\hangindent=20pt\hangafter=1}

\def\refhg{\hangindent=20pt\hangafter=1}

\input{tcilatex}
\begin{document}

\date{}
\title{Random multiplication versus random sum: auto-regressive-like models
with integer-valued random inputs}
\author{Abdelhakim Aknouche\\
 Department of Mathematics, College of Science\\
Qassim University, Al Qassim, Saudi Arabia\\
email: \href{aknouche_ab@yahoo.com}{aknouche\_ab@yahoo.com}
 \and  S\'onia Gouveia\thanks{corresponding author}\\
  Institute of Electronics and Informatics Engineering of Aveiro (IEETA)\\  Department of Electronics, Telecommunications and Informatics (DETI)\\
 University of Aveiro, Portugal\\ email: \href{mailto:sonia.gouveia@ua.pt}{sonia.gouveia@ua.pt} \and Manuel G. Scotto\\ 
Center for Computational and Stochastic Mathematics (CEMAT)\\
Department of Mathematics, IST\\
University of Lisbon, Portugal\\
email: \href{mailto: manuel.scotto@tecnico.ulisboa.pt}{manuel.scotto@tecnico.ulisboa.pt}
}
\maketitle

\begin{abstract}
A common approach to analyze count time series is to fit models based on random sum operators.
As an alternative, this paper introduces time series models based on a
random multiplication operator, which is simply the multiplication of a
variable operand by an integer-valued random coefficient, whose mean is
the constant operand. Such operation is endowed into auto-regressive-like models with integer-valued random inputs, addressed as RMINAR. Two special
variants are studied, namely the $\mathbb{N}_0$-valued\ random coefficient auto-regressive model and the $\mathbb{N}_0$-valued random coefficient multiplicative error model. Furthermore, $\mathbb{Z}$-valued extensions are considered. The dynamic
structure of the proposed models is studied in detail. In particular, their corresponding solutions are everywhere strictly stationary and ergodic, a fact that is not common neither in the literature on integer-valued time series models nor real-valued random coefficient auto-regressive models. Therefore, the parameters of the RMINAR model 
are estimated using a four-stage weighted least squares estimator, with consistency and asymptotic normality established everywhere in the parameter space. Finally, the new RMINAR models are
illustrated with some simulated and empirical examples.\\

\end{abstract}

\section{Introduction}

Modeling low integer-valued time series data is, nowadays, an ongoing concern in time
series research. To this end, three common approaches are generally undertaken. The first assumes a discrete conditional distribution whose
conditional mean is a parametric function of past observations. The
resulting models are known as 
observation-driven in the terminology of Cox ($1981$). The most known examples are the integer-valued
Generalized ARMA model (Zeger and Qaqish, $1988$; Benjamin {\it et al.}, $2003$; Zheng {\it et al.}, $2015$) and, in particular, the integer-valued GARCH (INGARCH) model (e.g.~Reydberg and
Shephard, $2000$; Heinen, $2003$; Ferland {\it et al.}, $2006$; Fokianos
{\it et al.}, $2009$; Zhu, $2011$; Davis and Liu, $2016$; Aknouche and Francq 2021). The primary feature of observation-driven models is that the likelihood function is explicit in terms of observations, which turns maximum likelihood estimation, inference, and forecasting quite easy to investigate. Other
M-estimation methods such as quasi-maximum likelihood (QML) and weighted
least squares (WLS) estimators are also straightforward to derive. However,
observation-driven models suffer from some limitations, namely that they
are often fully parametric and hence not robust to a distributional
misspecification. Moreover, their probabilistic structures (e.g.~ergodicity,
tail behavior, extremal properties) are inherently complex since they are not
defined through equations driven by independent and identically distributed
(iid) innovation sequences.

The second approach is addressed as parameter-driven (Cox, $1981$) and shares with the first approach the fact that also requires the specification of a discrete conditional
distribution for the data. Nonetheless, this distribution is rather
conditioned on a latent process since the conditional mean has a proper
dynamics in terms of its past (latent) values (Zeger, $1988$; Davis and
Rodriguez-Yam, $2005$; Davis and Wu, $2009$). Parameter-driven models have in general simple probability
structures and are quite flexible to represent dynamic dependence of count
data (e.g.~Davis and Dunsmuir, $2016$). In particular, the conditional mean
depends on present shocks unlike INGARCH models for which the conditional intensity only
depends on past observations. Due to the latent process, however, their estimation is rather difficult because the likelihood function
cannot in general be obtained in a closed form and involves
cumbersome multidimensional integration. Moreover, QML and WLS estimators are
not simple to derive either, since the conditional mean is not explicit in
terms of observations. This difficulty also arises in inference and
prediction which explains why parameter-driven count models have received less attention than observation-driven ones.

The third approach considers appropriate stochastic difference equations driven by iid inputs whose solutions are integer-valued sequences. The main concern of these models is to handle integer-valued random operations on inputs to produce integer-valued outputs that have similar features to real-world
integer-valued data. Random sum operations, aka thinning operations,
are the best-known examples. The rule is that, given a positive constant and an integer-valued random variable as operands, a random sum operation is the sum of iid discrete variables whose mean and number are the
constant and variable operands, respectively. In particular, the binomial thinning
operator (Steutel and van Harn, $1978$) produces a binomial distributed
variable that is bounded from above 
by the operand variable. Random sum operators not satisfying the latter feature are still called generalized thinning, e.g.~Poisson and negative binomial random sums (see e.g.~Scotto {\it et al.}, $%
2015$).

The most elegant property of random sum-based equations is that the
marginal\ distribution of the output sequence is readily known and depends
on the operator and input distributions. For instance, the first order
integer-valued auto-regressive process (INAR$\left( 1\right) $; McKenzie, $%
1985$; Al-Osh and Alzaid, $1987$) based on binomial thinning and driven by a
Poisson iid innovation has a Poisson marginal. Since any random
sum involves unobserved summands, the likelihood calculation is cumbersome
and requires high dimensional summations, just like parameter-driven models.
This is the reason why some authors classify thinning-based models as parameter-driven
(e.g.~Ahmad and Francq, 2016; Aknouche and Francq, 2023). Furthermore, the conditional mean and variance are parametric
functions of past observations as in observation-driven models. In fact,
the latent variables in a random sum only intervene in model's
conditional distribution through their mean and variance as being unknown constant
parameters. The same happens for random coefficient auto-regressive models
(RCAR; Nicholls and Quinn, $1982$). In this respect, thinning-based models
are more similar to observation-driven models. In particular, QML and WLS
estimators are quite easy to derive and analyze. Thus, thinning-based and RCAR models can be seen as semi- (or partially-) observation-driven
models.

Numerous thinning-based count models have been emerged so far (e.g.~Scotto
{\it et al.}, 2015; Weiss, 2018). They differ fundamentally in the
assumed distributions of the summands or/and the form of the stochastic
equations. Note that the implicit form of the random sum in terms of the
operand variable makes the study of the corresponding equation more complex
than in conventional stochastic equations (e.g.~ARMA, GARCH, RCAR). Most
existing thinning specifications are based on simple pure auto-regressions
(INAR) or moving averages (INMA) with low orders. Only a few research works
deal with integer-valued ARMA (INARMA)\ equations or their multivariate forms due
to their inherent complexity. In particular, invertibility, tails behavior,
asymptotic properties of maximum likelihood estimators, and forecasting
remain an issue for INARMA-like models. Many other statistical
aspects of INARMA models are not as developed as in standard continuous-valued ARMA models.
Also, non-linear forms are not explored as much except in a few special
cases. This is why most thinning-based models are unable to reproduce
various interesting features such as high over-dispersion, multi-modality,
persistence, etc.~(e.g.~Aknouche and Scotto, $2024$).

As a simple alternative to random sum operators, this paper proposes a
random multiplication operator and shows how to build on it simple and
analytically tractable integer-valued time series models. An $%
\mathbb{N}_0
$-valued 
random multiplication operator is in fact a random sum with identical summands and constitutes the direct multiplication of an operand variable by an integer-valued variable whose expected value (defined in $\mathbb{R}^{+}$) is precisely the constant operand.
An extended $\mathbb{Z}$-valued 
random multiplication makes it possible to deal with the constant and variable operands defined in $\mathbb{R}$ and $\mathbb{Z}$, respectively.
Compared to thinning operators, the random multiplication
is analytically simpler and can produce variables with higher volatility even with Poisson multipliers. Actually, a random multiplication-based model is
nothing but a model with integer-valued random inputs (coefficients and
innovations). In particular, a random multiplication-based auto-regressive
model with $\mathbb{N}_0$-valued inputs (henceforth RMINAR) is a special case of the RCAR model but with $\mathbb{N}_0$-valued random inputs. Likewise, a $\mathbb{Z}$-valued RMINAR model is an RCAR with $\mathbb{Z}$-valued random inputs.

Continuous-valued RCAR models have been widely studied since the late 1970s and most of their probabilistic and statistical properties are now well understood (e.g.~Nicholls and Quinn, $1982$; Tsay, $1987$; Schick, $1996$; Diaconis and
Freedman, $1999$; Aue {\it et al.}, $2006$; Aknouche, $2013$-$2015$; Aue and
Horvath, $2019$; Trapani, $2021$, Regis {\it et al.}, $2022$). Although
the proposed RMINAR model belongs to the general class of RCAR models, it holds
surprising features that differ from those of the continuous-valued counterpart and of the aforementioned integer-valued models. In particular, any RMINAR solution
is strictly stationary and ergodic regardless of its coefficients values, so it is useless testing strict stationary as
in the continuous-valued case (Aue and Horvath, $2011$; Aknouche, $2013$-$%
2015$). In other words, RMINAR models are everywhere stationary and ergodic
and thus can be strictly stationary with infinite means. As a
consequence, multi-stage WLS and QML estimators are consistent and asymptotically Normal everywhere in the parameter space and for all parameter components.

The rest of this paper is organized as follows. Section \ref{section2} sets up the random multiplication operator and shows its main properties. Section \ref{section3}
defines the RMINAR model for $%
\mathbb{N}
_{0}$-valued data, a $%
\mathbb{Z}
$-valued extension, and a multiplicative variant for $%
\mathbb{N}
_{0}$-valued data. Section \ref{4SWLSE} proposes four-stage WLS estimators (4SWLSE) for the mean and variance of the random inputs in the three models. A simulation study and two real applications with a $\mathbb{N}_{0}$-valued and a $\mathbb{Z}$-valued time series are given in Section \ref{section5}. Section \ref{section6} summarizes the conclusions of this work and the main proofs are left to an appendix. 

\section{Random operators: sum versus multiplication}
\label{section2}
This section overviews important properties of the random sum operator (RSO), denoted by $\circ _{s}$, and introduces the so-called random multiplication operator (RMO), denoted by $\odot _{m}$, highlighting fundamental differences between these operators. A few examples are given for the distribution of the random variable resulting from these random operations.
\subsection{Random sum operation}

Random sum operators also known as 
generalized thinning (Latour, 1998) are commonly
used in branching and INAR-like processes. In its general form a
RSO, denoted by $\circ _{s}$, is defined for any positive constant $%
\alpha $ and any integer-valued random variable $X$ by%
\begin{equation}
\alpha \circ _{s}X:=\sum\limits_{i=1}^{X}\xi _{i},  
\label{2.1}
\end{equation}%
where the integer-valued sequence $\left( \xi _{i}\right) $ is iid with mean 
$\mathbb{E}\left( \xi _{i}\right) =\alpha $ and variance $\sigma _{\xi }^{2}:=\mathbb{V}\left( \xi _{i}\right)
$.
It is further assumed that $X$ and $\left( \xi _{i}\right)$
are independent. The terms in $\left( \xi _{i}\right) $ are called summands or
counting series (or also offspring sequence) while $\alpha $ and $X$ are
the constant and variable operands, respectively. Thus, the operation $\alpha
\circ _{s}X$ is a random sum up to the operand variable $X$ where the common
mean of the iid summands $(\xi _{i})$ is the constant operand $%
\alpha$. The first two moments of $\alpha \circ
_{s}X$ are%
\begin{eqnarray}
\mathbb{E}\left(
\alpha \circ _{s}X\right) =\alpha \mathbb{E}\left( X\right) \text{ \ and \ }
\mathbb{E}\left( \alpha \circ _{s}X|X\right) &=&\alpha X,  
\label{2.2} \\
\mathbb{V}\left( \alpha \circ _{s}X\right) =\sigma _{\xi }^{2}\mathbb{E}\left( X\right)
+\alpha ^{2}\mathbb{V}\left( X\right)\text{ \ and \ }
\mathbb{V}\left( \alpha \circ _{s}X|X\right) &=&\sigma _{\xi }^{2}X.  
\label{2.3}
\end{eqnarray}%
An interesting property of the $\circ _{s}$ is that it
generally inherits its distribution from those of $(\xi _{i})$ and $X$. Indeed, the probability generating
function of $\alpha \circ _{s}X$ is
\begin{equation}
G_{\alpha \circ _{s}X}\left( z\right) :=E\left( z^{\alpha \circ
_{s}X}\right) =G_{X}\left( G_{\xi _{1}}\left( z\right) \right) ,  
\label{2.4} \nonumber
\end{equation}%
which shows that the probability law of $\alpha \circ _{s}X$\ is uniquely determined from
those of $\left( \xi _{i}\right) $ and $X$. The most commonly used RSO are as follows.\\
\textbf{Example 2.1 }
i) \textit{Binomial thinning}: The Bernoulli random
sum, aka the binomial thinning operator, and being denoted by $\circ $ (Steutel and van Harn, $1978$), assumes $%
\left( \xi _{i}\right) $ to be Bernoulli distributed with mean $\alpha \in
\left( 0,1\right) $. This implies that $\alpha \circ X|X \sim B\left( X,\alpha\right) $
is conditionally binomial distributed and $\alpha\circ X$  is stochastically smaller than $X$, that is
\begin{equation}
\alpha \circ X \leq X,  
\label{2.5}
\end{equation}%
hence the term thinning. 
For this operator,  
\begin{equation*}
\mathbb{V}\left( \alpha \circ X\right) =\alpha \left( 1-\alpha \right)
\mathbb{E}\left( X\right) +\alpha ^{2}\mathbb{V}\left( X\right) \text{ \ and \ }
\mathbb{V}\left( \alpha \circ X|X\right) =\alpha \left( 1-\alpha \right) X .
\end{equation*}

ii) \textit{Poisson random sum}: When $\left( \xi _{i}\right) $
are Poisson distributed with parameter $\alpha >0$, the random sum $\circ
_{s}$ is known as the Poisson generalized thinning and satisfies $%
\alpha \circ _{s}X|X \sim \mathcal{P}\left( \alpha X\right) $ where $\mathcal{%
P}\left( \alpha X\right) $ stands for the Poisson distribution with mean $%
\alpha X$. In this case, the range of $\alpha \circ _{s}X\in 
\mathbb{N}_0
:=\left\{ 0,1,\dots\right\} $ is the set of integers and the above thinning
property (\ref{2.5}) is no longer satisfied when $\alpha >1$, so the term
thinning makes no sense. In fact, as the range of $\alpha \circ _{s}X$ is
larger than that of $X$, the Poisson sum operator cannot be seen as a thinning operator in the sense of (\ref{2.5}). The variances in (\ref{2.3}) reduce to%
\begin{equation*}
\mathbb{V}\left(
\alpha \circ _{s}X\right) =\alpha \mathbb{E}\left( X\right) +\alpha ^{2}\mathbb{V}\left(
X\right) \text{ \ and \ }
\mathbb{V}\left( \alpha \circ _{s}X|X\right) =\alpha X.
\end{equation*}
iii) \textit{Negative binomial random sum}: If $\left( \xi
_{i}\right) $ are geometric distributed with parameter $\frac{1}{1+\alpha }$%
, the geometric random sum $\circ _{s}$ is known as the negative binomial generalized thinning and $\alpha \circ _{s}X|X \sim \mathcal{NB}\left( X,\frac{1}{1+\alpha }\right) $ is the negative binomial distribution with dispersion $%
X $ and probability $\frac{1}{1+\alpha }$. As a consequence $\alpha
\circ _{s}X\in \mathbb{N}_0$.

\subsection{Random multiplication operation}
Let $\Phi$ and $X$ be two independent non-negative integer-valued random variables, and $\phi$ a positive constant representing $\mathbb{E}\left( \Phi \right) =\phi >0$. The random multiplication operation (RMO) is defined as 
\begin{equation}
\phi \odot _{m}X
:=\Phi X=\sum_{i=1}^X\Phi = \sum_{i=1}^{\Phi} X.
\label{2.6}
\end{equation}%
Equality (\ref{2.6}) is satisfied for all $s\in \Omega $, assuming that all variables and sequences are defined on a probability space $\left( \Omega,\tciFourier ,P\right) $), and
\begin{equation*}
\left[ \phi \odot _{m}X\right] \left( s\right) =\left[ \sum_{i=1}^{X}\Phi %
\right] \left( s\right) =\left[ \sum_{i=1}^{\Phi }X\right] \left( s\right) =%
\left[ \Phi X\right] \left( s\right). 
\end{equation*}
Naturally, these everywhere or sure equalities (in the sense of being
satisfied $\forall s\in \Omega $) imply the almost sure equalities%
\begin{equation*}
\sum_{i=1}^{X}\Phi \overset{a.s.}{=}\sum_{i=1}^{\Phi }X%
\overset{a.s.}{=}\Phi X,  
\end{equation*}
which in turn implies that the equalities also hold in distribution. 
The converse, however, is not true. The sure equalities in (\ref{2.6}) are necessary to establish the almost sure convergence of the estimators in Section \ref{4SWLSE}. The $\phi \odot _{m}X$ operation can also be seen as a random sum (\ref{2.1}), with the restriction of having identical summands, designated by $\Phi $. 
The mean and variance of $\phi \odot _{m}X$ are given by%
\begin{eqnarray}
\mathbb{E}\left( \phi
\odot _{m}X\right) =\phi \mathbb{E}\left( X\right) \text{  \ and  \ } 
\mathbb{E}\left( \phi \odot _{m}X|X\right) &=&\phi X, 
\label{2.7} \\
\mathbb{V}\left( \phi \odot _{m}X\right) =\sigma _{\Phi }^{2}\mathbb{E}\left( X^{2}\right)
+\phi ^{2}\mathbb{V}\left( X\right) \text{ \ and \ } 
\mathbb{V}\left( \phi \odot _{m}X|X\right) &=&\sigma _{\Phi }^{2}X^{2}.  \label{2.8}
\end{eqnarray}%
While the operations $\alpha \circ _{s}X$ and $\phi \odot _{m}X$ share the
same conditional mean structure, the conditional variance is proportional
to the operand variable $X$ for the RSO (\ref{2.3}) and is proportional
to $X^{2}$ for the RMO (\ref{2.8}). The
variability implied by (\ref{2.8}) is therefore allowed to be larger than that
in (\ref{2.3}). Moreover, the range of the random multiplication $\phi \odot _{m}X$ is
\begin{equation}
\phi \odot _{m}X|(\Phi,X)\in \left\{ 0, \dots, \Phi X\right\} ,  \nonumber
\label{2.9}
\end{equation}%
so that $\odot _{m}$ is not a proper thinning operation in the sense of (\ref{2.5}). 

The random multiplication (\ref{2.6}) is much simpler and more
tractable than the random sum (\ref{2.1}) since it consists of a direct
multiplication by an integer-valued random coefficient. However, what is
gained in simplicity is lost in distributional reproducibility because the distribution of $\alpha \circ _{s}X$ is readily known
while that of $\phi \odot _{m}X$ is in general not usual. In
fact, although many well-known discrete distributions are stable under
independent summation (e.g.~binomial, Poisson and negative binomial), this is not the case for the multiplication operation. For example, the distribution of the product of two
independent Poisson variables is not Poisson distributed. This
makes the distribution of $\phi \odot _{m}X$ quite unusual except for
special cases (e.g.~the class of Bernoulli distributions is stable under
multiplication). For other cases, the
distributional properties of $\phi \odot _{m}X$ could easily be explored
using direct methods where specific distributions
for the random operand $\Phi $ in the $\phi\, \odot _{m}$ operation are of interest, as in the following examples.

\textbf{Example 2.2} i) \textit{Binomial random multiplication}: For a binomial
distributed $\Phi \sim \mathcal{B}\left( r,\tfrac{\phi }{r}\right)$ with  $r\geq
1$, the operation $\odot _{m}$ is called binomial multiplication. The
variances in (\ref{2.8}) translate into%
\begin{equation*}
\mathbb{V}\left( \phi \odot _{m}X\right) =\phi \left( 1-\tfrac{%
\phi }{r}\right) \mathbb{E}\left( X^{2}\right) +\phi ^{2}\mathbb{V}\left( X\right) \text{ \ and \ }
\mathbb{V}\left( \phi \odot _{m}X|X\right) =\phi \left( 1-\tfrac{\phi }{r}\right)
X^{2}.
\end{equation*}

ii) \textit{Poisson random multiplication}: When $\Phi \sim \mathcal{P}\left( \phi
\right) $ is Poisson distributed, the operation $\odot _{m}$ is called
Poisson multiplication. In this case, the variances in (\ref{2.8}) simplify to%
\begin{equation*}
\mathbb{V}\left( \phi
\odot _{m}X\right) =\phi \mathbb{E}\left( X^{2}\right) +\phi ^{2}\mathbb{V}\left( X\right)
\text{ \ and \ }
\mathbb{V}\left( \phi \odot _{m}X|X\right) =\phi X^{2}.
\end{equation*}%
Note that if $X \sim \mathcal{P}\left( \lambda \right) $, then $\mathbb{V}\left(
\phi \odot _{m}X\right) =\phi \left( \lambda +\lambda ^{2}\right) +\phi ^{2} \lambda
$ is allowed to be larger than $\mathbb{V}\left( \alpha \circ _{s}X\right)
=\alpha \lambda  + \alpha ^{2}\lambda$ implied by (\ref{2.3}) with $\alpha \circ
_{s}X|X \sim \mathcal{P}\left( \alpha X\right) $ and $X \sim \mathcal{P}\left(
\lambda \right) $.

iii) \textit{Negative binomial} \textit{I} \textit{random multiplication}: Assuming that $%
\Phi \sim \mathcal{NB}\left( r\phi ,\frac{r}{r+1}\right) $ is negative
binomial distributed (denoted by NB1, see Aknouche {\it et al.}, 2018%
), the operation $\odot _{m}$ is called NB1 multiplication. The variances in 
(\ref{2.8}) are thus%
\begin{equation*}
\mathbb{V}\left( \phi \odot _{m}X\right) =\phi \left( 1+\tfrac{1}{r}%
\right)  \mathbb{E}\left( X^{2}\right) +\phi ^{2}\mathbb{V}\left( X\right) \text{ \ and \ }
\mathbb{V}\left( \phi \odot _{m}X|X\right) = \phi \left( 1+\tfrac{1}{r}\right)  X^{2}.
\end{equation*}

iv) \textit{Negative binomial} \textit{II} \textit{random multiplication}: If $%
\Phi \sim \mathcal{NB}\left( r,\frac{r}{r+\phi }\right) $ ($r>0$) has a
negative binomial distribution (usually denoted by NB2, see Aknouche {\it et al.}, 2018; Aknouche and Francq 2021), then the operation $\odot _{m}$ is called NB2
multiplication. In this case, the variances in (\ref{2.8}) become%
\begin{equation*}
\mathbb{V}\left( \phi \odot _{m}X\right) =\phi \left( 1+\tfrac{1%
}{r}\phi \right) \mathbb{E}\left( X^{2}\right) +\phi ^{2}\mathbb{V}\left( X\right)  \text{ \ and \ }
\mathbb{V}\left( \phi \odot _{m}X|X\right) =\phi \left( 1+\tfrac{1}{r}\phi \right)
X^{2}.
\end{equation*}

Other constructions or extensions of the RMO can be considered including e.g.~a $\mathbb{Z}$-valued extension of the random multiplication in (\ref{2.6}) simply by assuming that $X$ and $\Phi$ are $\mathbb{Z}$-valued random variables and $\phi \in \mathbb{R}$.


\section{Random multiplication based auto-regressive models}
\label{section3}
This section presents a few examples of auto-regressive-like models constructed from the new RMO. These models will be addressed as random multiplication INAR of order $p$ and designated as RMINAR($p$).

\subsection{$\mathbb{N}_0$-valued random multiplication AR (RMINAR) model}

Let $\left\{ \Phi _{it},\text{ }t\in 
\mathbb{Z}
\right\} \ (i=1, \dots, p)$ and $\left\{ \varepsilon _{t},\text{ }t\in 
\mathbb{Z}
\right\} $ be mutually independent $%
\mathbb{N}_0
$-valued iid sequences with $\phi _{i}:=\mathbb{E}\left( \Phi _{it}\right) >0$, $%
\sigma_{\phi_i}^2:=\mathbb{V}\left( \Phi _{it}\right)>0$ ($i=1, \dots, p$), $\mu_{\varepsilon}:=\mathbb{E}\left(
\varepsilon _{t}\right) >0$, and $\sigma _{\varepsilon }^{2}:=\mathbb{V}\left( \varepsilon _{t}\right)>0$. A $%
\mathbb{N}_0
$-valued process $\left\{ Y_{t},\text{ }t\in 
\mathbb{Z}
\right\} $ is said to be an integer-valued random multiplication AR model, in short RMINAR($p$), if $Y_t$ admits the representation 
\begin{equation}
Y_{t}=\sum\limits_{i=1}^{p}\Phi _{it}Y_{t-i}+\varepsilon _{t}\text{, \ }t\in 
\mathbb{Z}
.  \label{3.1a}
\end{equation}%
Note that the RMINAR($p$) model $(\ref{3.1a})$ can be rewritten as follows
\begin{equation*}
Y_{t}=\sum\limits_{i=1}^{p}\phi _{i}\odot _{m}Y_{t-i}+\varepsilon _{t}\text{%
, \ }t\in 
\mathbb{Z}
\text{.}  
\end{equation*}%
The distribution of the input sequences $\left\{ \Phi _{it},\text{ }t\in 
\mathbb{Z}
\right\} \ (i=1, \dots, p)$ and $\left\{ \varepsilon _{t},\text{ }t\in 
\mathbb{Z}
\right\} $ can be specified as in Example 2.2 (binomial, Poisson, negative binomial, etc.). Thus, \eqref{3.1a} can be seen as a random coefficient AR model (RCAR) in the sense of Nicholls and Quinn $(1982)$ but with $\mathbb{N}_0$-valued random inputs $\left\{ \Phi _{it},\text{ }t\in 
\mathbb{Z}
\right\} \ (i=1, \dots, p)$ and $\left\{ \varepsilon _{t},\text{ }t\in 
\mathbb{Z}
\right\} $. The conditional mean and variance of  $(\ref{3.1a})$ are given by 
\begin{equation}
\mathbb{E}\left( Y_{t}|{\cal F}_{t-1}^{Y}\right) =\sum\limits_{i=1}^{p}\phi
_{i}Y_{t-i}+\mu_{\varepsilon} \text{ \ and \ }\mathbb{V}\left( Y_{t}|{\cal F}
_{t-1}^{Y}\right) =\sum\limits_{i=1}^{p}\sigma_{\phi_i}^2Y_{t-i}^{2}+\sigma
_{\varepsilon }^{2},  
\label{3.2}
\end{equation}%
where ${\cal F}_{t}^{Y}$ is the $\sigma $-algebra generated by $\left\{
Y_{t-u},u\geq 0\right\} $. From $(\ref{3.2})$ it is clear that the RMINAR$(p)$
model allows both conditional overdispersion and underdispersion. E.g.~when the inputs $\Phi _{it}$ and $\varepsilon
_{t}$ are Poisson distributed, then $(\ref{3.2})$ becomes%
\begin{equation*}
\mathbb{E}\left( Y_{t}|{\cal F}_{t-1}^{Y}\right) =\sum\limits_{i=1}^{p}\phi
_{i}Y_{t-i}+\mu_{\varepsilon}\text{ \ and \ }\mathbb{V}\left( Y_{t}|{\cal F}_{t-1}^{Y}\right) =\sum\limits_{i=1}^{p}\phi _{i}Y_{t-i}^{2}+\mu_{\varepsilon},
\end{equation*}%
and $\mathbb{V}\left( Y_{t}|{\cal F}_{t-1}^{Y}\right) >\mathbb{E}\left(
Y_{t}|{\cal F}_{t-1}^{Y}\right) $ as $Y_t \leq Y_t^2$ for all $t$. Also, $(\ref{3.2})$  can lead to $\mathbb{V}\left( Y_{t}|{\cal F}_{t-1}^{Y}\right) <\mathbb{E}\left( Y_{t}|{\cal F}_{t-1}^{Y}\right) $ e.g.~when the
inputs are binomial distributed with appropriate parameters.

Note that for $p=1$ the model $(\ref{3.1a})$ reduces to 
\begin{equation*}
Y_{t}=\Phi_{t}Y_{t-1}+\varepsilon _{t}\text{, \ }t\in 
\mathbb{Z},
\end{equation*}%
which is a homogeneous Markov chain with transition probabilities given by
\begin{eqnarray*}
\mathbb{P}\left( Y_{t}=j|Y_{t-1}=i\right) &=&\mathbb{P}\left( \Phi _{t}Y_{t-1}+\varepsilon _{t}=j|Y_{t-1}=i\right)  \\
&=&\mathbb{P}\left( \Phi _{t}i+\varepsilon _{t}=j\right)  \\
&=&\left\{ 
\begin{array}{ll}
\sum\limits_{k\leq j,\frac{j-k}{i}\in 
\mathbb{N}
_{0}}\mathbb{P}\left( \varepsilon _{t}=k\right) \mathbb{P}\left( \Phi _{t}=\frac{j-k}{i}%
\right)&, i>0 \\ 
\mathbb{P}\left( \varepsilon _{t}=j\right) &, i=0%
\end{array}.%
\right. 
\end{eqnarray*}%

The novel RMINAR($p$) model $(\ref{3.1a})$  can be compared to the random coefficient INAR
model of Zheng {\it et al.} ($2006$, $2007$), denoted as RCINAR$\left( p\right) $ and defined as%
\begin{equation}
X_{t}=\sum\limits_{i=1}^{p}\alpha _{it}\circ X_{t-i}+\zeta _{t},\text{ \ }%
t\in 
\mathbb{Z}
\text{,}  \label{3.3}
\end{equation}%
where $``\circ"$ is the binomial thinning operator whereas $\left\{ \alpha _{it}, \text{ }t\in \mathbb{Z}
\right\}$ ($i=1, \dots, p$) and $\left\{ \zeta
_{t},\text{ }t\in 
\mathbb{Z}
\right\} $ are mutually independent iid
sequences valued in $(0,1)$ and $%
\mathbb{N}_0
$, respectively.  Models in $(\ref{3.1a})$ and $(\ref{3.3})$ share the same conditional mean structure and similar quadratic
conditional variances in terms of past observations. Note that $(\ref{3.3})$
reduces to the INAR$(p)$ model of Du and Li $(1991)$ (see also McKenzie, $%
1985$; Al-Osh and Alzaid, $1987$ for the case $p=1$) and thus the conditional
mean and variance of $(\ref{3.3})$ are 
\begin{equation*}
\mathbb{E}\left( X_{t}|{\cal F}_{t-1}^{X}\right) =\sum\limits_{i=1}^{p}\alpha
_{i}\circ X_{t-i}+\zeta \text{ \ and \ }\mathbb{V}\left( X_{t}|{\cal F}
_{t-1}^{X}\right) =\sum\limits_{i=1}^{p}\sigma _{\alpha
_{i}}^{2}X_{t-i}^{2}+\left( \alpha
_{i}\left( \alpha_{i}-1\right) -\sigma
_{\alpha _{i}}^{2}\right) X_{t-i}+\sigma _{\zeta }^{2},
\end{equation*}
where $\alpha_{i}:=\mathbb{E}(\alpha_{it})$, $\sigma^2_{\alpha_{i}}:=\mathbb{V}(\alpha_{it})$, $\zeta :=\mathbb{E}\left( \zeta _{t}\right)$ and $\mathbb{V}\left( \zeta _{t}\right) =\sigma _{\zeta }^{2}>0$ (Zheng \textit{et al.}, $2006$). 

The probabilistic structure of $(\ref{3.1a})$ is already well known when the inputs
are real-valued (e.g..~Nicholls and Quinn, $1982$; Feigin and Tweedie, $1985$%
; Tsay, $1987$; Diaconis and Freedman, $1999$). In the integer-valued case,
however, there are some surprising properties. Let $\mathbf{Y%
}_{t}=\left( Y_{t}, \dots, Y_{t-p+1}\right) ^{\prime }$ and $\mathbf{\Xi }%
_{t}=\left( \varepsilon _{t},0, \dots, 0\right) ^{\prime }$ be $p$-dimensional column vectors, and define the $%
p\times p$ companion matrix%
\begin{equation}
A_{t}=\left( 
\begin{array}{cc}
\left(\Phi _{1t}, \dots, \Phi _{(p-1)t} \right) & \Phi _{pt} \\ 
\mathbf{I}_{p-1} & \mathbf{0}_{\left( p-1\right) \times 1}%
\end{array}%
\right),
\label{eq: At}
\end{equation}%
where $\mathbf{I}_{p}$ and $\mathbf{0}_{p \times 1}$ are the $p$-dimensional identity matrix and zero vector, respectively. Then model $(\ref{3.1a})$ can be
written in the following vector form%
\begin{equation}
\mathbf{Y}_{t}=A_{t}\mathbf{Y}_{t-1}+\mathbf{\Xi }_{t},\text{ \ }t\in 
\mathbb{Z}
\label{3.4}.
\end{equation}

For almost all common integer-valued time series models (observation-driven,
parameter-driven, random sum based), the conditions for strict
stationary and mean stationary on the auto-regressive parameter coincide. The formulation in $(\ref{3.1a})$, however, allows
the output process $\left\{ Y_{t},t\in 
\mathbb{Z}
\right\} $\ to be everywhere (i.e. for all parameter values) strictly stationary even with infinite mean. In
fact, the following result shows that any solution of $(\ref{3.1a})$ is everywhere (i.e.~universally) 
strictly stationary provided that $\mathbb{P}\left( \Phi _{it}=0\right) >0$ for all $i=1, \dots, p$. For simplicity in notation and readability, the latter condition is denoted as \textbf{A0}.
\begin{theo}
Under \mbox{\textbf{A0}},
\textit{\ the series} $%
\mathbf{Y}_{t}:=\sum\limits_{j=0}^{\infty }\prod\limits_{i=0}^{j-1}A_{t-i}\mathbf{\Xi }%
_{t-j}$ converges absolutely a.s. for all $t\in 
\mathbb{Z}
$\textit{\ and the process }$\left\{ Y_{t},t\in 
\mathbb{Z}
\right\} $\textit{\ given by }$Y_{t}=\left( 1,0, \dots, 0\right) ^{\prime }%
\mathbf{Y}_{t}$ is the unique (causal) strictly stationary and ergodic solution
to the RMINAR equation $(\ref{3.1a})$\textit{.}
\label{teo1}
\end{theo}
\begin{proof}
See Appendix A.
\end{proof}

Most of the usual count distributions (binomial, Poisson, negative binomial,
etc.) satisfy \textbf{A0} and thus guarantee that
model $(\ref{3.1a})$ has a (causal) strictly stationary and ergodic solution
whatever the value of the inputs $\Phi _{it}$ and $\varepsilon _{t}$. Then, the RMINAR $(\ref{3.1a})$ is universally (or everywhere) stable
with respect to stationarity and ergodicity.

\begin{rem} i) When \textbf{A0} is not satisfied (e.g.~for the truncated geometric distribution, modeling the number of Bernoulli trials to get the first success, and taking $\mathbb{P}\left( \Phi
_{1t}\in \left\{ 1,2,\dots\right\} \right) =1$), assuming that $\mathbb{E}\left( \log
^{+}\left( \varepsilon _{t}\right) \right) <\infty $ and $\mathbb{E}\left( \log
^{+}\left( \Phi _{it}\right) \right) <\infty $ ($i=1, \dots, p$), a sufficient
condition\ for the existence of a strictly stationary and ergodic solution
to $(\ref{3.1a})$ is that the largest Lyapunov exponent $\gamma$ (Furstenberg and Kesten, $%
1960$; Kesten, $1973$)%
\begin{equation*}
\gamma =\lim_{t\rightarrow \infty }\tfrac{1}{t}\log
\left\Vert A_{1} A_{2} \cdots A_{t}\right\Vert<0 \text{ a.s.,}
\end{equation*}%
where $\log^{+} (x):=\max \left( \log (x),0\right) $, $\left\Vert
.\right\Vert $ is an operator norm on the space of square real matrices of
dimension $p$ and $A_i$ ($i=1, \dots, t$) are $p\times p$ random matrices.

ii) When \textbf{A0} is satisfied and $p=1$, the Lyapunov exponent $\mathbb{E} \left(\log \left(\Phi _{1t}\right)\right)=-\infty $ provided that $\mathbb{E}\left( \log ^{+} \left(\Phi _{1t}\right) \right) <\infty $ with the
convention $\max (-\infty ,0)=0$.
\label{remk1}
\end{rem}
Using stochastic recurrence equations theory (Kesten, $1973$; Vervaat, $1979$%
; Goldie, $1991$; Grey, $1994$), the tail behavior of model $(\ref{3.1a})$ can be
easily revealed. Note that due to \textbf{A0}, the distribution of $\log
\left( \Phi _{it}\right) $ given $\Phi _{it}\neq 0$ is non-arithmetic for
all $i=1, \dots, p$. Moreover, the equation $(\ref{3.3})$ admits a strictly stationary
solution whatever the value of $\gamma$.
Therefore, the following proposition states a result that is an obvious corollary of Theorems 3-5 of Kesten $(1973)$ (see also Theorem 1 of Grey $(1994)$ when $p=1$). For this reason the proof of proposition \ref{prop31} is omitted. 

\begin{proposition}
Consider the RMINAR model \eqref{3.1a} under \textbf{A0} and $\mathbb{E}\left( \log ^{+}\left( \Phi _{it}\right) \right)
<\infty $ ($i=1, \dots, p$).

i) For $p=1$, assume that
\begin{equation}
\mathbb{E}\left( \Phi _{1t}^{\tau _{0}}\log ^{+}\left(\Phi _{1t}\right)\right) <\infty \text{ \ 
\textit{and} \ } \mathbb{E}\left( \Phi _{1t}^{\tau _{0}}\right) \geq 1\text{ \textit{, for
some} }\tau _{0}>0.  
\label{3.5}
\end{equation}%
Then there exists $\tau _{1}\in \left[ 0,\tau _{0}\right] $ such that 
 $\mathbb{E}\left( \Phi _{1t}^{\tau _{1}}\right) =1$ has a
unique solution. In addition, if $\mathbb{E}\left( \varepsilon _{t}^{\tau
_{1}}\right) <\infty$ then
\begin{equation}
\mathbb{P}\left( Y_{t}\geq y\right)  \rightarrow y^{-\tau _{1}},\text{ \textit{as} }%
y\rightarrow \infty \text{.} 
\label{3.6}
\end{equation}

ii) For $p>1$, assume there exists $\tau _{0}>0$ such that
\begin{equation}
\mathbb{E}\left( \sum_{j=1}^{p}\Phi _{jt}^{\tau _{0}}\right) \geq p^{\frac{\tau _{0}}{%
2}}\text{, }\mathbb{E}\left( \left\Vert A_{t}\right\Vert ^{\tau _{0}}\log
^{+}\left\Vert A_{t}\right\Vert \right) <\infty \text{ \ 
\textit{and} \ } \mathbb{E}\left(
\varepsilon _{t}^{\tau _{1}}\right) <\infty \text{, }\tau _{1}\in \left[
0,\tau _{0}\right] \text{.}  
\label{3.7}
\end{equation}

Then for all non-negative $p$-vector $\mathbf{x}$ such that $\left\Vert \mathbf{x}\right\Vert =1$\textit{,}
\begin{equation}
\mathbb{P}\left( \mathbf{x}^{\prime }\mathbf{Y}_{t}\geq y\right) \rightarrow y^{-\tau _{1}},%
\text{ \textit{as} }y\rightarrow \infty \text{.}
\label{3.8}
\end{equation}
\label{prop31}
\end{proposition}
\vspace{-1cm}
Since the inputs of the RMINAR($p$) model \eqref{3.1a} have, in general, count distributions (binomial, Poisson, negative binomial, double Poisson, etc.)~for which all
moments exists, the conditions \eqref{3.5} and \eqref{3.7} are generally satisfied.
Thus, the stationary solutions of \eqref{3.1a} would have regularly-varying
tails in the sense of \eqref{3.6} and \eqref{3.8}.

While the conditions of strict stationarity, ergodicity and regular variation for the RMINAR model \eqref{3.1a} are somewhat different from those of the real-valued RCAR case, the
moment conditions for the two cases are the same. These conditions for model \eqref{3.1a} are stated in proposition \ref{prop32}, the proof of which is very similar to that of Theorem 2.9 in Francq and Zakoian (2019) and hence is omitted. The reader is further referred to Feigin and Tweedie $(1985)$ for more details (see also Tsay, $1987$ and Ling, $%
1999$ for similar models).

\begin{proposition}
Suppose that $\mathbb{E}\left(\Phi _{it}^{m}\right)
<\infty $ and $\mathbb{E}\left( \varepsilon _{t}^{m}\right) <\infty$ for some positive integer $m$. Then the RMINAR model \eqref{3.1a} admits a strictly stationary and ergodic solution with $\mathbb{E}\left(
Y_{t}^{m}\right) <\infty $ if
\begin{equation}
\rho \left( \mathbb{E}\left( A_{t}^{\otimes m}\right) \right) <1, 
\label{3.9}
\end{equation}
where $A^{\otimes m}=A\otimes\cdots\otimes A$ is the Kronecker product with $m$
factors. 
If $\rho \left( \mathbb{E}\left( A_{t}^{\otimes m}\right) \right) \geq 1$ then $\mathbb{E}\left( Y_{t}^{m}\right) =\infty$.
\label{prop32}
\end{proposition}

From the result in \eqref{3.9} with $m=1$, the RMINAR($p$) model \eqref{3.1a} admits a stationary and ergodic
solution with $\mathbb{E}\left( Y_{t}\right) <\infty $ if%
\begin{equation}
\rho \left( A_t \right) <1,  
\label{3.10} \nonumber
\end{equation}%
which is equivalent to 
\begin{equation}
\sum\limits_{i=1}^{p}\phi _{i}<1.  \label{3.11}
\end{equation}%
Under $(\ref{3.11})$, the unconditional mean has the following expression 
\begin{equation}
\mathbb{E}\left( Y_{t}\right) =\mu_{\varepsilon} \left( 1-\sum\limits_{i=1}^{p}\phi _{i}\right)
^{-1},  
\label{3.12}
\end{equation}%
which is the same as the mean of RCINAR$\left( p\right) $ and AR$(p)$ models. 
If $\sum\limits_{i=1}^{p}\phi _{i}= 1$ or $\sum\limits_{i=1}^{p}\phi _{i}>1$ then the RMINAR($p$) model $(\ref{3.1a})$ still admits a
strictly stationary and ergodic solution $\left\{ Y_{t},t\in 
\mathbb{Z}
\right\} $ (under \textbf{A0}), but with an infinite mean $\mathbb{E}\left(
Y_{t}\right) =\infty $.

A sufficient condition for second-order stationary is given by 
$(\ref{3.9})$ with $m=2$. As an example, for the RMINAR($1$) the condition is 
\begin{equation*}
0\leq \mathbb{E}\left( \Phi _{1t}^{2}\right) =\sigma_{\phi_1}^2+\phi _{1}^2<1,
\end{equation*}%
and the unconditional variance of the process is
\begin{equation}
\mathbb{V}\left( Y_{t}\right) =\frac{\sigma_{\varepsilon}^{2}+\sigma_{\phi_1}^2\left( \mathbb{E}\left(Y_{t-1}\right)\right)
^{2}}{1-\left(\sigma_{\phi_1}^2+\phi _{1}^2\right) }.  
\label{3.13}
\end{equation}%
For any $p$, $\mathbb{V}\left( Y_{t}\right) $ can be given from $(\ref%
{3.4})$ using the vector stacking operator $vec$. Let $A:=\mathbb{E}\left(
A_{t}\right) $, $\mathbf{\Xi }:=\mathbb{E}\left( \mathbf{\Xi }_{t}\right) $, $%
\mathbf{\mu }:=\mathbb{E}\left( \mathbf{Y}_{t}\right) $, $\Gamma _{Y}:=%
\mathbb{V}(\mathbf{Y}_{t})=\mathbb{E}\left( \mathbf{Y}_{t}-\mathbf{\mu }%
\right) \left( \mathbf{Y}_{t}-\mathbf{\mu }\right) ^{\prime }$ and $\Gamma
_{\mathbf{\Xi }}:=\mathbb{V}\left( \mathbf{\Xi }_{t}\right) $. Then%
\begin{equation}
vec(\Gamma _{Y})=\left( I-E\left( A_{t}\otimes A_{t}\right) \right)
^{-1}\left( \left( E\left( A_{t}\otimes A_{t}\right) -\left( A\otimes
A\right) \right) vec\left( \mathbf{\mu \mu }^{\prime }\right) +vec\left(
\Gamma _{\mathbf{\Xi }}\right) \right) .  
\label{3.13b}\nonumber
\end{equation}
In particular, for $p=2$ it follows that
\begin{equation*}
\mathbb{V}\left( Y_{t}\right) =\frac{\left( 1-\phi _{2}\right) \left( \left(
\sigma _{\phi _{1}}^{2}+\sigma _{\phi _{2}}^{2}\right) \left( \mathbb{E}%
\left( Y_{t}\right) \right) ^{2}+\sigma _{\varepsilon }^{2}\right) }{%
1-\left( \phi _{1}^{2}+\phi _{1}^{2}\phi _{2}-\phi _{2}^{3}+\phi
_{2}^{2}+\phi _{2}+\left( 1-\phi _{2}\right) \left( \sigma _{\phi
_{1}}^{2}+\sigma _{\phi _{2}}^{2}\right) \right) }\text{.}
\end{equation*}%
On a final note, the RMINAR model is able to generate both
(unconditional) over-dispersion (e.g.~with Poisson distributed inputs) and
under dispersion (e.g.~with binomial inputs).

\subsection{A $\mathbb{Z}$-valued RMINAR extension}

The RMINAR$(p)$ model $(\ref{3.1a})$ can be extended to address $%
\mathbb{Z}
$-valued time series simply by considering $%
\mathbb{Z}
$-valued random inputs. Thus, let $\left\{ \Phi _{it},\text{ }t\in 
\mathbb{Z}
\right\} \ (i=1, \dots, p)$ and $\left\{ \varepsilon _{t},\text{ }t\in 
\mathbb{Z}
\right\} $ be mutually independent $%
\mathbb{Z}
$-valued iid sequences with $\phi _{i}\in 
\mathbb{R}
$ 
($i=1, \dots, p$) and $%
\mu_{\varepsilon} \in \mathbb{R}$ instead of being constrained to be positive values. 
A $%
\mathbb{Z}
$-valued RMINAR$(p)$ writes as in $(\ref{3.1a})$ as follows%
\begin{equation}
Y_{t}=\sum\limits_{i=1}^{p}\Phi _{it}Y_{t-i}+\varepsilon _{t},\ t\in 
\mathbb{Z}
,  \label{3.14}
\end{equation}%
where the random multiplication $\Phi _{it}Y_{t-i}\equiv \phi _{i}\odot
_{m}Y_{t-i}$ now acts on $%
\mathbb{Z}
$ with $\phi _{i}\in 
\mathbb{R}
$ and $Y_{t-i}\in 
\mathbb{Z}
$. Many distributions supported on $%
\mathbb{Z}
$ can be chosen for the inputs $\left\{ \Phi _{it},\text{ }t\in 
\mathbb{Z}
\right\} \ (i=1,\dots, p)$ and $\left\{ \varepsilon _{t},\text{ }t\in 
\mathbb{Z}
\right\} $. For example, a useful candidate is the Skellam distribution
(Irwin, $1937$; Skellam, $1946$) for the
difference of two independent Poisson variables (or simply Poisson
difference, PD). Another useful choice with heavier tails is the negative
binomial difference (NBD) distribution (cf. Barndorff-Nielsen {\it et al.}, $2012$; Barra and Koopman, $2018$).

The conditional mean and variance of model $(\ref{3.14})$ are given exactly as $(\ref{3.2})$. Many efforts have been made to define $\mathbb{Z}$-valued (or signed integer-valued) models following different approaches, namely thinning-based models (e.g.~Kim and Park, $2008$; Zhang
\textit{et al.}, $2009$; Alzaid and Omair, $2014$), parameter driven models
(Koopman \textit{et al.}, $2017$; Barra and Koopman, $2018$), and observation
driven models (Alomani \textit{et al.}, $2018$; Cui \textit{et al.}, $2021$).
The resulting specifications, however, appear to be less straightforward than those of the $\mathbb{Z}$-valued RMINAR($p$) proposed in this work. In particular, the INGARCH-type models
(Alomani \textit{et al.}, $2018$; Cui \textit{et al.}, $2021$) only concern
the conditional variance and not the conditional mean. On the contrary, the RMINAR($p$) approach $(\ref{3.14})$ simultaneously addresses the conditional mean and variance and is analytically much simpler.
Furthermore, given the connections between RMINAR($p$) and RCAR($p$), the existing real-valued RCAR($p$) tools (Nicholls and Quinn, $1982$)
can be adapted e.g.~for the estimation of RMINAR($p$) models.

The probability structure of $(\ref{3.14})$ is similar to that of $(\ref{3.1a})$ as $(\ref{3.14})$ can be written in the vector form $(\ref{3.4})$ with obvious notations. Furthermore, the
universal ergodicity property for the $\mathbb{N}_0$-valued RMINAR model $(\ref{3.1a})$ still holds for the $\mathbb{Z}$-valued RMINAR model $(\ref{3.14})$. Under the assumption $\mathbb{P}\left( \Phi _{it}=0\right) >0$ $(i=1, \dots, p)$, which is condition $\textbf{A0}$ where $\Phi_{it}$ is $\mathbb{Z}$-valued instead of $\mathbb{N}_0$-valued (hereafter referred as condition \textbf{A0*}), Theorem \ref{teo1} still holds true for the $\mathbb{Z}$-valued RMINAR model $(\ref{3.14})$.

\begin{theo}
Under \textbf{A0*} the conclusions
of Theorem $\ref{teo1}$ remain true for the $\mathbb{Z}$-valued RMINAR model $(\ref{3.14})$.
\label{teo2}
\end{theo} 

\begin{rem}
i) If \textbf{A0*} fails (e.g.~$\mathbb{P}\left( \Phi_{1t}\in \mathbb{Z}\backslash \left\{ 0\right\} \right) =1$), then assuming that $\mathbb{E}\left( \log
^{+}\left( \left\vert \varepsilon _{t}\right\vert \right) \right) <\infty $
and $\mathbb{E}\left( \log ^{+}\left( \left\vert \Phi _{it}\right\vert \right)
\right) <\infty $ $(i=1, \dots, p)$, the $\mathbb{Z}$-valued RMINAR model $(\ref{3.14})$ still admits a unique strictly
stationary and ergodic solution if $\gamma <0$.    
\label{rem32}
\end{rem} 
The tail behavior is much more difficult to analyze for the $\mathbb{Z}$-valued RMINAR model \eqref{3.14} than for the  $\mathbb{N}_{0}$-valued RMINAR model \eqref{3.1a}. However, it can be easily unmasked for $p=1$ by using
the following Kesten-Goldie theorem (Kesten, $1973$; Goldie, $1991$). This result is presented in proposition \ref{prop33} which constitutes a corollary of Kesten's theorem 5.

\begin{proposition}
Consider the $\mathbb{Z}$-valued RMINAR(1) model \eqref{3.14} subject
to condition \textbf{A0*} and $\mathbb{E}\left( \log ^{+}\left( \left\vert
\Phi _{1t}\right\vert \right) \right) <\infty$. Suppose further that
\begin{equation}
\mathbb{E}\left( \left\vert \Phi _{1t}\right\vert ^{\tau _{1}}\right) =1\text{, }%
\mathbb{E}\left( \left\vert \Phi _{1t}\right\vert ^{\tau _{1}}\log ^{+}\left\vert
\Phi _{1t}\right\vert \right) <\infty \text{ \ 
\textit{and} \ } \mathbb{E}\left(
\left\vert \varepsilon _{t}\right\vert ^{\tau _{1}}\right) <\infty \text{ 
\textit{, for some} }\tau _{1}>0.  
\label{3.15}\nonumber
\end{equation}%
Then
\begin{equation}
\mathbb{P}\left( Y_{t}\geq y\right)\rightarrow  c^{+}y^{-\tau _{1}} \text{ \ 
\textit{and} \ } \mathbb{P}(Y_{t}<-y)\rightarrow c^{-}y^{-\alpha }, \text{ \textit{as} }y\rightarrow \infty 
\text{,}  
\label{3.16}\nonumber
\end{equation}%
if and only if $\mathbb{P}(\varepsilon _{1}+\mu_{\varepsilon}=(1-\phi _{1})c)<1$ for any $c \in \mathbb{R}$ such that $c^{+}+c^{-}>0$.
\label{prop33}
\end{proposition}
Concerning the existence of moments for the $\mathbb{Z}$-valued RMINAR($p$) model \eqref{3.14} assume that $%
\mathbb{E}\left( \left\vert \varepsilon _{t}\right\vert ^{2m}\right) <\infty$, for $%
m\geq 1$. Using again the results of Feigin and Tweedie ($1985$, Theorems
4-5), a sufficient condition for the existence of a stationary and ergodic
solution with $\mathbb{E}(\left\vert Y_{t}\right\vert ^{2m})<\infty$ is that 
\begin{equation*}
\rho \left( \mathbb{E}\left( A_{t}^{\otimes 2m}\right) \right) <1.
\end{equation*}

As for the $\mathbb{N}_{0}$-valued RMINAR($p$) model \eqref{3.1a}, the marginal mean and variance of $\mathbb{Z}$-valued RMINAR($p$) model \eqref{3.14} are given, respectively, by \eqref{3.12} and \eqref{3.13} for $p=1$.

\subsection{RMINAR model with multiplicative errors}

The RMINAR($p$) models can be combined together with a scheme of multiplicative errors instead of additive ones. An example is the multiplicative thinning-based INGARCH model (MthINGARCH)
of Aknouche and Scotto $(2024)$, which has been shown to be flexible in representing many important features commonly observed
integer-valued time series, such as over-dispersion and heavy-tailedness.
The MthINGARCH($p$,$q$) process is defined as %
\begin{equation}
Y_{t}=\lambda _{t} \ \varepsilon _{t}\text{ \ and \ }\lambda _{t}=1+\omega
\circ m+\sum\limits_{i=1}^{p}\alpha _{i}\circ
Y_{t-i}+\sum\limits_{j=1}^{q}\beta _{j}\circ \lambda _{t-j},  
\label{3.17}
\end{equation}%
where $``\circ"$ stands for the BTO. The multiplicative error $\left( \varepsilon _{t}\right) $ is an iid integer-valued sequence with $\mathbb{E}\left( \varepsilon _{t}\right)=1$ and variance $\sigma _{\varepsilon }^{2}:=\mathbb{V}\left( \varepsilon _{t}\right)$ while the random variable $\lambda _{t}$ is such that $\lambda _{t}\in \mathbb{N}$, as a consequence of $(\ref{3.17})$. Given the restrictions on the BTO, the coefficients of the MthINGARCH satisfy $0\leq \alpha _{i}<1$, $0\leq \beta _{j}<1$ ($i=1, \dots, p$, $%
j=1, \dots, q$) and $0\leq \omega \leq 1$, being $\omega$ the expected value of the Bernoulli random variables $\left( \xi _{i}\right) \ (i=1, \dots, m)$ underlying the BTO. Finally, being $m$ defined as a fixed positive integer then $\omega \circ m \sim  B\left( m,\omega \right)$.

The MthINGARCH model $(\ref{3.17})$ can be modified by replacing the BTO $\circ$ by the RMO $\odot _{m}$ in $(\ref{2.6})$. This case considers i) the MthINGARCH framework $(\ref{3.17})$ with $\beta _{j}=0$ ($j=1, \dots, q$) for simplicity and ii) the replacement of the random variable $\omega\circ m$ by an
integer-valued random coefficient $\omega _{t}$ with $\mathbb{E}\left( \omega _{t}\right)=\omega$ and $\sigma _{\omega }^{2}:=\mathbb{V}\left( \omega _{t}\right)>0$. These specifications result in the following RMINAR($p$) model with multiplicative errors
\begin{equation}
Y_{t}= \lambda _{t} \ \varepsilon _{t} = \left( 1+\omega _{t}+\sum\limits_{i=1}^{p}\phi _{i} \odot _{m} Y_{t-i}\right) \varepsilon _{t} = \left( 1+\omega _{t}+\sum\limits_{i=1}^{p}\Phi _{it} Y_{t-i}\right)
\varepsilon _{t}, 
\label{3.18}
\end{equation}%
where the integer-valued iid sequences $\left\{ \omega _{t},\text{ }t\in \mathbb{Z} \right\} $, $\left\{ \Phi _{it},\text{ }t\in \mathbb{Z}\right\} \ (i=1, \dots, p)$ and $\left\{\varepsilon _{t},\text{ }t\in \mathbb{Z} \right\} $ are mutually independent. As before, $\phi _{i}:=\mathbb{E}\left( \Phi_{it}\right) \ (i=1, \dots, p)$ and $\sigma _{\phi
_{i}}^{2}:=\mathbb{V}\left( \Phi _{it}\right)>0$.
It follows from the RMINAR($p$) formulation $(\ref{3.18})$ that
\begin{eqnarray}
\mathbb{E}\left( Y_{t}|\tciFourier _{t-1}^{Y}\right)  &=&1+\omega
+\sum\limits_{i=1}^{p}\phi _{i}Y_{t-i},  
\label{3.19a} \\
\mathbb{V}\left( Y_{t}|\tciFourier _{t-1}^{Y}\right)  &=&\left( \sigma
_{\varepsilon }^{2}+1\right) \mathbb{V}\left( \lambda _{t}|\tciFourier _{t-1}^{Y}\right) +\sigma
_{\varepsilon }^{2}\left(\mathbb{E}\left( Y_{t}|\tciFourier _{t-1}^{Y}\right)\right) ^{2}, 
\label{3.19b}
\end{eqnarray}
where%
\begin{equation*}
\mathbb{V}\left( \lambda _{t}|\tciFourier _{t-1}^{Y}\right)=\sigma _{\omega
}^{2}+\sum\limits_{i=1}^{p}\sigma _{\phi _{i}}^{2}Y_{t-i}^{2},
\end{equation*}%
which highlights that the RMINAR($p$) with multiplicative errors $(\ref{3.18})$ presents the same conditional-variance-to-mean relationship as that of the BTO-based MthINGARCH($p$,$0$) model (Aknouche and Scotto, $2024$).

Under a similar assumption as \textbf{A0}, the ergodic properties of the RMINAR($p$) model with multiplicative errors \eqref{3.18} are obtained everywhere similarly to the above cases. The model \eqref{3.18} can be written in the following vector form%
\begin{equation}
\mathbf{Y}_{t}=A_{t}\mathbf{Y}_{t-1}+\mathbf{\Psi }_{t},\text{ \ }t\in \mathbb{Z}, 
\label{3.20}\nonumber
\end{equation}%
with the definition of the $p$-dimensional column vectors $\mathbf{Y}_{t}=\left( Y_{t}, \dots, Y_{t-p+1}\right)
^{\prime }$ and $\mathbf{\Psi }_{t}=\left( \varepsilon _{t}\left( 1+\omega
_{t}\right) , 0, \dots, 0\right) ^{\prime }$, and the $p\times p$ companion matrix%
\begin{equation*}
A_{t}=\left( 
\begin{array}{cc}
\left(\varepsilon _{t}\phi _{1t}, \dots, \varepsilon _{t}\phi _{(p-1)t}\right) & \varepsilon
_{t}\phi _{pt} \\ 
\mathbf{I}_{p-1} & \mathbf{0}_{\left( p-1\right) \times 1}%
\end{array}%
\right) .
\end{equation*}%

Similarly to Theorem \ref{teo1}, the following result shows that under $\mathbb{P}\left( \varepsilon _{t}=0\right) >0$ and $\mathbb{P}\left( \Phi
_{it}=0\right) >0 \ (i=1, \dots, p)$ (which will be referred as condition \textbf{A0**}), any solution of the RMINAR($p$) model with multiplicative errors \eqref{3.18} is everywhere stationary and ergodic.

\begin{theo}
Under \textbf{A0}**, the series 
\begin{equation}
\mathbf{Y}_{t}:=\sum\limits_{j=0}^{\infty }\prod\limits_{i=0}^{j-1}A_{t-i}\mathbf{\Psi }%
_{t-j},
\nonumber
\end{equation}
converges absolutely a.s. for all $t\in \mathbb{Z}$. In addition, the process $\left\{ Y_{t},t\in \mathbb{Z} \right\} $\ given by $Y_{t}=\left(1, 0, \dots, 0\right) ^{\prime }%
\mathbf{Y}_{t}$\ $(t\in \mathbb{Z})$ is the unique strictly stationary and ergodic solution to \eqref{3.18}. 
\label{teo3}
\end{theo} 


The tail behavior of the RMINAR model with multiplicative errors \eqref{3.18} can be obtained similarly to that of the $\mathbb{N}_0$-valued RMINAR model \eqref{3.1a} using Kesten's theorem. Proposition \ref{prop34} states such result and its proof its omitted, given that it is an obvious corollary of Kesten’s theorems 3-5 (Kesten, 1973). 

\begin{proposition} Under \textit{\textbf{A0}}**, all conclusions of
Proposition \ref{prop31} hold for the RMINAR model with multiplicative errors \eqref{3.18} while replacing $\Phi _{it}$ by $\varepsilon _{t}\Phi _{it}$.
\label{prop34}
\end{proposition}

\section{Four-stage WLS estimation}\label{4SWLSE}

This section presents an estimation strategy for the parameters mean $\theta _{0}$ and variance $\Lambda _{0}$ of the RMINAR models. 
For the $\mathbb{N}_0$-valued (\ref{3.1a}) and $\mathbb{Z}$-valued (\ref{3.14}) models, the parameters are set as $\theta _{0}=\left(
\mu _{0\varepsilon },\phi _{01},\dots,\phi _{0p}\right) ^{\prime }$ and $%
\Lambda _{0}=\left( \sigma _{0\varepsilon }^{2},\sigma _{0\phi
_{1}}^{2},\dots,\sigma _{0\phi _{p}}^{2}\right) ^{\prime }$ while for the model with multiplicative errors (\ref{3.18}), $\theta _{0}=\left( \omega _{0},\phi _{01},\dots,\phi
_{0p}\right) ^{\prime }$ and $\Lambda _{0}=\left( \sigma _{0\varepsilon
}^{2},\sigma _{0\omega }^{2},\sigma _{0\phi _{1}}^{2},\dots,\sigma _{0\phi
_{p}}^{2}\right) ^{\prime }$. The subscript 0 in the notation of $\theta _{0}$
and $\Lambda _{0}$ aims at distinguishing the true parameters from generic ones $\theta \in \Theta $ and $\Lambda \in \Pi $, where $%
\Theta $ and $\Pi $ represent the parameter spaces. 

As the RMINAR models (\ref{3.1a}), (\ref{3.14}) and (\ref{3.18}) have similar conditional means, being linear in the models' parameters, this work provides a unified estimation procedure based on a four-stage weighted least
squares estimation (4SWLSE) approach (see e.g.~Aknouche, $2015$). Let $\mu _{t}\equiv\mu
_{t}\left( \theta _{0}\right) :=\mathbb{E}\left( Y_{t}|\tciFourier _{t-1}^{Y}\right) $
and $V_{t}\equiv V_{t}\left(\Lambda _{0}\right) :=\mathbb{V}\left( Y_{t}|\tciFourier
_{t-1}^{Y}\right) $ be, respectively, the conditional mean and variance of the RMINAR process $Y_{t}$ in one of the three models (\ref{3.1a}), (\ref{3.14}) or (\ref{3.18}). In some cases, $V_{t}:=V_{t}\left( \theta
_0,\Lambda _{0}\right) $ also depends on $\theta _{0}$ (as in the multiplicative model). The principle and the rationale of the method are as follows. The 4SWLSE develops in 4 stages;
firstly, $\theta _{0}$
is first estimated by $\widehat{\theta}_{1n}$ from the regression $Y_{t}=\mu _{t}+e_{t}$, using a WLS estimator weighted by the conditional variance $%
V_{t}\left( \Lambda _{\ast }\right) $, evaluated at
some arbitrarily and fixed known vector $\Lambda _{\ast }\in \Pi $. The sequence $(e_{t})$
is a martingale difference with respect to $(\tciFourier _{t}^{Y})$. Secondly, $\Lambda _{0}$ is estimated by $\widehat{\Lambda}_{1n}$ from the regression $\left( Y_{t}-\mu _{t}\right) ^{2}=V_{t}+u_{t}$, where the conditional variance residual $(u_{t})$ is also a martingale difference with respect to $(\tciFourier _{t}^{Y})$. Thus, this step makes use of a WLS estimator weighted by $V_{t}^{2}(\Lambda_{*})$ arbitrarily evaluated. The third stage reestimates $\theta _{0}$ by $\widehat{\theta}_{2n}$ from the WLS approach used in the first stage but considering the estimated conditional variance as weights. The
same reasoning is used in the fourth stage to reestimate $%
\Lambda _{0}$, where the weight is $V_{t}^{2}(\widehat{\Lambda}_{1n})$. For any generic parameters $\theta \in \Theta $ and $%
\Lambda \in \Pi $, let $e_{t}\left( \theta \right) =Y_{t}-\mu _{t}\left(
\theta \right) $ and $u_{t}\left( \theta ,\Lambda \right) =e_{t}^{2}\left(
\theta \right) -V_{t}\left( \theta ,\Lambda \right) $, where $V_{t}\left( \theta ,\Lambda
\right) $ is the conditional variance function in which the true $\Lambda
_{0}$ is replaced by a generic $\Lambda $. Naturally, $e_{t}\left( \theta
_{0}\right) =e_{t}$ and $u_{t}\left( \theta _{0},\Lambda _{0}\right) =u_{t}$%
. Formally, the proposed 4SWLSE method is given
by the following cascade of four optimization problems%
\begin{eqnarray}
\text{i) }\widehat{\theta }_{1n} &=&\arg \min\limits_{\theta \in \Theta
}\sum_{t=1}^{n}\tfrac{e_{t}^{2}\left( \theta \right) }{V_{t}\left( \Lambda
_{\ast }\right) }\text{, ii)\ }\widehat{\Lambda }_{1n}=\arg
\min\limits_{\Lambda \in \Pi }\sum_{t=1}^{n}\tfrac{u_{t}^{2}\left( \widehat{%
\theta }_{1n},\Lambda \right) }{V_{t}^{2}(\Lambda_{*})} ,
\label{4.1a} \\
\text{iii) }\widehat{\theta }_{2n} &=&\arg \min\limits_{\theta \in \Theta
}\sum_{t=1}^{n}\tfrac{e_{t}^{2}\left( \theta \right) }{V_{t}\left( \widehat{%
\Lambda }_{1n}\right) }\text{, iv)\ }\widehat{\Lambda }_{2n}=\arg
\min\limits_{\Lambda \in \Pi }\sum_{t=1}^{n}\tfrac{u_{t}^{2}\left( \widehat{%
\theta }_{2n},\Lambda \right) }{V_{t}^{2}(\Lambda_{*})}\text{,}  
\label{4.1b}
\end{eqnarray}%
where the estimators $\widehat{\theta }_{\cdot n}$ and $\widehat{\Lambda }_{\cdot n}$ have the same notation as the corresponding target parameters, with an extra \textit{hat} symbol and the $n$ to refer to the sample size. Furthermore, the subscript $1$ or $2$ is used to highlight if the step performs estimation or re-estimation, respectively, of the corresponding parameter.

There is relevant literature connected with the 4SWLSE approach and the derivation of the properties of the estimators in \eqref{4.1a} and \eqref{4.1b} with respect to consistency, efficiency and asymptotic distribution. As an instance of that, Nicholls and Quinn $(1982)$ proposed the two-stage least squares estimation
(2SLSE) approach based on
\begin{equation}
\text{i) }\widehat{\theta }_{n}=\arg \min\limits_{\theta \in \Theta
}\sum_{t=1}^{n}e_{t}^{2}\left( \theta \right) \text{ \ and \ ii)\ }\widehat{%
\Lambda }_{n}=\arg \min\limits_{\Lambda \in \Pi
}\sum_{t=1}^{n}u_{t}^{2}\left( \widehat{\theta }_{1n},\Lambda \right),
\label{2SLSE}
\end{equation}%
to estimate $\theta _{0}$ and $\Lambda _{0}$ in a consecutive manner. Schick (1996) studied the properties of $\widehat{\theta }%
_{1n}$ in (\ref{4.1a})\ when $\Lambda _{\ast }$ is replaced by any consistent
estimator $\widehat{\Lambda }$ of $\Lambda _{0}$. For a RCAR$\left( 1\right) $ model, Aknouche $(2015)$ proposed a four-stage WLSE
similar to (\ref{4.1a}) considering the weights $V_{t}^{2}\left( \Lambda _{\ast }\right) $\ and $%
V_{t}\left( \widehat{\Lambda }_{1n}\right)$, respectively, for the second and fourth stages. The choice for the latter
weights results in an estimator for $\Lambda_0$ with the same asymptotic variance as that of the Gaussian QMLE. In our case, the 4SWLSE approach makes use of the latter weights in the second and third stages leading to estimators that are everywhere consistent and
asymptotically Gaussian with a certain optimality. Moreover, it will be proved that $\widehat{\theta }_{2n}$ and $\widehat{\Lambda }_{2n}$ \eqref{4.1b} are asymptotically
more efficient than $\widehat{\theta }_{1n}$ and $\widehat{\Lambda }_{1n}$ \eqref{4.1a},
respectively, and more efficient than the 2SLSE estimators $\widehat{\theta }_{n}$ and $\widehat{\Lambda}_{n}$ \eqref{2SLSE} and any
of the exponential family QMLE.

Note that for most usual discrete distributions, $\Lambda _{0}$ is a
function of $\theta _{0}$ and (possibly) some additional parameters. E.g.~$%
\Lambda _{0}=\theta _{0}$ when all coefficients are Poisson-distributed, and 
$\Lambda _{0}=\theta _{0}\oplus \left( 1_{1+p}+\theta _{0}\right) $ when are geometric distributed, where $1_{p}$ is a $p$-length vector of ones and $\oplus$ stands for the Hadamard matrix product. These cases only require a two-stage estimation procedure corresponding to the first and third stages of the 4SWLSE. If the
random coefficients $\Phi_{it}\sim \mathcal{NB}\left( \nu _{0i},\frac{\nu
_{0i}}{\nu _{i}+\phi _{0i}}\right) $ ($i=1,\dots,p$) and $\epsilon _{t}\sim 
\mathcal{NB}\left( \nu _{0},\frac{\nu _{0}}{\nu _{0}+\omega _{0}}\right) $
are negative binomial II distributed, then $\sigma _{\Phi _{i}}^{2}=\phi
_{0i}\left( 1+\frac{1}{\nu _{i}}\phi _{0i}\right) $ and $\sigma
_{\varepsilon }^{2}=\mu_{\epsilon}\left( 1+\frac{1}{\nu _{0}}\mu_{\epsilon}\right) 
$. Hence, $\nu _{0}$ and $\nu _{0i}>0$\ are estimated by $\widehat{\nu }%
_{n}^{-1}$ and $\widehat{\nu }_{in}^{-1}$, respectively, from the stage (iv) in \eqref{4.1b} using %
\begin{equation*}
\widehat{\nu }_{n}^{-1}=\tfrac{\widehat{\sigma }_{2n}^{2}-\widehat{\mu }%
_{\epsilon n}}{\widehat{\mu}_{\epsilon n}^{2}}\text{ \ and \ }\widehat{\nu }_{in}^{-1}=%
\tfrac{\widehat{\sigma }_{\Phi _{i}n}^{2}-\widehat{\phi }_{in}}{\widehat{%
\phi }_{in}^{2}}\text{, }i=1,\dots,p.
\end{equation*}

\subsection{The $\mathbb{N}_{0}$-valued RMINAR model}

Due to the mutual independence of the random sequences $\left\{ \Phi _{it},%
\text{ }t\in \mathbb{Z} \right\} \ (i=1,\dots,p)$ and $\left\{ \varepsilon _{t},\text{ }t\in \mathbb{Z}
\right\} $, the conditional mean and variance of the RMINAR model (\ref{3.1a}) have the
following linear-in-parameters expressions 
\begin{equation}
\mu _{t}\left( \theta _{0}\right) =\theta _{0}^{\prime }\mathcal{Y}_{t-1}%
\text{ \ and \ }V_{t}\left( \theta_0,\Lambda _{0}\right) =\mathcal{Z}_{t-1}^{\prime
}\Lambda _{0},
\label{eqweqqw}
\end{equation}%
where $\mathcal{Y}_{t-1}=\left( 1,Y_{t-1},\dots,Y_{t-p}\right) ^{\prime }\text{ and }%
\mathcal{Z}_{t-1}=\left( 1,Y_{t-1}^{2},\dots,Y_{t-p}^{2}\right) ^{\prime }=%
\mathcal{Y}_{t-1}\oplus\mathcal{Y}_{t-1}$.
Then, the estimators $[\widehat{\theta }_{1n},\widehat{\Lambda }%
_{1n},\widehat{\theta }_{2n},\widehat{\Lambda }_{2n}]$ in (\ref{4.1a})-(\ref{4.1b}) have the more specific closed-form expressions
\begin{eqnarray}
\text{i) }\widehat{\theta }_{1n} &=&\left( \sum_{t=1}^{n}\tfrac{\mathcal{Y}_{t-1}\mathcal{Y}_{t-1}^{\prime
}}{\mathcal{Z%
}_{t-1}^{\prime }\Lambda _{\ast }}\right) ^{-1}\sum_{t=1}^{n}\tfrac{\mathcal{Y}_{t-1}Y_{t}}{\mathcal{Z}%
_{t-1}^{\prime }\Lambda _{\ast }},  \label{4.3a} \\
\text{ii) }\widehat{\Lambda }_{1n} &=&\left( \sum_{t=1}^{n}\tfrac{\mathcal{Z}_{t-1}%
\mathcal{Z}_{t-1}^{\prime }}{
\left(
\mathcal{Z}_{t-1}^{\prime }\Lambda _{\ast }\right) ^{2}}\right) ^{-1}\sum_{t=1}^{n}\mathcal{Z}_{t-1}%
\tfrac{\left( Y_{t}-\mathcal{Y}_{t-1}^{\prime }\widehat{\theta }_{1n}\right)
^{2}}{\left( 
\mathcal{Z}_{t-1}^{\prime }\Lambda _{\ast }\right) ^{2}},%
\text{ \ \ \ \ \ \ }  \label{4.3b} \\
\text{iii) }\widehat{\theta }_{2n} &=&\left( \sum_{t=1}^{n}\tfrac{\mathcal{Y}_{t-1}\mathcal{%
Y}_{t-1}^{\prime }}{%
\mathcal{Z}_{t-1}^{\prime }\widehat{\Lambda }_{1n}}\right) ^{-1}\sum_{t=1}^{n}\tfrac{\mathcal{Y}_{t-1}Y_{t}}{%
\mathcal{Z}_{t-1}^{\prime }\widehat{\Lambda }_{1n}},  \label{4.3c} \\
\text{iv)\ }\widehat{\Lambda }_{2n} &=&\left( \sum_{t=1}^{n}\tfrac{\mathcal{Z}%
_{t-1}\mathcal{Z}_{t-1}^{\prime }}{\left(
\mathcal{Z}_{t-1}^{\prime }\widehat{\Lambda }_{1n}\right) ^{2}}\right) ^{-1}\sum_{t=1}^{n}\mathcal{Z}_{t-1}%
\tfrac{\left( Y_{t}-\mathcal{Y}_{t-1}^{\prime }\widehat{\theta }_{2n}\right)
^{2}}{\left( 
\mathcal{Z}_{t-1}^{\prime }\widehat{\Lambda }_{1n}\right)
^{2}},\text{ \ \ \ \ \ \ }  \label{4.3d}
\end{eqnarray}
where $\Lambda _{\ast }=\left( \sigma _{\ast \varepsilon }^{2},\sigma _{\ast
\phi _{1}}^{2},\dots,\sigma _{\ast \phi _{p}}^{2}\right) ^{\prime }>0$ is
arbitrarily fixed. When all random coefficients are Poisson or geometric
distributed, then the estimation steps (\ref{4.3b}) and (\ref{4.3d}) are useless but can be used to test the Poisson hypothesis $H_{0}:\theta _{0}=\Lambda _{0}$ or the geometric one $H_{0}:\Lambda _{0}=\theta _{0}\oplus \left( 1_{1+p}+\theta _{0}\right)$. 

The everywhere consistency and asymptotic normality of the estimators in  (\ref{4.3a})-(\ref{4.3d}) are established in Theorem \eqref{teo41}, with complete proof in Appendix B. Consider
\begin{eqnarray}
\Sigma \left( \Lambda _{0},\Lambda _{\ast }\right) :=A\left( \Lambda
_{\ast }\right) ^{-1}B\left( \Lambda _{0},\Lambda _{\ast }\right) A\left(
\Lambda _{\ast }\right) ^{-1}\text{\ and \ }\Omega \left( \Lambda _{\ast }\right)
:=C\left( \Lambda _{\ast }\right) ^{-1}D\left( \Lambda _{\ast }\right)
C\left( \Lambda _{\ast }\right) ^{-1}, \notag
\label{4.6c}
\end{eqnarray}
with auxiliary matrices $A, B, C$ and $D$ defined as
\begin{eqnarray}
A\left( \Lambda _{\ast }\right) :=\mathbb{E}\left( \tfrac{\mathcal{Y}_{t-1}\mathcal{Y}_{t-1}^{\prime
}}{\mathcal{Z}%
_{t-1}^{\prime }\Lambda _{\ast }}\right) &\text{,}& B\left( \Lambda _{0},\Lambda _{\ast }\right) :=\mathbb{E}\left( 
\tfrac{\left( \mathcal{Z}_{t-1}^{\prime }\Lambda _{0}\right) }{\left( 
\mathcal{Z}_{t-1}^{\prime }\Lambda _{\ast }\right) ^{2}}\mathcal{Y}_{t-1}%
\mathcal{Y}_{t-1}^{\prime }\right), \notag \label{4.6a} \\
C\left( \Lambda _{\ast }\right) :=\mathbb{E}\left( \tfrac{\mathcal{Z}_{t-1}%
\mathcal{Z}_{t-1}^{\prime }}{\left( %
\mathcal{Z}_{t-1}^{\prime }\Lambda _{\ast }\right) ^{2}}\right) &\text{,}& D\left( \Lambda _{\ast }\right)
:=\mathbb{E}\left( \tfrac{u_{t}^{2}}{\left( \mathcal{Z}_{t-1}^{\prime
}\Lambda _{\ast }\right) ^{4}}\mathcal{Z}_{t-1}\mathcal{Z}_{t-1}^{\prime
}\right). \notag  \label{4.6b}
\end{eqnarray}
Note that $B\left( \Lambda _{0},\Lambda _{0}\right) =A\left( \Lambda
_{0}\right)$ implies $\Sigma \left( \Lambda _{0},\Lambda _{0}\right) =A\left(
\Lambda _{0}\right) ^{-1}$ and that $D\left( \Lambda _{0}\right) =C\left(
\Lambda _{0}\right)$ leads to $\Omega \left( \Lambda _{0}\right) =C\left( \Lambda
_{0}\right) ^{-1}$. Also, recall that the finiteness of $\mathbb{E}\left( \varepsilon _{t}^{2}\right)$ and $\mathbb{E}\left( \Phi _{it}^{2}\right) $ holds by the model's definition.

\begin{theo}
\label{teo41}
 For the RMINAR model $(\ref{3.1a})$ under \textbf{A0}, for all $\Lambda _{\ast }>0$ it holds that%
\begin{eqnarray}
\text{i)} &&\widehat{\theta }_{1n}\overset{a.s.}{\underset{n\rightarrow
\infty }{\rightarrow }}\theta _{0}, \text{ \ ii) }\widehat{\theta }_{2n}%
\overset{a.s.}{\underset{n\rightarrow \infty }{\rightarrow }}\theta _{0},
\text{ \ iii) } \widehat{\Lambda }_{1n}\overset{a.s.}{\underset{n\rightarrow
\infty }{\rightarrow }}\Lambda _{0}, \text{ \ iv) }\widehat{\Lambda }_{2n}%
\overset{a.s.}{\underset{n\rightarrow \infty }{\rightarrow }}\Lambda _{0}. 
\label{4.7} 
\end{eqnarray}
If, in addition, $\mathbb{E}\left( \varepsilon _{t}^{4}\right) <\infty $ and $\mathbb{E}\left( \Phi _{it}^{4}\right) <1$ ($i=1,\dots,p$) then%
\begin{small}
\begin{eqnarray}
\text{i)} &&\sqrt{n}\left( \widehat{\theta }_{1n}-\theta _{0}\right) \overset%
{D}{\underset{n\rightarrow \infty }{\rightarrow }}\mathcal{N}\left( 0,\Sigma
\left( \Lambda _{0},\Lambda _{\ast }\right) \right),
\text{ii)} \sqrt{n}\left( \widehat{\theta }_{2n}-\theta _{0}\right) \overset%
{D}{\underset{n\rightarrow \infty }{\rightarrow }}\mathcal{N}\left( 0,\Sigma
\left( \Lambda _{0},\Lambda _{0}\right) \right) \notag \\
\text{ iii)\ } &&\sqrt{n}%
\left( \widehat{\Lambda }_{1n}-\Lambda _{0}\right) \overset{D}{\underset{%
n\rightarrow \infty }{\rightarrow }}\mathcal{N}\left( 0,\Omega \left(
\Lambda _{\ast }\right) \right),
\text{ iv)\ }\sqrt{n}\left( 
\widehat{\Lambda }_{2n}-\Lambda _{0}\right) \overset{D}{\underset{%
n\rightarrow \infty }{\rightarrow }}\mathcal{N}\left( 0,\Omega \left(
\Lambda _{0}\right) \right) .
\label{4.8}
\end{eqnarray}
\end{small}
\end{theo} 
\begin{proof}
See Appendix B.
\end{proof}

The most surprising fact given by the above result is that the proposed 4SWLSE estimators are consistent and asymptotically normal everywhere, even when
the RMINAR model has an infinite mean. In addition, $\widehat{\theta }_{2n}$ and $%
\widehat{\Lambda }_{2n}$ are more efficient than $\widehat{\theta }_{1n}$
and $\widehat{\Lambda }_{1n}$, respectively, as $\Sigma
\left( \Lambda _{0},\Lambda _{0}\right)\leq \Sigma
\left( \Lambda _{0},\Lambda _{\ast}\right)$ and $\Omega \left(
\Lambda _{0}\right) \leq \Omega \left(
\Lambda _{\ast}\right)$ for all $\Lambda_{\ast}$. The finiteness of the moments $%
\mathbb{E}\left( \varepsilon _{t}^{4}\right) <\infty $\textit{\ }and $\mathbb{E}\left( \Phi
_{it}^{4}\right) <1$\ ($i=1,\dots,p$), required for the asymptotic normality
to hold, does not constitute a restrictive condition as most usual discrete distributions (Poisson,
binomial, negative binomial, etc.) have finite higher-order moments.

Finally, note that the main advantage of precisely including the weights $%
\left( \mathcal{Z}_{t-1}^{\prime }\Lambda _{\ast }\right) ^{2}$ and $\left( 
\mathcal{Z}_{t-1}^{\prime }\widehat{\Lambda }_{1n}\right) ^{2}$ in the
second and fourth stages is to have the everywhere-consistency property of
the corresponding estimates $\widehat{\Lambda }_{.n}$. However, while $%
\widehat{\theta }_{2n}$ is optimal in the sense of Schick $(1996)$, the
resulting estimate $\widehat{\Lambda }_{2n}$ is not asymptotically optimal
in the class of all WLSEs. In fact, the optimal weight in the sense of
obtaining an estimate such that $C\left( \Lambda _{0}\right) $ and $D\left(
\Lambda _{0}\right) $ are proportional is $\mathbb{E}\left( u_{t}^{2}|\tciFourier
_{t-1}^{Y}\right) =\mathbb{V}\left( u_{t}|\tciFourier _{t-1}^{Y}\right) $ and the
resulting estimate would be more efficient than $\widehat{\Lambda }_{2n}$.
Nonetheless, the closed form expression of $\mathbb{V}\left( u_{t}|\tciFourier
_{t-1}^{Y}\right) $ in terms of model parameters is not always available and
depends on the model distribution which is not necessarily known in our
framework. Moreover, $\mathbb{V}\left( u_{t}|\tciFourier _{t-1}^{Y}\right) $ may
depend on some third- and fourth-moment parameters that need to be
specified. E.g.~the RMINAR(1) model has%
\begin{equation*}
\mathbb{V}\left( u_{t}|\tciFourier _{t-1}^{Y}\right) =\mathbb{E}\left( \Phi _{1t}-\phi
_{01}\right) ^{4}Y_{t-1}^{4}+6\sigma _{\phi }^{2}\sigma _{\varepsilon
}^{2}Y_{t-1}^{2}+\mathbb{E}\left( \left( \varepsilon _{t}-\mu _{\varepsilon }\right)
^{4}\right) -\left( \mathcal{Z}_{t-1}^{\prime }\Lambda _{0}\right) ^{2},
\end{equation*}%
which requires specifying $\mathbb{E}\left( \Phi _{1t}-\phi _{01}\right) ^{4}$ and $%
\mathbb{E}\left( \left( \varepsilon _{t}-\mu _{\varepsilon }\right) ^{4}\right) $ in
the model's definition. Therefore, the above weight consisting of the
squared conditional variance seems a reasonably good choice, especially
since for the ARCH model and the double auto-regression (DAR, Ling, $2007$)
model, $\left( \mathcal{Z}_{t-1}^{\prime }\Lambda _{0}\right) ^{2}$ and $%
\mathbb{V}\left( u_{t}|\tciFourier _{t-1}^{Y}\right) $ are proportional (see also
Aknouche, $2012$).

\subsection{The $\mathbb{Z}$-valued RMINAR model}

It is clear that the $%
\mathbb{Z}
$-valued RMINAR model \eqref{3.14} has a similar structure as its $\mathbb{N}_{0}$-valued counterpart (\ref{3.1a}). In particular, the conditional mean
and variance of (\ref{3.14}) are the same as those for (\ref{3.1a}) and are given
by (\ref{eqweqqw}). Thus the general 4SWLSE procedure (\ref{4.1a})-(\ref{4.1b}) reduces
for model (\ref{3.14}) to (\ref{4.3a})-(\ref{4.3d}). Keeping the same notations for the parameters, namely $\theta _{0}=\left( \mu _{0\varepsilon },\phi _{01},\dots,\phi _{0p}\right)
^{\prime }$ and $\Lambda _{0}=\left( \sigma _{0\varepsilon }^{2},\sigma
_{0\phi _{1}}^{2},\dots,\sigma _{0\phi _{p}}^{2}\right) ^{\prime }$, the estimators $[\widehat{\theta }_{1n},\widehat{\Lambda }_{1n},\widehat{\theta }%
_{2n},\widehat{\Lambda }_{2n}]$, and the asymptotic matrices $\Sigma$ and $\Omega$ in \eqref{4.8}, the
following result states that the same conclusions of Theorem \ref{teo41} hold for the $\mathbb{Z}$-valued RMINAR model \eqref{3.14}.
\begin{theo}
For the RMINAR model $(\ref{3.14})$ under \textbf{A0}*, the same conclusions of Theorem $\ref{teo41}$  hold for all arbitrarily and fixed known $%
\Lambda _{\ast }>0$.
\label{teorema42}
\end{theo} 

\begin{proof}
See Appendix B.
\end{proof}

\subsection{The RMINAR model with multiplicative errors}

For the RMINAR model with multiplicative errors \eqref{3.18}, the conditional mean (\ref{3.19a}) is linear in the parameter $\theta _{0}=\left( \omega
_{0},\phi _{01},\dots,\phi _{0p}\right) ^{\prime }$ as in models (\ref{3.1a}) and (\ref{3.14}) and rewrites as%
\begin{equation*}
\mu _{t}\left( \theta _{0}\right) :=\mathbb{E}\left( Y_{t}|\tciFourier _{t-1}\right)
=1+\mathcal{Y}_{t-1}^{\prime }\theta _{0}.
\end{equation*}%
In view of (\ref{3.19b}), however, it turns out that the conditional variance 
\begin{equation*}
V_t\left( \Lambda_0,\theta _{0}\right) :=\mathbb{V}\left( Y_{t}|\tciFourier _{t-1}\right)
=\left( \sigma_\varepsilon^2 +1\right) \left(\mathcal{Z}_{t-1}^{\prime } \Lambda_0\right) + \sigma_\varepsilon^2\left(1+\mathcal{Y}_{t-1}^{\prime }\theta _{0} \right)^2,
\end{equation*}%
is not linear in the variance parameter $\Lambda
_{0}=\left( \sigma _{0\varepsilon }^{2},\sigma _{0\omega }^{2},\sigma
_{0\phi _{1}}^{2},\dots,\sigma _{0\phi _{p}}^{2}\right) ^{\prime }$, which contains an extra parameter $\sigma _{0\varepsilon }^{2}$. Hence, the first and third
stages to estimate $\theta _{0}$\ are given in closed-form as for model (\ref{3.1a}) and (\ref{3.14}), but the second and fourth stages are now iterative.
Thus, the general 4SWLSE procedure (\ref{4.1a})-(\ref{4.1b}) reduces for model (\ref{3.18}) to
the following problems
\begin{eqnarray}
\text{i) }\widehat{\theta }_{1n} &=&\left( \sum_{t=1}^{n}\tfrac{\mathcal{Y}_{t-1}\mathcal{Y}_{t-1}^{\prime
}}{ V_{t}\left( \theta _{\ast},\Lambda
_{\ast}\right)}\right) ^{-1}\sum_{t=1}^{n}\tfrac{\mathcal{Y}_{t-1}Y_{t}}{V_{t}\left( \theta _{\ast},\Lambda
_{\ast}\right)}\text{ \ and \ ii) }\widehat{\Lambda }_{1n}=\arg
\min\limits_{\Lambda \in \Pi }\tfrac{1}{n}\sum_{t=1}^{n}\tfrac{%
u_{t}^{2}\left( \widehat{\theta }_{1n},\Lambda \right) }{V_{t}^2\left( \theta _{\ast},\Lambda
_{\ast}\right)}, \text{ \ \ \ \ \ \ \ \ } 
\label{4.9a} \\
\text{iii) }\widehat{\theta }_{2n} &=&\left( \sum_{t=1}^{n}\tfrac{\mathcal{Y}_{t-1}\mathcal{%
Y}_{t-1}^{\prime }}{V_{t}\left( \widehat{\theta} _{1 n},\widehat{\Lambda}
_{1 n}\right)}\right) ^{-1}\sum_{t=1}^{n}\tfrac{\mathcal{Y}_{t-1}Y_{t}}{%
V_{t}\left( \widehat{\theta} _{1 n},\widehat{\Lambda}
_{1 n}\right)}\text{ \ and \ iv)\ }\widehat{%
\Lambda }_{2n}=\arg \min\limits_{\Lambda \in \Pi }\tfrac{1}{n}\sum_{t=1}^{n}%
\tfrac{u_{t}^{2}\left( \widehat{\theta }_{2n},\Lambda \right) }{V_{t}^2\left( \widehat{\theta} _{2n},\widehat{\Lambda}
_{1 n}\right)}, \text{
\ \ \ \ \ \ \ \ }  \label{4.9b}
\end{eqnarray}%
where $\Lambda =\left( \sigma _{\varepsilon }^{2},\sigma _{\omega
}^{2},\sigma _{\phi _{1}}^{2},\dots,\sigma _{\phi _{p}}^{2}\right) ^{\prime
}\in \Pi $ and $\Lambda _{\ast }=\left( \sigma _{\ast \varepsilon
}^{2},\dots,\sigma
_{\ast \phi _{p}}^{2}\right) ^{\prime }>0$ are the generic
and the fixed known parameter vectors, respectively, and $\theta_{\ast }$ is fixed and known.

Next, Theorem \ref{teo43} establishes the consistency and asymptotic normality of the 4SWLSE given by (\ref{4.9a})-(\ref{4.9b}), relying on the assumptions that $\Lambda \in \Pi $ is compact and $\Lambda _{0}$ belongs to the interior of $\Pi$.
To state the result, further consider the matrices
\begin{equation*}
\Sigma \left( \theta _{\ast },\Lambda _{\ast }\right) := A\left( \theta
_{\ast },\Lambda _{\ast }\right) ^{-1}B\left( \theta _{\ast },\Lambda _{\ast
}\right) A\left( \theta _{\ast },\Lambda _{\ast }\right) ^{-1} \text{ and } \Omega \left( \Lambda _{\ast }\right) := C\left( \Lambda _{\ast }\right)
^{-1}D\left( \Lambda _{\ast }\right) C\left( \Lambda _{\ast }\right) ^{-1},
\label{4.11b}
\end{equation*}
where the auxiliary matrices $A, B, C$ and $D$ are defined as
\begin{eqnarray}
A\left( \theta _{\ast },\Lambda _{\ast }\right) :=\mathbb{E}\left( \tfrac{\mathcal{Y}_{t-1}%
\mathcal{Y}_{t-1}^{\prime }
}{V_{t}\left( \theta _{\ast },\Lambda _{\ast }\right) }\right) &\text{,}& B\left( \theta _{\ast },\Lambda
_{\ast }\right) :=\mathbb{E}\left( \tfrac{V_{t}\left( \theta _{\ast
},\Lambda _{0}\right) }{V_{t}\left( \theta _{\ast },\Lambda _{\ast }\right)
^{2}}\mathcal{Y}_{t-1}\mathcal{Y}_{t-1}^{\prime }\right),  \notag \\
C\left( \Lambda _{\ast }\right) :=\mathbb{E}\left( \tfrac{1}{V_{t}^{2}\left( \theta
_{0},\Lambda _{\ast }\right) }\tfrac{\partial V_{t}\left( \theta
_{0},\Lambda _{0}\right) }{\partial \Lambda }\tfrac{\partial V_{t}\left(
\theta _{0},\Lambda _{0}\right) }{\partial \Lambda ^{\prime }}\right) &\text{,}& D\left( \Lambda _{\ast }\right) :=\mathbb{E}\left( \tfrac{u_{t}^{2}\left( \theta
_{0},\Lambda _{0}\right) }{V_{t}^{4}\left( \theta _{0},\Lambda _{\ast
}\right) }\tfrac{\partial V_{t}\left( \theta _{0},\Lambda _{0}\right) }{%
\partial \Lambda }\tfrac{\partial V_{t}\left( \theta _{0},\Lambda
_{0}\right) }{\partial \Lambda ^{\prime }}\right).  \notag 
\end{eqnarray}

\begin{theo}
\label{teo43}
 \textit{Under model }$(\ref{3.18})$\textit{\ subject to 
\textbf{A0}** and \textbf{A1:} $\Lambda \in \Pi $ is compact,} 
\begin{eqnarray}
\text{i) \ } \widehat{\theta }_{1n}\overset{a.s.}{\underset{n\rightarrow
\infty }{\rightarrow }}\theta _{0}, 
\text{ \ ii) \ }\widehat{\theta }_{2n}%
\overset{a.s.}{\underset{n\rightarrow \infty }{\rightarrow }}\theta _{0}, 
\text{ \ iii) \ } \widehat{\Lambda }_{1n}\overset{a.s.}{\underset{n\rightarrow
\infty }{\rightarrow }}\Lambda _{0},
\text{ \ iv) \ }\widehat{\Lambda }_{2n}%
\overset{a.s.}{\underset{n\rightarrow \infty }{\rightarrow }}\Lambda _{0}. 
\label{4.12b}
\end{eqnarray}%
If, in addition, $\mathbb{E}\left( \varepsilon _{t}^{4}\right) <\infty$, $\mathbb{E}\left( \Phi _{it}^{4}\right) <1 \ (i=1,\dots,p)$ and \textbf{A2:} $\Lambda _{0}$ belongs to the interior of $\Pi$ holds, then
\begin{small}
\begin{eqnarray}
\text{i)} &&\sqrt{n}\left( \widehat{\theta }_{1n}-\theta _{0}\right) \overset%
{D}{\underset{n\rightarrow \infty }{\rightarrow }}\mathcal{N}\left( 0,\Sigma
\left( \Lambda _{0},\Lambda _{\ast }\right) \right),
\text{ii)} \sqrt{n}\left( \widehat{\theta }_{2n}-\theta _{0}\right) \overset%
{D}{\underset{n\rightarrow \infty }{\rightarrow }}\mathcal{N}\left( 0,\Sigma
\left( \Lambda _{0},\Lambda _{0}\right) \right) \notag \\
\text{ iii)\ } &&\sqrt{n}%
\left( \widehat{\Lambda }_{1n}-\Lambda _{0}\right) \overset{D}{\underset{%
n\rightarrow \infty }{\rightarrow }}\mathcal{N}\left( 0,\Omega \left(
\Lambda _{\ast }\right) \right), \text{ iv)\ }\sqrt{n}\left( 
\widehat{\Lambda }_{2n}-\Lambda _{0}\right) \overset{D}{\underset{%
n\rightarrow \infty }{\rightarrow }}\mathcal{N}\left( 0,\Omega \left(
\Lambda _{0}\right) \right).
\label{4.13}
\end{eqnarray}
 \end{small}  
\end{theo}

\begin{proof}
See Appendix B.
\end{proof}

For some particular cases, the estimators in \eqref{4.9a}-\eqref{4.9b} can be considerably simplified. This is the case of the Poisson and the geometric distributions, for which the random coefficient variances $\Delta _{0}:=\left( \sigma _{0\omega
}^{2},\sigma _{0\phi _{1}}^{2},\dots,\sigma _{0\phi _{p}}^{2}\right) ^{\prime
} $ are functions of the random coefficient means $\theta _{0}$ with no extra parameters. Indeed, only the variance $\sigma _{0\varepsilon }^{2}$ of the innovation sequence has to be estimated in $\Lambda _{0}=\left( \sigma _{0\varepsilon }^{2},\Delta _{0}^{\prime }\right) $, as the components in $\Delta _{0}^{\prime }$ are
functions of the mean parameters in $\theta _{0}$. Thus, the 4SWLSE procedure \eqref{4.9a}-\eqref{4.9b} simplifies to
the triplet $[\widehat{\theta }_{1n},\widehat{\sigma }_{\varepsilon n}^{2},%
\widehat{\theta }_{2n}]$ where $\widehat{\sigma }_{\varepsilon n}^{2}$ is an
estimator for $\sigma _{0\varepsilon }^{2}$. Since the RMINAR
model with multiplicative errors \eqref{3.18} shares the same variance-to-mean structure as the MthINGARCH
model of Aknouche and Scotto $(2023)$, the same estimator is used for $\sigma
_{0\varepsilon }^{2}$. Rewrite (\ref{3.19b}) as follows%
\begin{equation*}
V_{t}:=\mathbb{V}\left( Y_{t}|\tciFourier _{t-1}^{Y}\right) =V_{t}\left( \theta_0, \Lambda
_{0}\right) =\left( \sigma _{\varepsilon }^{2}+1\right) \delta _{t}^2+\sigma
_{\varepsilon }^{2}\mu _{t}^{2},
\end{equation*}%
where $\delta _{t}^2\equiv \delta _{t}^2\left( \Delta _{0}\right) =\mathbb{V}\left( \lambda
_{t}|\tciFourier _{t-1}^{Y}\right) =\Delta _{0}\mathcal{Z}_{t-1}^{\prime }$.
Let $\widehat{\mu }_{t}:=\mu _{t}\left( \widehat{\theta }_{1n}\right) =1+%
\mathcal{Y}_{t-1}^{\prime }\widehat{\theta }_{1n}$, and $\widehat{\delta}_t^2:=\delta _{t}^2\left( \widehat{\Delta }_{n}\right) $ where $\widehat{%
\Delta }_{n}$ is a consistent estimator for $\Delta _{0}$ obtained as the
corresponding function of $\widehat{\theta }_{1n}$. Hence, (\ref{4.9a})-(\ref{4.9b}) becomes%
\begin{eqnarray}
\text{i) }\widehat{\theta }_{1n} &=&\left( \sum_{t=1}^{n}\tfrac{\mathcal{Y}_{t-1}\mathcal{Y}_{t-1}^{\prime
}}{V_t(\theta_{*},\Lambda_{*})}\right) ^{-1}\sum_{t=1}^{n}\tfrac{\mathcal{Y}_{t-1}Y_{t}}{V_t(\theta_{*},\Lambda_{*})}, \label{4.10a} \\
\text{ii) }\widehat{\sigma }_{\varepsilon n}^{2} &=&\tfrac{1}{n}%
\sum\limits_{t=1}^{n}\tfrac{\left( Y_{t}-\widehat{\mu }_{t}\right) ^{2}-%
\widehat{\delta }_{t}^2}{\widehat{\delta }_{t}^2+\widehat{\mu }_{t}^{2}}, \label{4.10b} \\
\text{iii) }\widehat{\theta }_{2n} &=&\left( \sum_{t=1}^{n}\tfrac{\mathcal{Y}_{t-1}\mathcal{%
Y}_{t-1}^{\prime}}{V_t(\widehat{\theta}_{1n},\widehat{\Lambda}_{1n})}\right) ^{-1}\sum_{t=1}^{n}\tfrac{\mathcal{Y}_{t-1}Y_{t}}{%
V_t(\widehat{\theta}_{1n},\Lambda_{0})},  \label{4.10c}
\end{eqnarray}%
where $\widehat{\Lambda}_{1n}$ is the corresponding function of $\widehat{\theta}_{1n}$ and (\ref{4.10a}) is directly obtained from Aknouche and Scotto ($2024$, eq.~3.13). 
For this particular case, the consistency and distribution of the resulting estimators are organized in the following corollary to Theorem \ref{teo43}.

\begin{proposition}
\textit{Under model }$(\ref{3.18})$ subject to 
\textbf{A0}**
\begin{equation*}
i\textit{)} \ \widehat{\theta }_{1n}\overset{a.s.}{\underset{n\rightarrow \infty }%
{\rightarrow }}\theta _{0}\text{, ii) }\widehat{\sigma }_{\varepsilon n}^{2}%
\overset{a.s.}{\underset{n\rightarrow \infty }{\rightarrow }}\sigma
_{0\varepsilon }^{2}\text{, iii) }\widehat{\theta }_{2n}\overset{a.s.}{%
\underset{n\rightarrow \infty }{\rightarrow }}\theta _{0}.  \label{4.14}
\end{equation*}%
In addition,
\begin{eqnarray*}
\text{i)} &&\sqrt{n}\left( \widehat{\theta }_{1n}-\theta _{0}\right) \overset%
{D}{\underset{n\rightarrow \infty }{\rightarrow }}\mathcal{N}\left( 0,\Sigma
\left( \theta_{*},\Lambda _{\ast }\right) \right),\\
\text{ ii)} && \sqrt{n}%
\left( \widehat{\sigma }_{\varepsilon n}^{2}-\sigma _{0\varepsilon
}^{2}\right) \overset{\mathcal{D}}{\underset{n\rightarrow \infty }{%
\rightarrow }}\mathcal{N}\left( 0,\Gamma \right),   \label{4.15a} \\
\text{iii)} &&\sqrt{n}\left( \widehat{\theta }_{2n}-\theta _{0}\right) 
\overset{D}{\underset{n\rightarrow \infty }{\rightarrow }}\mathcal{N}\left(
0,\Sigma \left( \theta_{0},\Lambda _{0}\right) \right),  \label{4.15b}
\end{eqnarray*}%
\textit{where}%
\begin{equation*}
\Gamma :=\mathbb{E}\left( \tfrac{\left( Y_{t}-\mu _{t}\left( \theta _{0}\right)
\right) ^{2}-\left( \delta _{t}^2\left( \Delta _{0}\right) +\left( \delta
_{t}^2\left( \Delta _{0}\right) +\mu _{t}^{2}\left( \theta _{0}\right) \right)
\sigma _{0\varepsilon }^{2}\right) }{\delta _{t}^2\left( \Delta _{0}\right)
+\mu _{t}^{2}\left( \theta _{0}\right) }\right) ^{2}\text{.}  \label{4.16}
\end{equation*}
\end{proposition}

\section{Numerical illustrations}\label{section5}
The behavior of the 4SWLSE in finite samples is assessed
for the three above models via simulation experiments. Furthermore, two real data applications are analyzed. 
\subsection{Simulation study}
For the $%
\mathbb{N}
_{0}$-valued RMINAR model (\ref{3.1a}) two instances are considered. Firstly, a RMINAR$\left( 4\right) $ driven by Poisson inputs with three sets of parameters, namely: i) stationary RMINAR(4) model with finite mean, $%
\sum_{i=1}^{4}\phi _{0i}=0.7<1$, ii) stationary RMINAR(4) model with
infinite mean and $\sum_{i=1}^{4}\phi _{0i}=1$, and iii) stationary RMINAR$%
(4)$ model with infinite mean and $\sum_{i=1}^{4}\phi _{0i}=1.2>1$ (cf.
Table \ref{table5.1}). The second instance (cf. Table \ref{table5.2}), still concerns a $%
\mathbb{N}
_{0}$-valued RMINAR(3) model, but considering several distributions for the inputs,
namely{\small \ }$\varepsilon _{t}\sim \mathcal{B}in\left( 5,\frac{\mu
_{\varepsilon }}{5}\right) $ (binomial) $\Phi _{1t}\sim \mathcal{P}\left(
\phi _{01}\right) $ (Poisson), $\Phi _{2t}\sim \mathcal{NB}\left( 3,\frac{3}{%
3+\phi _{02}}\right) $ (NB2), and $\Phi _{3t}\sim \mathcal{NB}\left( 2\phi
_{03},\frac{2}{3}\right) $ (NB1). As previously, two instances of parameters
are considered: the stationary model with finite mean $\sum_{i=1}^{3}\phi
_{0i}<1$ and the stationary model with infinite mean, $\sum_{i=1}^{3}\phi
_{0i}>1$. Regarding the $%
\mathbb{Z}
$-valued RMINAR model, $p=3$ is considered where all inputs are Skellam
distributed with two cases of stationary, namely stationary with finite
mean, corresponding to $\rho \left( \Phi \right) =0.5761<1$, and stationary
with infinite mean and $\rho \left( \Phi \right) =1$, where%
\begin{equation*}
\Phi :=E\left( \Phi _{t}\right) =\left( 
\begin{array}{ccc}
\phi _{01} & \phi _{02} & \phi _{03} \\ 
1 & 0 & 0 \\ 
0 & 1 & 0%
\end{array}%
\right) 
\end{equation*}%
and $\rho \left( \Phi \right) $ denoting  the spectral radius of $\Phi $ (cf.
Table \ref{table5.3}). Finally, for the multiplicative RMINAR(3) model, all inputs are Poisson
distributed with three cases for the parameters just like for the $%
\mathbb{N}
_{0}$-valued RMINAR.

The 4SWLSE in (\ref{4.3a})-(\ref{4.3d}) and (\ref{4.9a})-(\ref{4.9b}) are run on $1000$ simulated series generated from the mentioned RMINAR models with sample-sizes $1000$. For all
instances, the means, the standard deviations (StD), and
asymptotic standard errors (ASE) of $\widehat{\theta }_{2n}$ and $\widehat{%
\Lambda }_{2n}$ are obtained over the $1000$ replications (cf. Tables \ref{table5.1}-\ref{table5.3}). For
the third model, only $%
\widehat{\theta }_{2n}$  and $\widehat{\sigma }_{\varepsilon n}^{2}$ are obtained, since all inputs are Poisson distributed. The ASEs of all estimates are obtained from the sample estimates of the
asymptotic variances in (\ref{4.6a})-(\ref{4.6c}) and (\ref{4.11b}), where the expectations are
replaced by their sample mean estimates.%

\begin{table}[ht]
\begin{center}
\begin{tabular}[t]{lc}
\hline
$n=1000$ & (a) $\sum_{i=1}^{4}\phi _{0i}=0.7<1$ \\ \hline
$%
\begin{array}{c}
\widehat{\theta }_{2n} \\ 
\text{Mean} \\ 
\text{StD} \\ 
\text{ASE}%
\end{array}%
$ & $%
\begin{array}{c}
\mu _{\varepsilon }=2 \\ 
2.1249 \\ 
0.3399 \\ 
0.2540%
\end{array}%
\begin{array}{c}
\phi _{01}=0.3 \\ 
0.2913 \\ 
0.0562 \\ 
0.0526%
\end{array}%
\begin{array}{c}
\phi _{02}=0.2 \\ 
0.1911 \\ 
0.0514 \\ 
0.0504%
\end{array}%
\begin{array}{c}
\phi _{03}=0.1 \\ 
0.0946 \\ 
0.0475 \\ 
0.0449%
\end{array}%
\begin{array}{c}
\phi _{04}=0.1 \\ 
0.0929 \\ 
0.0462 \\ 
0.0436%
\end{array}%
$ \\ \hline
$%
\begin{array}{c}
\widehat{\Lambda }_{2n} \\ 
\text{Mean} \\ 
\text{StD} \\ 
\text{ASE}%
\end{array}%
$ & $%
\begin{array}{c}
\sigma _{\varepsilon }^{2}=2 \\ 
1.9826 \\ 
0.1321 \\ 
0.0318%
\end{array}%
\begin{array}{c}
\sigma _{\phi _{1}}^{2}=0.3 \\ 
0.2829 \\ 
0.0056 \\ 
0.0019%
\end{array}%
\begin{array}{c}
\sigma _{\phi _{2}}^{2}=0.2 \\ 
0.1863 \\ 
0.0078 \\ 
0.0022%
\end{array}%
\begin{array}{c}
\sigma _{\phi _{3}}^{2}=0.1 \\ 
0.0901 \\ 
0.0074 \\ 
0.0020%
\end{array}%
\begin{array}{c}
\sigma _{\phi _{4}}^{2}=0.1 \\ 
0.0891 \\ 
0.0076 \\ 
0.0021%
\end{array}%
$ \\ \hline
$n=1000$ & (b) $\sum_{i=1}^{4}\phi _{0i}=1$ \\ \hline
$%
\begin{array}{c}
\widehat{\theta }_{2n} \\ 
\text{Mean} \\ 
\text{StD} \\ 
\text{ASE}%
\end{array}%
$ & $%
\begin{array}{c}
\mu _{\varepsilon }=1 \\ 
1.1151 \\ 
0.2513 \\ 
0.1567%
\end{array}%
\begin{array}{c}
\phi _{01}=0.4 \\ 
0.3907 \\ 
0.0574 \\ 
0.0537%
\end{array}%
\begin{array}{c}
\phi _{02}=0.3 \\ 
0.2932 \\ 
0.0605 \\ 
0.0539%
\end{array}%
\begin{array}{c}
\phi _{03}=0.1 \\ 
0.0936 \\ 
0.0462 \\ 
0.0447%
\end{array}%
\begin{array}{c}
\phi _{04}=0.2 \\ 
0.1963 \\ 
0.0574 \\ 
0.0497%
\end{array}%
$ \\ \hline
$%
\begin{array}{c}
\widehat{\Lambda }_{2n} \\ 
\text{Mean} \\ 
\text{StD} \\ 
\text{ASE}%
\end{array}%
$ & $%
\begin{array}{c}
\sigma _{\varepsilon }^{2}=1 \\ 
0.9455 \\ 
0.0841 \\ 
0.0130%
\end{array}%
\begin{array}{c}
\sigma _{\phi _{1}}^{2}=0.4 \\ 
0.3769 \\ 
0.0083 \\ 
0.0026%
\end{array}%
\begin{array}{c}
\sigma _{\phi _{2}}^{2}=0.3 \\ 
0.2818 \\ 
0.0082 \\ 
0.0028%
\end{array}%
\begin{array}{c}
\sigma _{\phi _{3}}^{2}=0.1 \\ 
0.0900 \\ 
0.0076 \\ 
0.0026%
\end{array}%
\begin{array}{c}
\sigma _{\phi _{4}}^{2}=0.2 \\ 
0.1823 \\ 
0.0122 \\ 
0.0032%
\end{array}%
$ \\ \hline
$n=1000$ & (c) $\sum_{i=1}^{4}\phi _{0i}=1.2>1$ \\ \hline
$%
\begin{array}{c}
\widehat{\theta }_{2n} \\ 
\text{Mean} \\ 
\text{StD} \\ 
\text{ASE}%
\end{array}%
$ & $%
\begin{array}{c}
\mu _{\varepsilon }=0.1 \\ 
0.1028 \\ 
0.0190 \\ 
0.0188%
\end{array}%
\begin{array}{c}
\phi _{01}=0.5 \\ 
0.5087 \\ 
0.0662 \\ 
0.0635%
\end{array}%
\begin{array}{c}
\phi _{02}=0.2 \\ 
0.2073 \\ 
0.0540 \\ 
0.0573%
\end{array}%
\begin{array}{c}
\phi _{03}=0.3 \\ 
0.2892 \\ 
0.0671 \\ 
0.0639%
\end{array}%
\begin{array}{c}
\phi _{04}=0.2 \\ 
0.2008 \\ 
0.0592 \\ 
0.0583%
\end{array}%
$ \\  \hline
$%
\begin{array}{c}
\widehat{\Lambda }_{2n} \\ 
\text{Mean} \\ 
\text{StD } \\ 
\text{ASE}%
\end{array}%
$ & $%
\begin{array}{c}
\sigma _{\varepsilon }^{2}=0.1 \\ 
0.0551 \\ 
0.0141 \\ 
0.0045%
\end{array}%
\begin{array}{c}
\sigma _{\phi _{1}}^{2}=0.5 \\ 
0.5048 \\ 
0.0241 \\ 
0.0043%
\end{array}%
\begin{array}{c}
\sigma _{\phi _{2}}^{2}=0.2 \\ 
0.1956 \\ 
0.0295 \\ 
0.0046%
\end{array}%
\begin{array}{c}
\sigma _{\phi _{3}}^{2}=0.3 \\ 
0.3147 \\ 
0.0322 \\ 
0.0046%
\end{array}%
\begin{array}{c}
\sigma _{\phi _{4}}^{2}=0.2 \\ 
0.1703 \\ 
0.0430 \\ 
0.0054 
\end{array}%
 $
 \\ \hline
 \end{tabular}
\caption{Mean, StD and ASE of $\widehat{\theta }_{2n}$ and $\widehat{\Lambda }_{2n}$ for a Poisson RMINAR(4) under (a) strict stationarity with finite mean, (b)-(c) strict stationarity with infinite mean.}
\label{table5.1}
\end{center}
\end{table}
\begin{center}
 (Table \ref{table5.1} here)
\end{center}

\begin{table}[ht]
\begin{center}
 \begin{tabular}[t]{lc}
\hline
$n=1000$ & (a) $\sum_{i=1}^{3}\phi _{0i}=0.6<1$ \\ \hline
$%
\begin{array}{c}
\widehat{\theta }_{2n} \\ 
\text{Mean} \\ 
\text{StD} \\ 
\text{ASE}%
\end{array}%
$ & $%
\begin{array}{c}
\mu _{\varepsilon }=2 \\ 
2.0834 \\ 
0.2228 \\ 
0.2120%
\end{array}%
\begin{array}{c}
\phi _{01}=0.3 \\ 
0.2952 \\ 
0.0518 \\ 
0.0538%
\end{array}%
\begin{array}{c}
\phi _{02}=0.2 \\ 
0.1929 \\ 
0.0467 \\ 
0.0519%
\end{array}%
\begin{array}{c}
\phi _{03}=0.1 \\ 
0.0866 \\ 
0.0443 \\ 
0.0480%
\end{array}%
$ \\ \hline
$%
\begin{array}{c}
\widehat{\Lambda }_{2n} \\ 
\text{Mean} \\ 
\text{StD} \\ 
\text{ASE}%
\end{array}%
$ & $%
\begin{array}{c}
\sigma _{\varepsilon }^{2}=2 \\ 
2.0834 \\ 
0.2228 \\ 
0.0390%
\end{array}%
\begin{array}{c}
\sigma _{\phi _{1}}^{2}=0.3 \\ 
0.2952 \\ 
0.0518 \\ 
0.0024%
\end{array}%
\begin{array}{c}
\sigma _{\phi _{2}}^{2}=0.2 \\ 
0.1929 \\ 
0.0467 \\ 
0.0037%
\end{array}%
\begin{array}{c}
\sigma _{\phi _{3}}^{2}=0.1 \\ 
0.0866 \\ 
0.0443 \\ 
0.0038%
\end{array}%
$ \\ \hline
$n=1000$ & (b) $\sum_{i=1}^{3}\phi _{0i}=1.1>1$ \\ \hline
$%
\begin{array}{c}
\widehat{\theta }_{2n} \\ 
\text{Mean} \\ 
\text{StD} \\ 
\text{ASE}%
\end{array}%
$ & $%
\begin{array}{c}
\mu _{\varepsilon }=0.5 \\ 
0.5261 \\ 
0.0859 \\ 
0.0711%
\end{array}%
\begin{array}{c}
\phi _{01}=0.3 \\ 
0.2907 \\ 
0.0501 \\ 
0.0491%
\end{array}%
\begin{array}{c}
\phi _{02}=0.2 \\ 
0.2019 \\ 
0.0492 \\ 
0.0481%
\end{array}%
\begin{array}{c}
\phi _{03}=0.6 \\ 
0.6127 \\ 
0.0757 \\ 
0.0682%
\end{array}%
$ \\ \hline
$%
\begin{array}{c}
\widehat{\Lambda }_{2n} \\ 
\text{Mean} \\ 
\text{StD} \\ 
\text{ASE}%
\end{array}%
$ & $%
\begin{array}{c}
\sigma _{\varepsilon }^{2}=0.45 \\ 
0.4216 \\ 
0.0425 \\ 
0.0100%
\end{array}%
\begin{array}{c}
\sigma _{\phi _{1}}^{2}=0.3 \\ 
0.2975 \\ 
0.0113 \\ 
0.0026%
\end{array}%
\begin{array}{c}
\sigma _{\phi _{2}}^{2}=0.2133 \\ 
0.2075 \\ 
0.0137 \\ 
0.0038%
\end{array}%
\begin{array}{c}
\sigma _{\phi _{3}}^{2}=0.9 \\ 
0.8947 \\ 
0.0227 \\ 
0.0063%
\end{array}%
$%
 \\ \hline
\end{tabular}
\caption{Mean, StD and ASE of $\widehat{\theta }_{2n}$ and $\widehat{
\Lambda }_{2n}$ for a RMINAR(3) with $\varepsilon
_{t}\sim \mathcal{B}in\left( 5,\frac{\mu _{\varepsilon }}{5}\right), 
\Phi _{1t}\sim \mathcal{P}\left( \phi _{01}\right), \Phi _{2t}\sim 
\mathcal{NB}\left( 3,\frac{3}{3+\phi _{02}}\right)$ and $\Phi
_{3t}\sim \mathcal{NB}\left( 2\phi _{03},\frac{2}{3}\right)$ in both
finite and infinite mean cases.}
\label{table5.2}
\end{center}
\end{table}
\begin{center}
 (Table \ref{table5.2} here)
\end{center}

\begin{table}[ht]
\begin{center}
\begin{tabular}[t]{lc}
\hline
$n=1000$ & (a) $%
\begin{array}{c}
\mu _{1\varepsilon }=0.7\text{, }\mu _{2\varepsilon }=0.3\text{, }\phi
_{11}=0.1\text{, }\phi _{21}=0.3 \\ 
\phi _{12}=0.2\text{, }\phi _{22}=0.1\text{, }\phi _{13}=0.4\text{, }\phi
_{23}=0.2 \\ 
\rho \left( \Phi \right) =0.5761<1%
\end{array}%
$ \\ \hline
$%
\begin{array}{c}
\widehat{\theta }_{2n} \\ 
\text{Mean} \\ 
\text{StD} \\ 
\text{ASE}%
\end{array}%
$ & $%
\begin{array}{c}
\mu _{\varepsilon }=0.4 \\ 
0.4020 \\ 
0.0646 \\ 
0.0658%
\end{array}%
\begin{array}{c}
\phi _{01}=-0.2 \\ 
-0.1954 \\ 
0.0394 \\ 
0.0408%
\end{array}%
\begin{array}{c}
\phi _{02}=0.1 \\ 
0.0958 \\ 
0.0365 \\ 
0.0399%
\end{array}%
\begin{array}{c}
\phi _{03}=0.2 \\ 
0.1978 \\ 
0.0449 \\ 
0.0457%
\end{array}%
$ \\ \hline
$%
\begin{array}{c}
\widehat{\Lambda }_{2n} \\ 
\text{Mean} \\ 
\text{StD} \\ 
\text{ASE}%
\end{array}%
$ & $%
\begin{array}{c}
\sigma _{\varepsilon }^{2}=1 \\ 
0.9904 \\ 
0.0417 \\ 
0.0079%
\end{array}%
\begin{array}{c}
\sigma _{\phi _{1}}^{2}=0.4 \\ 
0.3969 \\ 
0.0116 \\ 
0.0036%
\end{array}%
\begin{array}{c}
\sigma _{\phi _{2}}^{2}=0.3 \\ 
0.2945 \\ 
0.0153 \\ 
0.0041%
\end{array}%
\begin{array}{c}
\sigma _{\phi _{3}}^{2}=0.6 \\ 
0.5980 \\ 
0.0117 \\ 
0.0038%
\end{array}%
$ \\ \hline
& (b) $%
\begin{array}{c}
\mu _{1\varepsilon }=0.3=\mu _{2\varepsilon },\text{ }\phi _{11}=0.5\text{, }%
\phi _{21}=0.1 \\ 
\phi _{12}=0.4\text{, }\phi _{22}=0.1\text{, }\phi _{13}=0.5\text{, }\phi
_{23}=0.2 \\ 
\rho \left( \Phi \right) =1%
\end{array}%
$ \\ \hline
$%
\begin{array}{c}
\widehat{\theta }_{2n} \\ 
\text{Mean} \\ 
\text{StD} \\ 
\text{ASE}%
\end{array}%
$ & $%
\begin{array}{c}
\mu _{\varepsilon }=0 \\ 
-0.0537 \\ 
0.1403 \\ 
0.0681%
\end{array}%
\begin{array}{c}
\phi _{01}=0.4 \\ 
0.3917 \\ 
0.0611 \\ 
0.0450%
\end{array}%
\begin{array}{c}
\phi _{02}=0.3 \\ 
0.2803 \\ 
0.0709 \\ 
0.0466%
\end{array}%
\begin{array}{c}
\phi _{03}=0.3 \\ 
0.2767 \\ 
0.0675 \\ 
0.0516%
\end{array}%
$ \\ \hline
$%
\begin{array}{c}
\widehat{\Lambda }_{2n} \\ 
\text{Mean} \\ 
\text{StD} \\ 
\text{ASE}%
\end{array}%
$ & $%
\begin{array}{c}
\sigma _{\varepsilon }^{2}=0.6 \\ 
0.5809 \\ 
0.0687 \\ 
0.0104%
\end{array}%
\begin{array}{c}
\sigma _{\phi _{1}}^{2}=0.6 \\ 
0.5900 \\ 
0.0183 \\ 
0.0055%
\end{array}%
\begin{array}{c}
\sigma _{\phi _{2}}^{2}=0.5 \\ 
0.4939 \\ 
0.0143 \\ 
0.0033%
\end{array}%
\begin{array}{c}
\sigma _{\phi _{3}}^{2}=0.7 \\ 
0.6780 \\ 
0.0297 \\ 
0.0050%
\end{array}%
$%
\\ \hline
\end{tabular}
\caption{Mean, StD and ASE of $\widehat{\theta }_{2n}$ and $\widehat{
\Lambda }_{2n}$ for a $\mathbb{Z}$-valued RMINAR(3)  with 
$\varepsilon _{t}\sim Skellam\left( \mu _{1\varepsilon },\mu _{2\varepsilon
}\right), \Phi _{it}\sim Skellam\left( \phi _{1i},\phi _{2i}\right),\mu _{1\varepsilon }=\mu _{1\varepsilon }-\mu _{2\varepsilon },  
\sigma _{\varepsilon }^{2}=\mu _{1\varepsilon }+\mu _{2\varepsilon },
\phi _{0i}=\phi _{1i}-\phi _{2i}$, and $\sigma _{\phi _{i}}^{2}=\phi
_{1i}+\phi _{2i}$.}
\label{table5.3}
\end{center}
\end{table}
\begin{center}
 (Table \ref{table5.3} here)
\end{center}


\begin{table}[ht]
\begin{center}
\begin{tabular}[t]{lc}
\hline
$n=1000$ & (a) $\sum_{i=1}^{4}\phi _{0i}=0.7<1$ \\ \hline
$%
\begin{array}{c}
(\widehat{\theta }_{2n},\widehat{\sigma }_{\varepsilon n}^{2}) \\ 
\text{Mean} \\ 
\text{StD} \\ 
\text{ASE}%
\end{array}%
$ & $%
\begin{array}{c}
\omega _{0}=1 \\ 
0.9855 \\ 
0.1521 \\ 
0.1553%
\end{array}%
\begin{array}{c}
\phi _{01}=0.4 \\ 
0.3981 \\ 
0.0659 \\ 
0.0683%
\end{array}%
\begin{array}{c}
\phi _{02}=0.3 \\ 
0.3039 \\ 
0.0620 \\ 
0.0638%
\end{array}%
\begin{array}{c}
\sigma _{\varepsilon }^{2}=1 \\ 
1.0042 \\ 
0.1222 \\ 
0.0287%
\end{array}%
$ \\ \hline
$n=1000$ & (b) $\sum_{i=1}^{4}\phi _{0i}=1$ \\ \hline
$%
\begin{array}{c}
(\widehat{\theta }_{2n},\widehat{\sigma }_{\varepsilon n}^{2}) \\ 
\text{Mean} \\ 
\text{StD} \\ 
\text{ASE}%
\end{array}%
$ & $%
\begin{array}{c}
\omega _{0}=1 \\ 
1.0168 \\ 
0.1609 \\ 
0.1669%
\end{array}%
\begin{array}{c}
\phi _{01}=0.5 \\ 
0.5008 \\ 
0.0692 \\ 
0.0738%
\end{array}%
\begin{array}{c}
\phi _{02}=0.5 \\ 
0.5031 \\ 
0.0814 \\ 
0.0768%
\end{array}%
\begin{array}{c}
\sigma _{\varepsilon }^{2}=1 \\ 
1.0064 \\ 
0.1207 \\ 
0.0273%
\end{array}%
$ \\ \hline
$n=1000$ & (c) $\sum_{i=1}^{4}\phi _{0i}=1.1>1$ \\ \hline
$%
\begin{array}{c}
(\widehat{\theta }_{2n},\widehat{\sigma }_{\varepsilon n}^{2}) \\ 
\text{Mean} \\ 
\text{StD} \\ 
\text{ASE}%
\end{array}%
$ & $%
\begin{array}{c}
\omega _{0}=1 \\ 
1.0075 \\ 
0.1579 \\ 
0.1684%
\end{array}%
\begin{array}{c}
\phi _{01}=0.5 \\ 
0.5013 \\ 
0.0820 \\ 
0.0727%
\end{array}%
\begin{array}{c}
\phi _{02}=0.6 \\ 
0.6069 \\ 
0.0817 \\ 
0.0826%
\end{array}%
\begin{array}{c}
\sigma _{\varepsilon }^{2}=1 \\ 
1.0472 \\ 
0.1458 \\ 
0.0304%
\end{array}%
$%
\\ \hline
\end{tabular}
\caption{Mean, StD and ASE of $\widehat{\theta }_{2n}$ and %
$\widehat{\sigma }_{\varepsilon n}^{2}$ for a multiplicative
\text{ RMINAR}(4) with Poisson inputs under (a) strict
stationarity with finite mean, and (b)-(c) strict stationarity with infinite mean.} 
\label{table5.4}
\end{center}
\end{table}
\begin{center}
 (Table \ref{table5.4} here)
\end{center}

From Tables \ref{table5.1}-\ref{table5.4} some general conclusions can be drawn. Firstly, the 4SWLE generates fairly good estimates with small bias and small standard errors for the three models. In particular, StDs and ASEs are quite close to each other. Secondly, the 4SWLE is convergent regardless of the placement of the parameters in the strict stationary domain, which is the entire parameter space. In particular, the estimates seem insensitive to the existence or not of the model mean. Other unreported simulations show that the results are consistent with the asymptotic theory in that the larger the sample size,
the more accurate the results. 

\subsection{Application to the Euro-Pound sterling exchange rate data}
The RMINAR model with either additive or multiplicative forms (\ref{3.1a} and \ref{3.17}) is applied to the number of ticks per minute
changes in the Euro to British Pound exchange rate (ExRate for short) on
December 12, from 9:00 $a.m.$ to 9:00 $p.m.$. The dataset is taken from Gorgi $%
(2020)$ who applied (heavy-tailed) observation-driven count models to
it. The ExRate series was also considered by Aknouche and Scotto (2024)
using the MthINGARCH$\left( 1,1\right) $ model. The dataset contains 720
observations and is highly over-dispersed with a sample mean and variance of 13.2153 and 224.2498, respectively (see Figure \ref{fig1}).%

\begin{figure}[h]
    \centering
\includegraphics[width=6.1in]{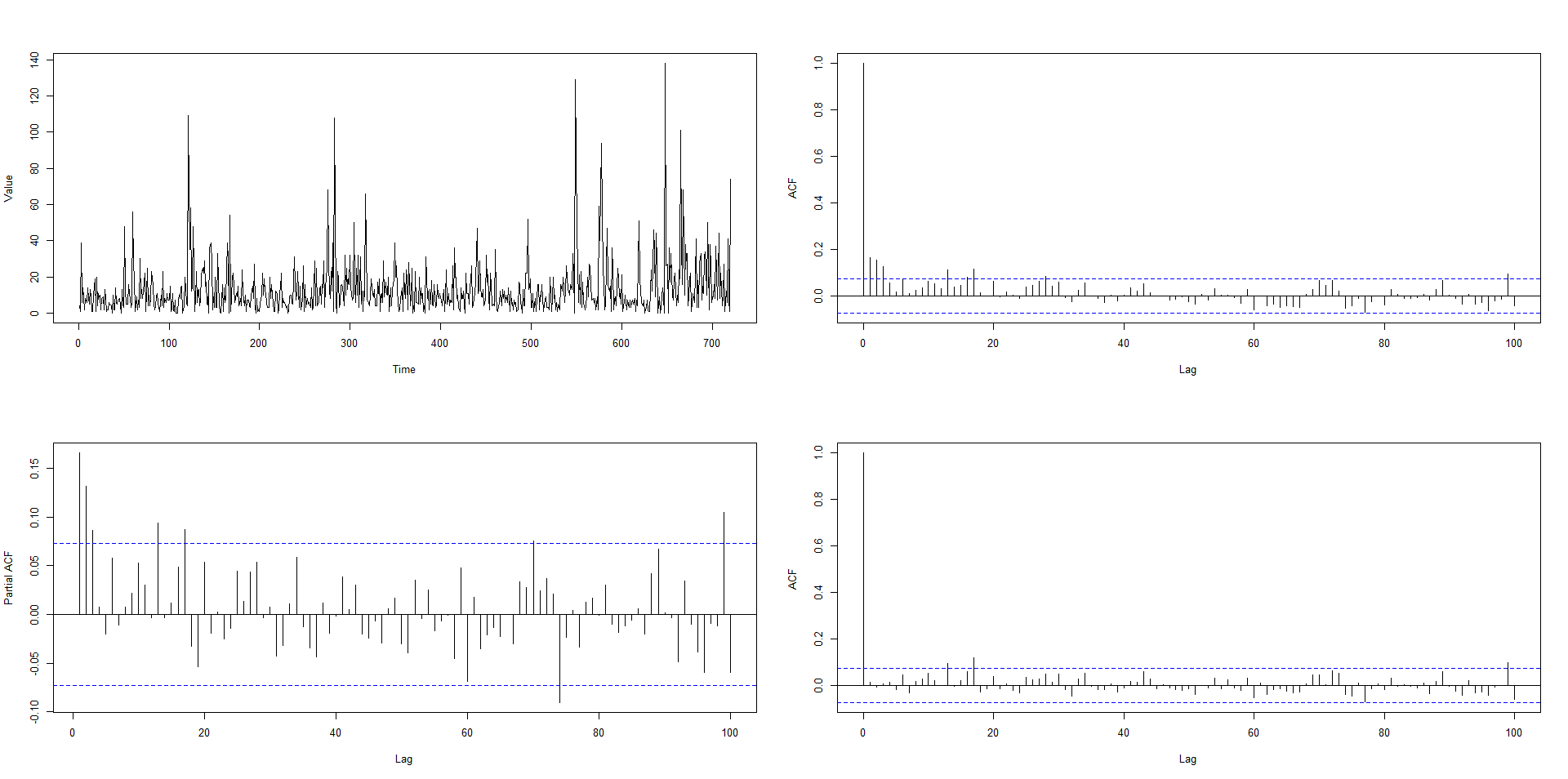}
\caption{(a) The ExRate series, (b) sample auto-correlation, (c) sample partial auto-correlation, (d) sample auto-correlation of the Pearson residuals of the additive RMINAR(3) model.}
\label{fig1}
    \end{figure}
\begin{center}
 (Figure \ref{fig1} here)
\end{center}

The first step in modeling the ExRate series consists of identifying the
order $p$ of the RMINAR$\left( p\right)$ models. Since
the likelihood of the unspecified RMINAR$\left( p\right) $ models is not
simple to obtain due to the model random coefficients, standard
information-based criteria (AIC, BIC, etc.) are out of reach. However, the
additive and multiplicative RMINAR$\left( p\right)$ models (\ref{3.1a}) and (\ref{3.17}), having an RCAR
form, have the same auto-covariance structure as an AR model since they are
both a weak AR$\left( p\right) $. Thus, the identification of $p$ can be
carried out by inspecting the shape of the sample partial auto-correlation
of the series (cf. Figure \ref{fig1}) and reinforced by other evocative criteria. For example, the generated
unconditional mean and variance (cf. (\ref{3.12}) and (\ref{3.13}),
respectively) by the additive RMINAR$\left( p\right) $ model could be
compared to the sample mean and sample variance of the series. Finally, the
mean absolute residual (MAR) defined as MAR$:=\tfrac{1}{n}%
\sum\limits_{t=1}^{n}\left\vert Y_{t}-\widehat{\mu }_{t}\right\vert $, the
mean square residual (MSR) MSR$:=\tfrac{1}{n}\sum\limits_{t=1}^{n}\left(
Y_{t}-\widehat{\mu }_{t}\right) ^{2}$, and the mean square Pearson residual,
MSPR$:=\tfrac{1}{n}\sum\limits_{t=1}^{n}\frac{\left( Y_{t}-\widehat{\mu }%
_{t}\right) ^{2}}{\widehat{V}_{t}}$ are also used. Note that the sample partial auto-correlation of the series (cf. Figure 5.1) shows a pronounced cut-off after the lag $p=3$, which suggests that $p=3$ is an adequate choice, but the values $p=1$ and $p=2$ are also possible. Table \ref{table5.5} displays the mean, variance, MAR, and MSPR of the estimated additive and multiplicative RMINAR$(p)$ models for $p=\in \{1,2\}$. 
For the additive RMINAR$\left(p\right) $ model, $p=3$ leads to the smallest MAR and MSR while the MSPR is comparable to those given
by $p\in \{1,2\}$.\ Additionally, $p=3$ provides the generated mean closest to the
sample mean which is not the case for the generated variance.
The multiplicative RMINAR$\left(3\right)$ model provides
the smallest MAR and MSR while the best MSPR is recorded for $p=1$, which
generates the closest mean. So the analysis proceeds with $p=3$, but following
the parsimony principle, $p\in\{1,2\}$ are also considered for comparison. The choice for optimal $p$ will be based on the out-of-sample forecast
ability of the RMINAR$\left(p\right)$ model in the set $p\in \left\{ 1,2,3\right\}$, as shown
below (see Table \ref{table11} in the Supplementary material).

\begin{table}[ht]
\begin{center}
\begin{tabular}{crccccc}
\hline\hline
A & $p$ & $\mathbb{E}\left( Y\right) $ & $\mathbb{V}\left( Y\right) $ & MAR
& MSR & MSPR \\ \hline
& $1$ & 13.2026 & 223.1584 & 9.3449 & 217.5623 & 1.0000 \\ \hline
& $2$ & 13.2793 & 254.0780 & 9.3048 & 213.7554 & 0.9996 \\ \hline
& $3$ & 13.2255 & 303.9467 & 9.2035 & 211.0007 & 1.0005 \\ \hline\hline
M & $p$ & $\mathbb{E}\left( Y\right) $ & - & MAR & MSR & MSPR \\ \hline
& $1$ & 12.0650 & - & 9.3569 & 217.5965 & 0.9958 \\ \hline
& $2$ & 11.8887 & - & 9.3381 & 214.0984 & 1.1203 \\ \hline
& $3$ & 11.6881 & - & 9.2736 & 212.7709 & 1.2195 \\ \hline\hline
\end{tabular}
\caption{Selecting the best order $p$, considering additive (A) or multiplicative (M) RMINAR($p$) models on the ExRate time series.}
\label{table5.5}
\end{center}
\end{table}
\begin{center}
 (Table \ref{table5.5} here)
\end{center}

The estimation of the RMINAR$\left(p\right)$ parameters is performed with the WLS procedure described in Section \ref{4SWLSE}, where the scheme in (\ref{2SLSE}) is directly applied for the additive RMINAR$\left( p\right)$
model. Regarding the
multiplicative RMINAR$\left( p\right) $ model, the procedure (\ref{4.1a}) is
used for simplicity and to inspect the impact of assuming a given
variance-to-mean relationship for the random parameters. In particular, the parameter variances are assumed proportional
to the parameter means,
i.e. $\left( \sigma _{\omega }^{2},\sigma _{\alpha _{1}}^{2},\dots,\sigma
_{\alpha _{p}}^{2}\right) =c\left( \omega ,\alpha _{1},\dots,\alpha
_{p}\right) $ where the constant of proportionality $c=0.135$ was estimated as to minimize the $\left\vert \mbox{MSPR-1}\right\vert$. For instance, unreported
results showed that both the Poisson assumption $\left( \sigma _{\omega
}^{2},\sigma _{\alpha _{1}}^{2},\dots,\sigma _{\alpha _{p}}^{2}\right) =\left(
\omega ,\alpha _{1},\dots,\alpha _{p}\right) $ and the geometric assumption $%
\left( \sigma _{\omega }^{2},\sigma _{\alpha _{1}}^{2},\dots,\sigma _{\alpha
_{p}}^{2}\right) =\left( \omega ,\alpha _{1},\dots,\alpha _{p}\right) \oplus
\left( \omega +1,\alpha _{1}+1,\dots,\alpha _{p}+1\right) $ give bad MSPR values.
Another reason for not applying directly the general nonlinear procedures %
(\ref{4.9a} ii) and (\ref{4.9b} ii) for estimating the variances of the
multiplicative RMINAR$\left( p\right) $ model is that the estimates were sensitive to the choice of the initial
parameter values in the optimization routine.

The initial values for the weights used for
the estimation procedure (\ref{2SLSE}) follow the general principle of starting the 4SWLSE with arbitrary weights. Therefore, the 4SWLE is
applied as many times as possible to arrive at very close successive
estimates. At each iteration, the estimated variance parameters are injected
as initial weights for the conditional variance in the next iteration. It
was concluded that all procedures stabilized at most at the sixth iteration.

Table \ref{table5.6} displays the parameter estimates and estimated asymptotic
standard errors for the RMINAR$\left( p\right) $ models
with $p\in\{1,2,3\}$. All parameters are significant and the corresponding models
are stationary in mean. In addition, all estimated additive RMINAR$%
\left( p\right) $ models are second-order stationary (SSC) as the SSC condition (\ref{3.9}) is verified. Note that the selected additive RMINAR$\left( 3\right) $ model has, in general, better in-sample properties than the MthINGARCH model (Aknouche and Scotto, 2024) since in the former case, the MAR and MSR are smaller while the MSPR is closer to 1. In addition, the RMINAR$\left(3\right)$ generates an unconditional mean that is closer to the sample mean of the series. However, the unconditional variance of the RMINAR$\left(3\right)$ is not the best out of the set of considered models, and the RMINAR$(1)$ is that providing a better variance than that of the MthINGARCH model.%

\begin{table}[ht]
\begin{center}
\begin{tabular}{lccl||ccccc}
\hline\hline
A & $p=1$ & $p=2$ & $p=3$ & M & $p=1$ & $p=2$ & $p=3$ \\ \hline
$\widehat{\mu }_{\varepsilon }$ & $\underset{\left( 0.2540\right) }{11.1727}$
& $\underset{\left( 0.3672\right) }{9.3941}$ & $\underset{\left(
0.4674\right) }{8.2255}$ & $\widehat{\mu }_{\varepsilon }$ & $1.0029$ & $%
1.0035$ & $1.0082$ \\ \hline
$\widehat{\phi }_{1}$ & $\underset{\left( 0.0298\right) }{0.1537}$ & $%
\underset{\left( 0.0343\right) }{0.1277}$ & $\underset{\left( 0.0382\right) }%
{0.1346}$ & $\widehat{\omega }$ & $\underset{\left( 0.8275\right) }{9.9797}$
& $\underset{\left( 0.8404\right) }{7.9889}$ & $\underset{\left(
0.8148\right) }{6.9170}$ \\ \hline
$\widehat{\phi }_{2}$ & - & $\underset{\left( 0.0356\right) }{0.1649}$ & $%
\underset{\left( 0.0393\right) }{0.1520}$ & $\widehat{\alpha }_{1}$ & $%
\underset{\left( 0.0457\right) }{0.1728}$ & $\underset{\left( 0.0495\right) }%
{0.1610}$ & $\underset{\left( 0.0544\right) }{0.1854}$ \\ \hline
$\widehat{\phi }_{3}$ & - & - & $\underset{\left( 0.0356\right) }{0.0914}$ & 
$\widehat{\alpha }_{2}$ & - & $\underset{\left( 0.0524\right) }{0.1670}$ & $%
\underset{\left( 0.0516\right) }{0.1441}$ \\ \hline
$\widehat{\sigma }_{\varepsilon }^{2}$ & $\underset{\left( 8.5459\right) }{%
212.4713}$ & $\underset{\left( 8.0192\right) }{157.1334}$ & $\underset{%
\left( 7.4415\right) }{112.2556}$ & $\widehat{\alpha }_{3}$ & - & - & $%
\underset{\left( 0.0557\right) }{0.0787}$ \\ \hline
$\widehat{\sigma }_{\phi _{1}}^{2}$ & $\underset{\left( 0.0099\right) }{%
0.0136}$ & $\underset{\left( 0.0084\right) }{0.0103}$ & $\underset{\left(
0.0108\right) }{0.0183}$ & $\widehat{\sigma }_{\varepsilon }^{2}$ & $%
\underset{\left( 0.0733\right) }{1.0143}$ & $\underset{\left( 0.0508\right) }%
{0.7118}$ & $\underset{\left( 0.0390\right) }{0.5585}$ \\ \hline
$\widehat{\sigma }_{\phi _{2}}^{2}$ & - & $\underset{\left( 0.0290\right) }{%
0.1854}$ & $\underset{\left( 0.0266\right) }{0.1519}$ & - & - & - & - \\ 
\hline
$\widehat{\sigma }_{\phi _{3}}^{2}$ & - & - & $\underset{\left(
0.0293\right) }{0.1882}$ & - & - & - & - \\ \hline
\text{SSC} & $0.0372$ & $0.3701$ & $0.5238$ & - & - & - & - \\ \hline\hline
\end{tabular}
\caption{4SWLSE estimates and asymptotic standard errors (in parenthesis) for the additive (A) or multiplicative (M) RMINAR$(p)$ models on the ExRate time series. The SSC stands for second-order stationary condition.}
\label{table5.6}
\end{center}
\end{table}
\begin{center}
 (Table \ref{table5.6} here)
\end{center}

Finally, the out-of-sample forecasting ability of the RMINAR$%
(3)$ models is examined and compared to that of the MthINGARCH(1,1) (Aknouche and Scotto, 2024). The setting $p=3$ is selected beause it yields the best out-of-sample forecast performance among the RMINAR models (see the Supplementary material). The RMINAR$\left( 3\right) $ models were estimated using the first $%
n_{c}$\ observations of the series ($1<n_{c}<n=720$). Then, the one-step ahead forecasts, $\widehat{\mu }_{t}$, over the period ($%
n_{c}+1,\dots,n$) were computed. The evaluation was based on i)\ the mean square forecast error given by MSFE $=\frac{1}{%
n-n_{c}}$\ $\sum\limits_{t=n_{c}+1}^{n}(Y_{t}-\widehat{\mu }_{t})^{2}$, ii)
the mean absolute forecast error, MAFE $=\frac{1}{n-n_{c}}%
\sum\limits_{t=n_{c}+1}^{n}\left\vert Y_{t}-\widehat{\mu }_{t}\right\vert $,
and iii) the mean squared Pearson forecast error, MSPFE $=\frac{1}{n-n_{c}}%
\sum\limits_{t=n_{c}+1}^{n}\frac{\left( Y_{t}-\widehat{\mu }_{t}\right) ^{2}%
}{\widehat{V}_{t}}$. The estimated conditional variance is given by $\widehat{V}_{t}:=\mathcal{Z}_{t-1}^{\prime }\widehat{\Lambda }%
_{n}$ for the additive model and by $\widehat{V}_{t}:=(\widehat{\sigma }%
_{n}^{2}+1)\mathcal{Z}_{t-1}^{\prime }\widehat{\Delta }_{n}+\widehat{\sigma }%
_{n}^{2}\widehat{\mu }_{t}^{2}$\ for the multiplicative model. Table \ref{table5.7} displays MSFE, MAFE, and MSPFE for the RMINAR models estimated from the series with sample size $n_{c}\in
\left\{ 300,420,520,680,700\right\} $. 
The corresponding results for the MthINGARCH(1,1) are directly taken from Aknouche and Scotto (2024). Note that the additive and multiplicative RMINAR$(p)$ models provide, in general, a comparable out-of-sample ability to the
MthINGARCH model. Specifically, for some values of $n_{c}$, especially those closer to the sample size $n=720$, the
RMINAR models provide better MSFE and MAFE. In contrast, the MthINGARCH slightly
outperforms the RMINARs for smaller $n_{c}$ values. This could be explained by the
fact that the MthINGACH model, having a moving average component, is more
persistent (and then has higher memory) than the RMINAR models. Note also that the multiplicative RMINAR provides the
worst MSPFE compared to the remaining models due to the assumed
variance-to-mean relationship.

\begin{table}[ht]
\begin{center}
\begin{tabular}{lc|ccccc}
\hline\hline
& $n_{c}$ & $300$ & $420$ & $520$ & $680$ & $700$ \\ \hline
A RMINAR(3) & $%
\begin{array}{c}
\text{MSFE} \\ 
\text{MAFE} \\ 
\text{MSPFE}%
\end{array}%
$ & \multicolumn{1}{|l}{$%
\begin{array}{c}
241.7308 \\ 
9.4727 \\ 
1.3535%
\end{array}%
$} & \multicolumn{1}{l}{$%
\begin{array}{c}
295.1503 \\ 
10.2364 \\ 
1.7079%
\end{array}%
$} & \multicolumn{1}{l}{$%
\begin{array}{c}
398.7492 \\ 
11.7017 \\ 
2.3473%
\end{array}%
$} & \multicolumn{1}{l}{$%
\begin{array}{c}
309.7075 \\ 
13.3228 \\ 
1.2327%
\end{array}%
$} & \multicolumn{1}{l}{$%
\begin{array}{c}
366.5074 \\ 
13.1845 \\ 
1.4226%
\end{array}%
$} \\ \hline
M RMINAR(3) & $%
\begin{array}{c}
\text{MSFE} \\ 
\text{MAFE} \\ 
\text{MSPFE}%
\end{array}%
$ & \multicolumn{1}{|l}{$%
\begin{array}{c}
241.1297 \\ 
9.4708 \\ 
1.6006%
\end{array}%
$} & \multicolumn{1}{l}{$%
\begin{array}{c}
288.3198 \\ 
10.0013 \\ 
1.8979%
\end{array}%
$} & \multicolumn{1}{l}{$%
\begin{array}{c}
396.5317 \\ 
11.6276 \\ 
2.3258%
\end{array}%
$} & \multicolumn{1}{l}{$%
\begin{array}{c}
313.0352 \\ 
13.3125 \\ 
2.1731%
\end{array}%
$} & \multicolumn{1}{l}{$%
\begin{array}{c}
369.2506 \\ 
13.1978 \\ 
2.1345%
\end{array}%
$} \\ \hline
MthINGARCH & $%
\begin{array}{c}
\text{MSFE} \\ 
\text{MAFE} \\ 
\text{MSPFE}%
\end{array}%
$ & \multicolumn{1}{|l}{$%
\begin{array}{c}
239.611 \\ 
9.4884 \\ 
1.1744%
\end{array}%
$} & \multicolumn{1}{l}{$%
\begin{array}{c}
291.263 \\ 
10.169 \\ 
1.5721%
\end{array}%
$} & \multicolumn{1}{l}{$%
\begin{array}{c}
395.830 \\ 
11.769 \\ 
2.2917%
\end{array}%
$} & \multicolumn{1}{l}{$%
\begin{array}{c}
305.882 \\ 
13.340 \\ 
1.2246%
\end{array}%
$} & \multicolumn{1}{l}{$%
\begin{array}{c}
363.530 \\ 
13.319 \\ 
1.6525%
\end{array}%
$} \\ \hline \hline
\end{tabular}
\caption{Out-of-sample forecasting ability of the additive (A)
RMINAR$\left(3\right)$, the multiplicative (M) RMINAR$\left(3\right)$, and the MthINGARCH$\left(1,1\right)$ on the ExRate time series. The $n_c$ stands for the number of observations used for parameter estimation. The MthINGARCH$\left(1,1\right)$ results are obtained from Aknouche and Scotto ($2024$).}
\label{table5.7}
\end{center}
\end{table}

\begin{center}
 (Table \ref{table5.7} here)
\end{center}

\subsection{Application to the daily (integer-valued) stock return series}
The second application concerns fitting the daily stock returns (Return) of
Bank of America from July 1, 2016 to September 28, 2018. The dataset with a
total of 566 observations is taken from Xu and Zhu $(2022)$ who divided the
returns by tick price to get a signed integer-valued series. The sample mean and variance of the series are 0.1184 and
978.8231, respectively, showing a strong overdispersion (see Figure \ref{fig2}). The authors proposed a $%
\mathbb{Z}
$-valued (GARCH-like) conditional volatility model for the series.
Such an approach could ignore the conditional mean effect and, in contrast, the RMINAR
model can represent both the conditional mean and conditional
volatility observed in the data. 

\begin{figure}[h]
    \centering
\includegraphics[width=6.1in]{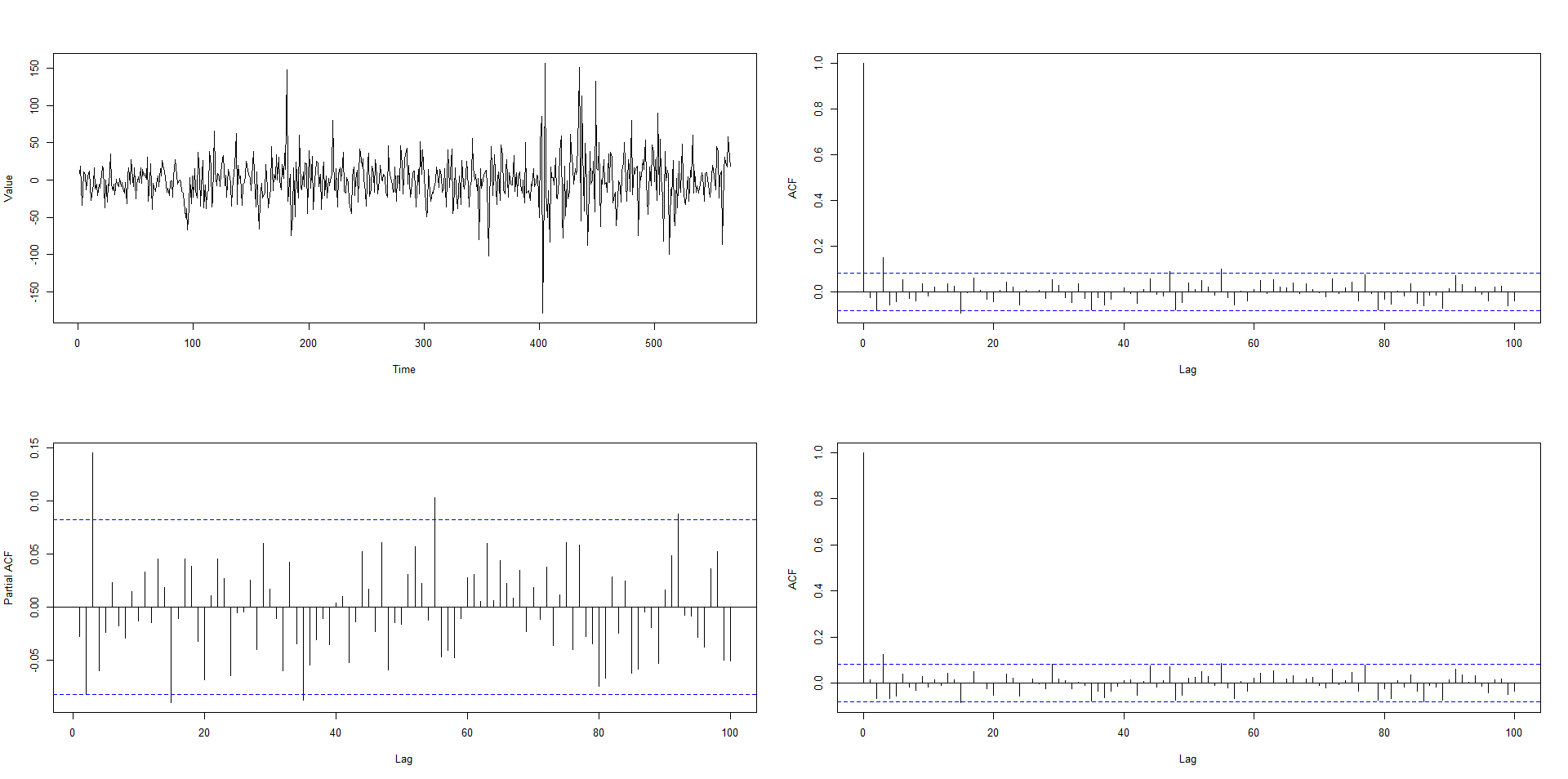}
\caption{(a) The Return series, (b) sample auto-correlation, (c) sample partial auto-correlation, (d) sample auto-correlation of the Pearson residuals of the $\mathbb{Z}$-valued RMINAR(1) model.}
\label{fig2}
    \end{figure}
\begin{center}
 (Figure \ref{fig2} here)
\end{center}

The sample partial auto-correlation of the series shows a small and
isolated peak at lag $3$ while the remaining auto-correlations are fairly small. Thus, as in the above ExRate example, the orders $p\in\{1,2,3\}$ were considered. Table \ref{table5.8} shows some measures to support the choice of the model order $p$. It
can be seen that $p=1$ gives the best mean and variance approximations,
while $p=3$ provides the smallest MAR and MSR and simultaneously the worst MSPR.

\begin{table}[ht]
\begin{center}
\begin{tabular}{cccccc}
\hline\hline
$p$ & $\mathbb{E}\left( Y\right) $ & $\mathbb{V}\left( Y\right) $ & MAR & MSR
& MSPR \\ \hline
$1$ & $0.1596$ & $999.885$ & $21.9768$ & $976.2304$ & $0.9982$ \\ \hline
$2$ & $-0.0584$ & $962.0471$ & $21.8832$ & $971.7091$ & $0.9961$ \\ \hline
$3$ & $0.0484$ & $904.8307$ & $21.6772$ & $951.4233$ & $1.0972$ \\ 
\hline\hline
\end{tabular}
\caption{Selecting the best order $p$ for the $\mathbb{Z}$-valued RMINAR$\left( p\right)$ model on the Return time series.} 
\label{table5.8}
\end{center}
\end{table}
\begin{center}
 (Table \ref{table5.8} here)
\end{center}
Parameter estimation is carried out using the algorithm (\ref{2SLSE}) with results reported in Table \ref{table5.9}. All estimated models are second-order stationary and all parameters are significant except $%
\sigma_{\phi _{3}}^{2}$ which is zero due to the artificial
non-negative least squares constraint used. Thus the random coefficient $%
\Phi _{t3}$ is degenerated at $\widehat{\phi}_{3}=0.1021$ and since $\Phi
_{t3}$ should be integer valued this implies it should be zero, which excludes
the order $p=3$. On the other hand, the retained orders $p=1$ and $p=2$
provide comparable in-sample performance (cf. Table \ref{table5.8}).

\begin{table}[ht]
\begin{center}
\begin{tabular}{llll}
\hline\hline
& $p=1$ & $p=2$ & $p=3$ \\ \hline
$\widehat{\mu }_{\varepsilon }$ & $\underset{\left( 0.0733\right) }{-0.03067}
$ & $\underset{\left( 0.0892\right) }{-0.0610}$ & $\underset{\left(
0.1161\right) }{0.0450}$ \\ \hline
$\widehat{\phi }_{1}$ & $\underset{\left( 0.0075\right) }{0.1645}$ & $%
\underset{\left( \ 0.0074\right) }{-0.0120}$ & $\underset{\left(
0.0091\right) }{-0.0092}$ \\ \hline
$\widehat{\phi }_{2}$ & - & $\underset{\left( 0.0097\right) }{-0.0322\text{ }%
}$ & $\underset{\left( 0.0127\right) }{-0.0229}$ \\ \hline
$\widehat{\phi }_{3}$ & - & - & $\underset{\left( 0.0092\right) }{0.1021}$
\\ \hline
$\widehat{\sigma }_{\varepsilon }^{2}$ & $\underset{\left( 41.9269\right) }{%
783.9099}$ & $\underset{\left( 41.9794\right) }{671.1159}$ & $\underset{%
\left( 38.0848\right) }{535.3662}$ \\ \hline
$\widehat{\sigma }_{\phi _{1}}^{2}$ & $\underset{\left( 0.0468\right) }{%
0.2151}$ & $\underset{\left( 0.0422\right) }{0.1781}$ & $\underset{\left(
0.0461\right) }{0.2484}$ \\ \hline
$\widehat{\sigma }_{\phi _{2}}^{2}$ & - & $\underset{\left( 0.0363\right) }{%
0.1231}$ & $\underset{\left( 0.0368\right) }{0.1488}$ \\ \hline
$\widehat{\sigma }_{\phi _{3}}^{2}$ & - & - & $\underset{\left(
0.0158\right) }{0.0000}$ \\ \hline
\text{SSC} & $0.1741$ & $0.2799$ & $0.4004$ \\ \hline\hline
\end{tabular}
\caption{4SWLSE results for the $\mathbb{Z}$-valued RMINAR$(p)$ model on the Return time series. Same caption as in Table \ref{table5.6}.}
\label{table5.9}
\end{center}
\end{table}
\begin{center}
 (Table \ref{table5.9} here)
\end{center}

The truncated series with $n_{c}\in \left\{
300,350,400,450,500\right\} $ are used to analyze the out-of-sample performance of the RMINAR$\left( p\right) $ model. Table \ref{table5.10} shows that the order $p=2$ provides the smallest MSFE and
MAFE, for the sizes $n_{c}\in \left\{
300,350,400,550\right\}$. However, the order $p=1$ gives better measures when $n_{c}\in \left\{
450,500\right\} $. Overall, there is no dominant order for all sample sizes $%
n_{c}$ and $p=1$ is be preferred following the parsimony principle.

\begin{table}[ht]
\begin{center}
\begin{tabular}{lc|cccccc}
\hline\hline
$p$ & $n_{c}$ & $300$ & $350$ & $400$ & $450$ & $500$ & $550$ \\ \hline
1 & $%
\begin{array}{c}
\text{MSFE} \\ 
\text{MAFE} \\ 
\text{MSPFE}%
\end{array}%
$ & \multicolumn{1}{|l}{$%
\begin{array}{c}
1404.3517 \\ 
26.0506 \\ 
2.2259%
\end{array}%
$} & \multicolumn{1}{l}{$%
\begin{array}{c}
1581.2839 \\ 
27.5554 \\ 
2.5298%
\end{array}%
$} & \multicolumn{1}{l}{$%
\begin{array}{c}
1854.6241 \\ 
30.1032 \\ 
2.9709%
\end{array}%
$} & \multicolumn{1}{l}{$%
\begin{array}{c}
1042.7839 \\ 
24.2280 \\ 
1.2397%
\end{array}%
$} & \multicolumn{1}{l}{$%
\begin{array}{c}
1106.3495 \\ 
24.5137 \\ 
1.2872%
\end{array}%
$} & $%
\begin{array}{c}
1203.23862 \\ 
27.60783 \\ 
1.37160%
\end{array}%
$ \\ \hline
2 & $%
\begin{array}{c}
\text{MSFE} \\ 
\text{MAFE} \\ 
\text{MSPFE}%
\end{array}%
$ & \multicolumn{1}{|l}{$%
\begin{array}{c}
1413.7965 \\ 
26.1138 \\ 
1.9150%
\end{array}%
$} & \multicolumn{1}{l}{$%
\begin{array}{c}
1582.6909 \\ 
27.5715 \\ 
2.5863%
\end{array}%
$} & \multicolumn{1}{l}{$%
\begin{array}{c}
1843.1118 \\ 
30.0951 \\ 
2.7693%
\end{array}%
$} & \multicolumn{1}{l}{$%
\begin{array}{c}
1060.6812 \\ 
24.6735 \\ 
1.3997%
\end{array}%
$} & \multicolumn{1}{l}{$%
\begin{array}{c}
1122.3612 \\ 
24.8328 \\ 
1.4153%
\end{array}%
$} & $%
\begin{array}{c}
1196.0980 \\ 
27.6674 \\ 
1.5050%
\end{array}%
$ \\ \hline \hline
\end{tabular}
\caption{Out-of-sample forecasting ability of the $\mathbb{Z}$-valued RMNINAR$\left( p\right)$ model on the Return time series.}
\label{table5.10}
\end{center}
\end{table}
\begin{center}
 (Table \ref{table5.10} here)
\end{center}

\section{Conclusion}\label{section6}

This paper proposed a random multiplication operator (RMO) to construct
integer-valued time series models in both $%
\mathbb{N}
$-valued and $%
\mathbb{Z}
$-valued cases. Compared to the random sum operator (RSO) which generally
requires specifying the full distribution (e.g.~binomial, geometric,
Poisson, etc.), the RMO is semi-parametric in the sense that it only specifies
the mean and variance. Additionally, the RMO allows larger over-dispersion and can be more heavy tailed but, more importantly, the variance does not
necessarily have to depend on the mean as is the case for most RSOs. This
allows more flexibility in modeling since the operator variance is estimated
in a separate procedure independently of the operator mean. Finally, the RMO
is simpler than the RSO because it consists of the usual
multiplication of the constant and variable operands.

This paper also shows how to build up simple and tractable RMO-based
integer-valued time series models that consist of auto-regressive type models with integer-valued random coefficients. Although
random coefficient time series models have been well-known for a long time,
their use was entirely focused on real-valued random coefficients. In
contrast, integer-valued random coefficients make the underlying model
universally (or everywhere) stationary and ergodic when generated from
stationary and ergodic inputs. An important consequence is that most
important M estimation methods such as quasi-maximum likelihood and weighted
least squares (WLS) are consistent and asymptotically Normal everywhere
for all parameter components. This makes the models more flexible than other
existing classes of integer-valued time series models mentioned in the
introduction.


The RMINAR framework can be extended in many perspectives. Firstly, the
persistence ability of the RMINAR models could be improved by proposing
INARMA models based on random multiplications with some moving average
dynamics. In particular, INGARCH-type models based on thinning (Aknouche and
Scotto, 2024) could be adapted while replacing the RSO with the RMO. Second,
multivariate RMINAR forms could be easily introduced and estimated using the
methods in Nicholls and Quinn $(1982)$ and Praskova and Vanecek $(2011)$.
These models for both $\mathbb{N}$-valued and $\mathbb{Z}$-valued series could constitute very straightforward alternatives to multivariate INGARCH models and multivariate INAR-type models.

\section*{Funding}

This work was partially funded by the Foundation for Science and Technology, FCT (\url{https://www.fct.pt/}), Portugal through national (MEC) and European structural (FEDER) funds, in the scope of the research projects IEETA/UA (UIDB/00127/2020, \url{www.ieeta.pt}) and CEMAT/IST/UL (UIDB/04621/2020, \url{http://cemat.ist.utl.pt}). 

\begin{appendices}
\section{Proof of Theorem  \ref{teo1}}
The proof of Theorem \ref{teo1} relies on Lemma \ref{lemmaA} and the proofs of \ref{teo2} and \ref{teo3} follow similarly.

\begin{lemma}
\label{lemmaA}
Consider the stochastic recurrence equation
\begin{equation}
Z_{t}=A_{t}Z_{t-1}+B_{t},\text{ \ }t\in \mathbb{Z},
\label{a1} \nonumber
\end{equation}%
driven by the iid sequence of pairs $\left\{ \left( A_{t}\mathbf{,%
}B_{t}\right) ,t\in \mathbb{Z} \right\}$ such that $A_{t}$ is the companion matrix defined in \eqref{eq: At} and $\left\{ a_{it},\text{ }t\in \mathbb{Z} \right\} (i=1,\dots,p)$ are iid and mutually
independent, satisfying
\begin{equation}
\mathbb{P}\left( \left( a_{1t},\dots,a_{pt}\right) ^{\prime }=\mathbf{0}_{p\times
1}\right) =\prod\limits_{i=1}^{p}\mathbb{P}\left( a_{it}=0\right) >0.  
\label{a3}
\end{equation}%
Then, the random series
\begin{equation}
\sum\limits_{j=0}^{\infty }\prod\limits_{i=0}^{j-1}A_{t-i}B_{t-j},  
\label{a4}
\end{equation}%
converges absolutely a.s.~for all $t\in \mathbb{Z}$.\\
\end{lemma}
\begin{proof}{\textbf{(Lemma \ref{lemmaA})}} \\
For $p=1$, the result in \eqref{a4} directly follows from Vervaat ($1979$, Theorem 1.6 (c); Lemma 1.7). See also Brandt ($1986$, Theorem 1 and p.~215) for the case of stationary and ergodic non-necessarily independent
sequences. For $p>1$, by \eqref{a3}, the iid property of $\left\{ \left( A_{t}%
\mathbf{,}B_{t}\right) ,t\in \mathbb{Z} \right\} $, the mutual independence of $\left\{ a_{it},\text{ }t\in \mathbb{Z} \right\} \ (i=1,\dots,p)$, and the form \eqref{eq: At} of the companion matrix $A_{t}$, it follows that
\begin{equation*}
\mathbb{P}\left( A_{t}A_{t-1}\cdots A_{t-p+1}=\mathbf{0}_{p}\right)
=\prod\limits_{i=1}^{p}\mathbb{P}\left( a_{it}=0\right) >0.
\end{equation*}
Indeed, it can be easily seen that the matrix $A_{t}A_{t-1}\cdots A_{t-p+1}$
contains no element 1 and that all of its elements are formed by algebraic
operations of $\left\{ \left( a_{1k},\dots, a_{pk}\right) \text{, }%
k=t-p+1,\dots, t\right\} $ vectors that satisfy \eqref{a3}. Therefore,%
\[
\#\left\{ j\in \mathbb{N}_{0}:\prod\limits_{i=0}^{j-1}A_{t-i}B_{t-j}=0\right\} =\infty \text{ a.s.,}
\]%
where $\#A\in \mathbb{N}
_{0}\cup \left\{ \infty \right\} $ is the number of elements of the
set $A$. Thus, the series in \eqref{a4} has only finitely many positive terms a.s., which implies its a.s.~absolute convergence.
\end{proof}\\

\begin{proof}{\textbf{(Theorem \ref{teo1})}}\\
The fact that%
\[
\sum\limits_{j=0}^{\infty }\prod\limits_{i=0}^{j-1}A_{t-i}B_{t-j}<\infty \text{ a.s.,}
\]%
immediately follows from Lemma \ref{lemmaA}. By definition, $\mathbf{Y}_{t}:=\sum%
\limits_{j=0}^{\infty }\prod\limits_{i=0}^{j-1}A_{t-i}B_{t-j}
$, $t\in \mathbb{Z}$ and then $\left\{ \mathbf{Y}_{t}, \mbox{ } t\in \mathbb{Z} \right\}$ is a solution of the RMINAR model in the vector form \eqref{3.4}. Such a solution is unique and causal
(i.e.~future independent). Moreover, it is strictly stationary and ergodic
in view of the iid property of the random inputs of the RMINAR model \eqref{3.1a}.
\end{proof}\\






\section{Proof of Theorems \ref{teo41}, \ref{teorema42} and \ref{teo43}}
The proof of Theorem \ref{teo41} is first presented by introducing some lemmas (\ref{a.2}-\ref{a.6}) that support the results in the theorem. Then, the proof of Theorem \ref{teorema42} is omitted as it is quite similar to that of Theorem \ref{teo41}. The proofs of some results conveyed in Theorem \ref{teo43} are also omitted as they can be derived similarly as in the proof of Theorem \ref{teo41}.

\begin{lemma}
Consider the RMINAR model \eqref{3.1a} under \textbf{A0}.  For all positive real $\left( a_{i}\right) {(i=0,\dots ,p)}$ and all integers $r,s,k$ such that $r+s\leq 2k$, it holds that
\begin{equation}
\tfrac{1}{n}\sum\limits_{t=1}^{n}\tfrac{Y_{t-i}^{r}Y_{t-j}^{s}}{\left(
a_{0}+a_{1}Y_{t-1}^{2}+\cdots+a_{p}Y_{t-p}^{2}\right) ^{k}}\overset{a.s.}{%
\underset{n\rightarrow \infty }{\rightarrow }}\mathbb{E}\left( \tfrac{%
Y_{t-i}^{r}Y_{t-j}^{s}}{\left(
a_{0}+a_{1}Y_{t-1}^{2}+\cdots+a_{p}Y_{t-p}^{2}\right) ^{k}}\right) \text{, }%
i,j=1,\dots,p\text{.}  \label{A.4}
\end{equation}
\label{a.2}
\end{lemma}
\begin{proof}{\textbf{(Lemma \ref{a.2})}}

The result in (\ref{A.4}) follows from the stationar and ergodicity
of $\left\{ Y_{t},t\in 
\mathbb{Z}
\right\} $ under \textbf{A0}, the a.s. boundedness of $\tfrac{%
Y_{t-i}^{r}Y_{t-j}^{s}}{\left(
a_{0}+a_{1}Y_{t-1}^{2}+\cdots +a_{p}Y_{t-p}^{2}\right) ^{k}}$\ which entails $%
\mathbb{E}\left( \tfrac{Y_{t-i}^{r}Y_{t-j}^{s}}{\left(
a_{0}+a_{1}Y_{t-1}^{2}+\cdots+a_{p}Y_{t-p}^{2}\right) ^{k}}\right) <\infty $ and the ergodic theorem. 
\end{proof}
\begin{lemma}
Consider the RMINAR model \eqref{3.1a} under \textbf{A0}.  Then, for all $\Lambda _{\ast }>\mathbf{0}_{p+q+1}$,
\begin{eqnarray}
\tfrac{1}{n}\sum_{t=1}^{n}\mathcal{Z}_{t-1}\left( \tfrac{\left( Y_{t}-%
\mathcal{Y}_{t-1}^{\prime }\widehat{\theta }_{1n}\right) ^{2}}{\left( 
\mathcal{%
Z}_{t-1}^{\prime }\Lambda _{\ast }\right) ^{2}}-\tfrac{\left( Y_{t}-\mathcal{%
Y}_{t-1}^{\prime }\theta _{0}\right) ^{2}}{\left( 
\mathcal{Z}_{t-1}^{\prime }\Lambda
_{\ast }\right) ^{2}}\right)  &=&o_{a.s.}\left( 1\right),   \label{A.5} \\
\tfrac{1}{\sqrt{n}}\sum_{t=1}^{n}\mathcal{Z}_{t-1}\left( \tfrac{\left( Y_{t}-%
\mathcal{Y}_{t-1}^{\prime }\widehat{\theta }_{1n}\right) ^{2}}{\left(
\mathcal{%
Z}_{t-1}^{\prime }\Lambda _{\ast }\right) ^{2}}-\tfrac{\left( Y_{t}-\mathcal{%
Y}_{t-1}^{\prime }\theta _{0}\right) ^{2}}{\left( 
\mathcal{Z}_{t-1}^{\prime }\Lambda
_{\ast }\right) ^{2}}\right)  &=&o_{p}\left( 1\right),  \label{A.6} 
\end{eqnarray}
where $o_{a.s.}\left( 1\right) $ and $o_{p}\left(
1\right) $ converge to zero, respectively, as $n\rightarrow \infty $ a.s.~and in probability.
\label{a.3}
\end{lemma}
\begin{proof}{\textbf{(Lemma \ref{a.3})}}

i) A Taylor expansion of the function $g\left( \widehat{\theta }_{1n}\right) =\left( Y_{t}-\mathcal{Y}_{t-1}^{\prime }%
\widehat{\theta }_{1n}\right) ^{2}$ around $\theta _{0}$\ gives    
\begin{eqnarray}
&&\tfrac{1}{n}\sum_{t=1}^{n}\mathcal{Z}_{t-1}\tfrac{\left( Y_{t}-\mathcal{Y}%
_{t-1}^{\prime }\widehat{\theta }_{1n}\right) ^{2}}{\left(
Z_{t-1}^{\prime }\Lambda _{\ast }\right) ^{2}}=\tfrac{1}{n}\sum_{t=1}^{n}%
\tfrac{\left( Y_{t}-\mathcal{Y}_{t-1}^{\prime }\theta _{0}\right) ^{2}}{%
\left( 
\mathcal{Z}_{t-1}^{\prime }\Lambda _{\ast }\right) ^{2}} +\tfrac{1}{n}\sum_{t=1}^{n}\tfrac{2\left( Y_{t}-\mathcal{Y}_{t-1}^{\prime
}\theta _{u}\right)
}{\left( 
\mathcal{Z}_{t-1}^{\prime }\Lambda
_{\ast }\right) ^{2}}\mathcal{Y}_{t-1}^{\prime }\left( \widehat{\theta }%
_{1n}-\theta _{0}\right) .  \label{A.8}
\end{eqnarray}
Therefore, the result (\ref{A.5}) follows from the a.s. boundedness of 
\begin{equation*}
\tfrac{2\left(
Y_{t}-\mathcal{Y}_{t-1}^{\prime }\theta _{u}\right) 
}{\left( 
%
\mathcal{Z}_{t-1}^{\prime }\Lambda _{\ast }\right) ^{2}}\mathcal{Y}%
_{t-1}^{\prime },
\end{equation*}
which entails the finiteness of 
$\mathbb{E}\left( \tfrac{2\left(
Y_{t}-\mathcal{Y}_{t-1}^{\prime }\theta _{u}\right) 
}{\left( 
\mathcal{Z}_{t-1}^{\prime }\Lambda _{\ast }\right) ^{2}}\mathcal{Y}%
_{t-1}^{\prime }\right)
$, the ergodic theorem and the strong consistency of $\widehat{\theta }_{1n}$.

ii) One can write $\widehat{\theta }_{1n}-\theta _{0}=\tfrac{1%
}{\sqrt{n}}O_{p}\left( 1\right) $, where $O_{p}\left( 1\right) $ denotes a
term bounded in probability. Hence (\ref{A.8}) becomes
\begin{eqnarray*}
&&\left. \tfrac{1}{\sqrt{n}}\sum_{t=1}^{n}\mathcal{Z}_{t-1}\tfrac{\left(
Y_{t}-\mathcal{Y}_{t-1}^{\prime }\widehat{\theta }_{1n}\right) ^{2}}{\left(
\mathcal{Z}_{t-1}^{\prime }\Lambda _{\ast }\right) ^{2}}=\tfrac{1}{\sqrt{n}}%
\sum_{t=1}^{n}\tfrac{\left( Y_{t}-\mathcal{Y}_{t-1}^{\prime }\theta
_{0}\right) ^{2}}{\left( 
\mathcal{Z}_{t-1}^{\prime }\Lambda _{\ast }\right) ^{2}}%
\right.+\tfrac{1}{n}\sum_{t=1}^{n}\tfrac{%
2\left( Y_{t}-\mathcal{Y}_{t-1}^{\prime }\theta _{u}\right) 
}{\left( 
\mathcal{Z}_{t-1}^{\prime }\Lambda _{\ast }\right) ^{2}}%
\mathcal{Y}_{t-1}^{\prime }O_{p}\left( 1\right),
\end{eqnarray*}
so the result (\ref{A.6}) follows from the ergodic theorem, the consistency of $\widehat{%
\theta }_{1n}$ and the fact that $\theta _{u}$ is between $\widehat{\theta }%
_{1n}$ and $\theta _{0}$.
\end{proof}

\begin{lemma}
Consider the RMINAR model \eqref{3.1a} under \textbf{A0}. Then,
\begin{eqnarray}
\tfrac{1}{n}\sum_{t=1}^{n}\mathcal{Y}_{t-1}Y_{t}\left( \tfrac{1}{\left(
\mathcal{Z}_{t-1}^{\prime }\widehat{\Lambda }_{1n}\right) ^{2}}-\tfrac{1}{%
\left( 
\mathcal{Z}_{t-1}^{\prime }\Lambda _{0}\right) ^{2}}\right) 
&=&o_{a.s.}\left( 1\right),   \label{A.9} \\
\tfrac{1}{\sqrt{n}}\sum_{t=1}^{n}\mathcal{Y}_{t-1}Y_{t}\left( \tfrac{1}{%
\left( 
\mathcal{Z}_{t-1}^{\prime }\widehat{\Lambda }_{1n}\right) ^{2}}-\tfrac{1%
}{\left( 
\mathcal{Z}_{t-1}^{\prime }\Lambda _{0}\right) ^{2}}\right)  &=&o_{p}\left(
1\right),   \label{A.10} \\
\tfrac{1}{n}\sum_{t=1}^{n}\left( \tfrac{1}{\left( 
\mathcal{Z}_{t-1}^{\prime
}\widehat{\Lambda }_{1n}\right) ^{2}}-\tfrac{1}{\left( 
\mathcal{Z}_{t-1}^{\prime }\Lambda
_{0}\right) ^{2}}\right) \mathcal{Y}_{t-1}\mathcal{Y}_{t-1}^{\prime }
&=&o_{a.s.}\left( 1\right).   \label{A.11}
\end{eqnarray}
\label{a.4}
\end{lemma}
\begin{proof}{\textbf{(Lemma \ref{a.4})}}\\
Similarly to the proof of condition (\ref{A.5}), a Taylor expansion of the function 
$g\left( \widehat{\theta }_{1n}\right) =\tfrac{1}{\left( 
\mathcal{Z}%
_{t-1}^{\prime }\widehat{\Lambda }_{1n}\right) ^{2}}$
around $\theta _{0}$\
shows that $I=o_{a.s.}\left( 1\right) $. 
The proofs of (\ref{A.10}) and (\ref{A.11})
follow in the same way as for the result (\ref{A.9}). 
\end{proof}
\begin{lemma}
Consider the RMINAR model \eqref{3.1a} under \textbf{A0}. Then,
\begin{eqnarray}
\tfrac{1}{n}\sum_{t=1}^{n}\mathcal{Z}_{t-1}\left( \tfrac{\left( Y_{t}-%
\mathcal{Y}_{t-1}^{\prime }\widehat{\theta }_{2n}\right) ^{2}}{\left( 
\mathcal{%
Z}_{t-1}^{\prime }\widehat{\Lambda }_{1n}\right) ^{2}}-\tfrac{\left( Y_{t}-%
\mathcal{Y}_{t-1}^{\prime }\theta _{0}\right) ^{2}}{\left( 
\mathcal{Z}_{t-1}^{\prime
}\Lambda _{0}\right) ^{2}}\right)  &=&o_{a.s.}\left( 1\right),   \label{A.12}
\\
\tfrac{1}{\sqrt{n}}\sum_{t=1}^{n}\mathcal{Z}_{t-1}\left( \tfrac{\left( Y_{t}-%
\mathcal{Y}_{t-1}^{\prime }\widehat{\theta }_{2n}\right) ^{2}}{\left( 
\mathcal{%
Z}_{t-1}^{\prime }\widehat{\Lambda }_{1n}\right) ^{2}}-\tfrac{\left( Y_{t}-%
\mathcal{Y}_{t-1}^{\prime }\theta _{0}\right) ^{2}}{\left( 
\mathcal{Z}_{t-1}^{\prime
}\Lambda _{0}\right) ^{2}}\right)  &=&o_{p}\left( 1\right),   \label{A.13} \\
\tfrac{1}{n}\sum_{t=1}^{n}\left( \tfrac{1}{\left( 
\mathcal{Z}_{t-1}^{\prime
}\widehat{\Lambda }_{1n}\right) ^{2}}-\tfrac{1}{\left( 
\mathcal{Z}_{t-1}^{\prime }\Lambda
_{0}\right) ^{2}}\right) \mathcal{Z}_{t-1}\mathcal{Z}_{t-1}^{\prime }
&=&o_{a.s.}\left( 1\right).   \label{A.14}
\end{eqnarray}
\label{a.5}
\end{lemma}
\begin{proof}{\textbf{(Lemma \ref{a.5})}}\\
The proof of condition (\ref{A.12}) uses the same device as that used in the proof of (\ref{A.9}),
giving 
\begin{eqnarray*}
&&\tfrac{1}{n}\sum_{t=1}^{n}\mathcal{Z}_{t-1}\left( \tfrac{\left( Y_{t}-%
\mathcal{Y}_{t-1}^{\prime }\widehat{\theta }_{2n}\right) ^{2}}{\left(
\mathcal{%
Z}_{t-1}^{\prime }\widehat{\Lambda }_{1n}\right) ^{2}}-\tfrac{\left( Y_{t}-%
\mathcal{Y}_{t-1}^{\prime }\theta _{0}\right) ^{2}}{\left( 
\mathcal{Z}_{t-1}^{\prime
}\Lambda _{0}\right) ^{2}}\right)
=o_{a.s.}\left( 1\right) .
\end{eqnarray*}
Results (\ref{A.13}) and (\ref{A.14}) follow using similar arguments.
\end{proof}
\begin{lemma}
Consider the RMINAR model \eqref{3.1a} under \textbf{A0}, $\mathbb{E}\left( \varepsilon _{t}^{4}\right) <\infty$ and $\mathbb{E}\left( \Phi
_{it}^{4}\right) <1$\textit{\ (}$i=1,\dots,p$)\textit{\ then, for all }$%
\Lambda _{\ast }>\mathbf{0}_{p+q+1}$,
\begin{eqnarray}
&&\tfrac{1}{\sqrt{n}}\sum_{t=1}^{n}\tfrac{\mathcal{Y}_{t-1}e_{t}}{\mathcal{Z}%
_{t-1}^{\prime }\Lambda _{\ast }}\overset{D}{\underset{n\rightarrow \infty }{%
\rightarrow }}\mathcal{N}\left( 0,B\left( \Lambda _{0},\Lambda _{\ast
}\right) \right),  \label{A.15} \\
&&\tfrac{1}{\sqrt{n}}\sum_{t=1}^{n}\tfrac{\mathcal{Z}_{t-1}u_{t}}{\left(
\mathcal{Z}%
_{t-1}^{\prime }\Lambda _{\ast }\right) ^{2}}\overset{D}{\underset{%
n\rightarrow \infty }{\rightarrow }}\mathcal{N}\left( 0,D\left( \Lambda
_{\ast }\right) \right).   \label{A.16}
\end{eqnarray}
\label{a.6}
\end{lemma}
\begin{proof}{\textbf{(Lemma \ref{a.6})}}

i) By \textbf{A0} and the ergodic theorem, it follows that 
\begin{eqnarray*}
\sum_{t=1}^{n}\left( \tfrac{1}{\sqrt{n}}\tfrac{\mathcal{Y}_{t-1}e_{t}}{%
\mathcal{Z}_{t-1}^{\prime }\Lambda _{\ast }}\right) \left( \tfrac{1}{\sqrt{n}%
}\tfrac{\mathcal{Y}_{t-1}e_{t}}{\mathcal{Z}_{t-1}^{\prime }\Lambda _{\ast }}%
\right) ^{\prime } &=&\tfrac{1}{n}\sum_{t=1}^{n}\tfrac{e_{t}^{2}\mathcal{Y}_{t-1}%
\mathcal{Y}_{t-1}^{\prime }}{\left( 
\mathcal{Z}_{t-1}^{\prime }\Lambda _{\ast }\right) ^{2}} \overset{a.s.}{\underset{n\rightarrow \infty }{\rightarrow }}B\left(
\Lambda _{0},\Lambda _{\ast }\right),
\end{eqnarray*}%
where $B\left( \Lambda _{0},\Lambda _{\ast }\right) $ is finite by the a.s.~boundedness of $\tfrac{\mathcal{Z}_{t-1}^{\prime }\Lambda _{0}\mathcal{Y}_{t-1}%
\mathcal{Y}_{t-1}^{\prime }}{\left( 
\mathcal{Z}_{t-1}^{\prime }\Lambda _{\ast }\right) ^{2}}$ and the finiteness of the second moments of the
random coefficients. Since $\left\{ e_{t},t\in 
\mathbb{Z}
\right\} $ is square integrable martingale difference with respect to $%
\left\{ \tciFourier _{t}^{Y},t\in 
\mathbb{Z}
\right\} $, then (\ref{A.15}) follows from the central limit theorem (CLT) for
square-integrable martingales (e.g.~Billingsley, $2008$; Hall and Heyde, $%
1980$).

ii) Note that 
\begin{eqnarray}
\sum_{t=1}^{n}\left( \tfrac{1}{\sqrt{n}}\tfrac{\mathcal{Z}_{t-1}u_{t}}{%
\left( 
\mathcal{Z}_{t-1}^{\prime }\Lambda _{\ast }\right) ^{2}}\right) \left( 
\tfrac{1}{\sqrt{n}}\tfrac{\mathcal{Z}_{t-1}u_{t}}{\left( 
\mathcal{Z}_{t-1}^{\prime
}\Lambda _{\ast }\right) ^{2}}\right) ^{\prime } &=&\tfrac{1}{n}%
\sum_{t=1}^{n}\tfrac{u_{t}^{2}\mathcal{Z}_{t-1}\mathcal{Z}_{t-1}^{\prime }}{%
\left( 
\mathcal{Z}_{t-1}^{\prime }\Lambda _{\ast }\right) ^{4}}  
\overset{a.s.}{\underset{n\rightarrow \infty }{\rightarrow }}D\left(
\Lambda _{\ast }\right) .  \notag
\end{eqnarray}%
Hence (\ref{A.16}) is a consequence of the CLT for the square-integrable $%
\left\{ \tciFourier _{t}^{Y},t\in 
\mathbb{Z}
\right\} $-martingale difference $\left\{ u_{t},t\in 
\mathbb{Z}
\right\} $.
\end{proof}

Next, follows the proof of Theorem \ref{teo41} based on the established lemmas \ref{a.2} to \ref{a.6}.

\begin{proof}{\textbf{(Theorem \ref{teo41})}}\\
a) For the proof of (\ref{4.7})-i, combining (\ref{3.1a}) and  (\ref{4.3a}) leads to
\begin{equation}
\widehat{\theta }_{1n}-\theta _{0}=\left( \tfrac{1}{n}\sum_{t=1}^{n}\tfrac{\mathcal{Y}_{t-1}\mathcal{Y}%
_{t-1}^{\prime }
}{\mathcal{Z}_{t-1}^{\prime }\Lambda _{\ast }}\right) ^{-1}\tfrac{1}{n}\sum_{t=1}^{n}\mathcal{Y}_{t-1}%
\tfrac{e_{t}}{\mathcal{Z}_{t-1}^{\prime }\Lambda _{\ast }},  \label{A.17}
\end{equation}
so the result (\ref{4.7})-i follows from Lemma \ref{a.2} while using the fact that 
\begin{equation*}
\mathbb{E}\left( \mathcal{Y}_{t-1}\tfrac{e_{t}}{\mathcal{Z}_{t-1}^{\prime }\Lambda
_{\ast }}\right) = \mathbb{E}\left( \mathcal{Y}_{t-1}\tfrac{\mathbb{E}\left( e_{t}|\tciFourier
_{t-1}^{Y}\right) }{\mathcal{Z}_{t-1}^{\prime }\Lambda _{\ast }}\right) =0.
\end{equation*}
b) For the proof of (\ref{4.7})-ii, in view of (\ref{4.3b}), the regression%
\begin{equation}
\left( Y_{t}-\mathcal{Y}_{t-1}^{\prime }\theta _{0}\right) ^{2}=\mathcal{Z}%
_{t-1}^{\prime }\Lambda _{0}+u_{t},  \label{A.18}
\end{equation}%
and Lemma \ref{a.3}, it follows that 
\begin{eqnarray*}
\widehat{\Lambda }_{1n}-\Lambda _{0}=\left( \tfrac{1}{n}\sum_{t=1}^{n}\tfrac{%
\mathcal{Z}_{t-1}%
\mathcal{Z}_{t-1}^{\prime }}{\left( 
\mathcal{Z}_{t-1}^{\prime }\Lambda _{\ast }\right) ^{2}}\right) ^{-1}\tfrac{1}{n}\sum_{t=1}^{n}\mathcal{Z}%
_{t-1}\tfrac{u_{t}}{\left( 
\mathcal{Z}_{t-1}^{\prime }\Lambda _{\ast }\right) ^{2}} +o_{a.s.}\left( 1\right).
\end{eqnarray*}
Therefore, the result (\ref{4.7})-ii follows from Lemma \ref{a.2} and the fact that 
\begin{eqnarray*}
\mathbb{E}\left( \tfrac{u_{t}}{\left( 
\mathcal{Z}_{t-1}^{\prime }\Lambda _{\ast }\right) ^{2}}%
\right) =\mathbb{E}\left( \tfrac{\mathbb{E}\left( u_{t}|\tciFourier _{t-1}^{Y}\right) }{\left(
\mathcal{Z}%
_{t-1}^{\prime }\Lambda _{\ast }\right) ^{2}}\right) =0.
\end{eqnarray*}
c) The proof of (\ref{4.7})-iii is based on (\ref{4.3c}), (\ref{3.1a}) and Lemma \ref{a.4}, leading to 
\begin{equation*}
\widehat{\theta }_{2n}-\theta _{0}=\left( \tfrac{1}{n}\sum_{t=1}^{n}\tfrac{\mathcal{Y}_{t-1}\mathcal{Y}%
_{t-1}^{\prime }
}{\mathcal{Z}_{t-1}^{\prime }\Lambda _{0}}\right) ^{-1}\tfrac{1}{n}\sum_{t=1}^{n}\mathcal{Y}_{t-1}%
\tfrac{e_{t}}{\mathcal{Z}_{t-1}^{\prime }\Lambda _{0}}+o_{a.s.}\left(
1\right) ,
\end{equation*}%
so the result follows from Lemma \ref{a.2} in the same way as (\ref{4.7})-i.

d) The proof of (\ref{4.7})-iv follows by combining (\ref{4.3c}), (\ref{A.18}) and Lemma \ref{a.5} to obtain%
\begin{equation*}
\widehat{\Lambda }_{2n}-\Lambda _{0}=\left( \tfrac{1}{n}\sum_{t=1}^{n}\tfrac{%
\mathcal{Z}_{t-1}%
\mathcal{Z}_{t-1}^{\prime }}{\left( 
\mathcal{Z}_{t-1}^{\prime }\Lambda _{0}\right) ^{2}}\right) ^{-1}\tfrac{1}{n}\sum_{t=1}^{n}\tfrac{\mathcal{Z}%
_{t-1}u_{t}}{\left( 
\mathcal{Z}_{t-1}^{\prime }\Lambda _{0}\right) ^{2}}
+o_{a.s.}\left( 1\right) 
\end{equation*}%
and the result is a consequence of Lemma  \ref{a.2} using the same argument in
proving (\ref{4.7})-ii.

e) For the proof of (\ref{4.8})-i, condition (\ref{A.17}) is rewritten as 
\begin{equation*}
\sqrt{n}\left( \widehat{\theta }_{1n}-\theta _{0}\right) =\left( \tfrac{1}{n}%
\sum_{t=1}^{n}\tfrac{\mathcal{%
Y}_{t-1}\mathcal{Y}_{t-1}^{\prime }}{\mathcal{Z}_{t-1}^{\prime }\Lambda _{\ast }}\right) ^{-1}\tfrac{1}{\sqrt{n}}%
\sum_{t=1}^{n}\tfrac{\mathcal{Y}_{t-1}e_{t}}{\mathcal{Z}_{t-1}^{\prime
}\Lambda _{\ast }},
\end{equation*}%
and the result follows from Lemma \ref{a.2} and Lemma \ref{a.6} while using Slutsky's Lemma. 

f) For the proof of (\ref{4.8})-ii, from $(\ref{4.3c}$) and Lemma \ref{a.4}, one obtains
\begin{equation*}
\sqrt{n}\left( \widehat{\theta }_{2n}-\theta _{0}\right) =\left( \tfrac{1}{n}%
\sum_{t=1}^{n}\tfrac{\mathcal{Y}%
_{t-1}\mathcal{Y}_{t-1}^{\prime }}{\mathcal{Z}_{t-1}^{\prime }\Lambda _{0}}\right) ^{-1}\tfrac{1}{\sqrt{n}}%
\sum_{t=1}^{n}\tfrac{\mathcal{Y}_{t-1}e_{t}}{\mathcal{Z}_{t-1}^{\prime
}\Lambda _{0}}+o_{p}\left( 1\right), 
\end{equation*}%
and the result follows from Lemma \ref{a.6} and Lemma \ref{a.2}.

g) For the proof of (\ref{4.8})-iii, consider the Lemma \ref{a.3}. One can write
\begin{equation*}
\sqrt{n}\left( \widehat{\Lambda }_{1n}-\Lambda _{0}\right) =\left( \tfrac{1}{%
n}\sum_{t=1}^{n}\tfrac{\mathcal{Z}_{t-1}\mathcal{Z}_{t-1}^{\prime }}{\left(
\mathcal{Z}%
_{t-1}^{\prime }\Lambda _{\ast }\right) ^{2}}\right) ^{-1}\tfrac{1}{\sqrt{n}}%
\sum_{t=1}^{n}\tfrac{\mathcal{Z}_{t-1}u_{t}}{\left( 
\mathcal{Z}_{t-1}^{\prime }\Lambda
_{\ast }\right) ^{2}}+o_{p}\left( 1\right),
\end{equation*}%
with obvious notation. Hence, the result (\ref{4.8})-ii follows from Lemma \ref{a.6} and Lemma \ref{a.2}.

h) The proof of (\ref{4.8})-iv relies on (\ref{4.3d}) and Lemma \ref{a.5}. Then, 
\begin{equation*}
\sqrt{n}\left( \widehat{\Lambda }_{2n}-\Lambda _{0}\right) =\left( \tfrac{1}{%
n}\sum_{t=1}^{n}\tfrac{\mathcal{Z}_{t-1}\mathcal{Z}_{t-1}^{\prime }}{\left(
\mathcal{Z}%
_{t-1}^{\prime }\Lambda _{0}\right) ^{2}}\right) ^{-1}\tfrac{1}{\sqrt{n}}%
\sum_{t=1}^{n}\tfrac{\mathcal{Z}_{t-1}u_{t}}{\left( 
\mathcal{Z}_{t-1}^{\prime }\Lambda
_{0}\right) ^{2}}+o_{p}\left( 1\right) 
\end{equation*}%
and the result follows from Lemma \ref{a.6} (or \ref{A.16} with $\Lambda _{\ast
}=\Lambda _{0}$) and Lemma \ref{a.2}. This concludes the proof of the theorem.
\end{proof}\\


The proof of the results conveyed in Theorem \ref{teo43} is based on the intermediary results stated in the following lemmas. Consider the following notation
\begin{eqnarray*}
L_{n}\left( \Lambda ,\Lambda _{\ast },\widehat{\theta }_{1n}\right)  &=&%
\tfrac{1}{n}\sum_{t=1}^{n}l_{t}\left( \Lambda ,\Lambda _{\ast },\widehat{%
\theta }_{1n}\right)  \\
l_{t}\left( \Lambda ,\Lambda _{\ast },\widehat{\theta }_{1n}\right)  &=&%
\tfrac{u_{t}^{2}\left( \widehat{\theta }_{1n},\Lambda \right) }{%
V_{t}^{2}\left( \widehat{\theta }_{1n},\Lambda _{\ast }\right) }=\tfrac{%
\left( \left( Y_{t}-\mathcal{Y}_{t-1}^{\prime }\widehat{\theta }_{1n}\right)
^{2}-V_{t}\left( \widehat{\theta }_{1n},\Lambda \right) \right) ^{2}}{
V_{t}^{2} \left( \widehat{\theta }_{1n},\Lambda _{\ast }\right) }
\end{eqnarray*}%
where $V_{t}\left( \theta _{0},\Lambda _{0}\right) :=\mathbb{V}\left(
Y_{t}|\tciFourier _{t-1}^{Y}\right) $. 

\begin{lemma}
Under \textbf{A0**}, it holds that
\begin{equation*}
\sup\limits_{\Lambda \in \Pi }\left\vert L_{n}\left( \Lambda ,\Lambda _{\ast
},\widehat{\theta }_{1n}\right) -L_{n}\left( \Lambda ,\Lambda _{\ast
},\theta _{0}\right) \right\vert \overset{a.s.}{\underset{n\rightarrow
\infty }{\rightarrow }}0.
\end{equation*}
\label{a.7}
\end{lemma}

\begin{proof}{\textbf{ (Lemma \ref{a.7})}}\\
It holds that%
\begin{eqnarray*}
&&\left\vert L_{n}\left( \Lambda ,\Lambda _{\ast },\widehat{\theta }%
_{1n}\right) -L_{n}\left( \Lambda ,\Lambda _{\ast },\theta _{0}\right)
\right\vert  \\
&\leq &\tfrac{1}{n}\sum_{t=1}^{n}\left\vert \tfrac{\left( \left( Y_{t}-%
\mathcal{Y}_{t-1}^{\prime }\widehat{\theta }_{1n}\right) ^{2}-V_{t}\left( 
\widehat{\theta }_{1n},\Lambda \right) \right) ^{2}}{\left( 
V_{t}\left( 
\widehat{\theta }_{1n},\Lambda _{\ast }\right) \right) ^{2}}-\tfrac{\left(
\left( Y_{t}-\mathcal{Y}_{t-1}^{\prime }\theta _{0}\right) ^{2}-V_{t}\left(
\theta _{0},\Lambda \right) \right) ^{2}}{\left( 
V_{t}\left( \theta _{0},\Lambda
_{\ast }\right) \right) ^{2}}\right\vert .
\end{eqnarray*}%
Using the Taylor expansion of the function 
$
g_{1}\left( \widehat{\theta }%
_{1n}\right) =\tfrac{\left( \left( Y_{t}-\mathcal{Y}_{t-1}^{\prime }\widehat{%
\theta }_{1n}\right) ^{2}-V_{t}\left( \widehat{\theta }_{1n},\Lambda \right)
\right) ^{2}}{\left( 
V_{t}\left( \widehat{\theta }_{1n},\Lambda _{\ast
}\right) \right) ^{2}}    
$
around $\theta _{0}$\, as in Lemma \ref{a.3}, the result follows from the a.s. boundedness of $\tfrac{\partial }{\partial
\theta ^{\prime }}g_{1}\left( \theta _{u}\right) $ and the strong
consistency of $\widehat{\theta }_{1n}$.
\end{proof}

\begin{lemma}
Under \textbf{A0**}, it holds that
\begin{eqnarray*}
&&\text{i) }\mathbb{E}\left( l_{t}\left( \Lambda ,\Lambda _{\ast },\theta _{0}\right)
\right) <\infty,\\
&&\text{ii) } \mathbb{E}\left( l_{t}\left( \Lambda ,\Lambda _{\ast
},\theta _{0}\right) \right) \leq \mathbb{E}\left( l_{t}\left( \Lambda _{0},\Lambda
_{\ast },\theta _{0}\right) \right) \text{ for all }\Lambda \in \Pi, \text{ }
\\
&&\text{iii) }\mathbb{E}\left(l_{t}\left( \Lambda ,\Lambda _{\ast },\theta _{0}\right) \right)
=\mathbb{E}\left(l_{t}\left( \Lambda _{0},\Lambda _{\ast },\theta _{0}\right)\right)
\Longrightarrow \Lambda =\Lambda _{0}.
\end{eqnarray*}
\label{a.8}
\end{lemma}

\begin{proof}{\textbf{(Lemma \ref{a.8})}}

i) The result obviously follows from the boundedness of $%
\tfrac{u_{t}^{2}\left( \widehat{\theta }_{1n},\Lambda \right) }{%
V_{t}^{2}\left( \widehat{\theta }_{1n},\Lambda _{\ast }\right) }$. 

ii) Standard arguments show that for all $\Lambda \in \Pi $%
\begin{eqnarray}
\mathbb{E}\left( l_{t}\left( \Lambda ,\Lambda _{\ast },\theta _{0}\right) \right) 
&=&\mathbb{E}\left( \tfrac{\left( e_{t}^{2}-V_{t}\left( \theta _{0},\Lambda
_{0}\right) +V_{t}\left( \theta _{0},\Lambda _{0}\right) -V_{t}\left( \theta
_{0},\Lambda \right) \right) ^{2}}{V_{t}^{2}\left( \widehat{\theta }%
_{1n},\Lambda _{\ast }\right) }\right)   \notag \\
&=&\mathbb{E}\left( \tfrac{\left( e_{t}^{2}-V_{t}\left( \theta _{0},\Lambda
_{0}\right) \right) ^{2}}{V_{t}^{2}\left( \widehat{\theta }_{1n},\Lambda
_{\ast }\right) }\right) +\mathbb{E}\left( \tfrac{\left( V_{t}\left( \theta
_{0},\Lambda _{0}\right) -V_{t}\left( \theta _{0},\Lambda \right) \right)
^{2}}{V_{t}^{2}\left( \widehat{\theta }_{1n},\Lambda _{\ast }\right) }%
\right)   \notag \\
&\geq &\mathbb{E}\left( \tfrac{\left( e_{t}^{2}-V_{t}\left( \theta _{0},\Lambda
_{0}\right) \right) ^{2}}{V_{t}^{2}\left( \widehat{\theta }_{1n},\Lambda
_{\ast }\right) }\right) =\mathbb{E}\left( l_{t}\left( \Lambda _{0},\Lambda _{\ast
},\theta _{0}\right) \right) .  \label{A.19}
\end{eqnarray}

iii) Inequality (\ref{A.19}) becomes an equality if and only if%
\begin{equation*}
\mathbb{E}\left( \tfrac{\left( V_{t}\left( \theta _{0},\Lambda _{0}\right)
-V_{t}\left( \theta _{0},\Lambda \right) \right) ^{2}}{V_{t}^{2}\left( 
\widehat{\theta }_{1n},\Lambda _{\ast }\right) }\right) =0,
\end{equation*}%
which holds if and only if\\
$
\left. \tfrac{\left( V_{t}\left( \theta _{0},\Lambda _{0}\right)
-V_{t}\left( \theta _{0},\Lambda \right) \right) ^{2}}{u_{t}^{2}\left( 
\widehat{\theta }_{1n},\Lambda _{\ast }\right) }=0\Leftrightarrow
V_{t}\left( \theta _{0},\Lambda \right) =V_{t}\left( \theta _{0},\Lambda
_{0}\right) \right.  \\
\Leftrightarrow \sigma _{\varepsilon }^{2}\left( \left( \mathcal{Y}
_{t-1}^{\prime }\theta _{0}\right) ^{2}+\delta _{t}^{2}\left( \Delta \right)
\right) +\delta _{t}^{2}\left( \Delta \right) =\sigma _{0\varepsilon
}^{2}\left( \left( \mathcal{Y}_{t-1}^{\prime }\theta _{0}\right) ^{2}+\delta
_{t}^{2}\left( \Delta _{0}\right) \right) +\delta _{t}^{2}\left( \Delta
_{0}\right)  \\
\Leftrightarrow \sigma _{\varepsilon }^{2}\left( \left( \mathcal{Y}%
_{t-1}^{\prime }\theta _{0}\right) ^{2}+\mathcal{Z}_{t-1}^{\prime }\Delta
\right) +\mathcal{Z}_{t-1}^{\prime }\Delta =\sigma _{0\varepsilon
}^{2}\left( \left( \mathcal{Y}_{t-1}^{\prime }\theta _{0}\right) ^{2}+%
\mathcal{Z}_{t-1}^{\prime }\Delta _{0}\right) +\mathcal{Z}_{t-1}^{\prime
}\Delta _{0} \\
\Leftrightarrow \sigma _{\varepsilon }^{2}\left( \left( \mathcal{Y}
_{t-1}^{\prime }\theta _{0}\right) ^{2}+\mathcal{Z}_{t-1}^{\prime }\Delta
\right) =\sigma _{0\varepsilon }^{2}\left( \left( \mathcal{Y}_{t-1}^{\prime
}\theta _{0}\right) ^{2}+\mathcal{Z}_{t-1}^{\prime }\Delta _{0}\right) \text{
and }\mathcal{Z}_{t-1}^{\prime }\Delta =\mathcal{Z}_{t-1}^{\prime }\Delta
_{0} \\
\Leftrightarrow \sigma _{\varepsilon }^{2}=\sigma _{0\varepsilon }^{2}%
\text{ and }\Delta =\Delta _{0}\Leftrightarrow \Lambda =\Lambda _{0}.
$
\end{proof}

\begin{lemma}
Under \textbf{A0** }for all $\Lambda \neq
\Lambda _{0}$, there is a neighborhood $\mathcal{V}\left( \Lambda
\right)$ such that%
\begin{equation*}
\underset{n\rightarrow \infty }{\lim \inf }\inf_{\Lambda _{1}\in \mathcal{V}%
\left( \Lambda \right) }L_{n}\left( \Lambda _{1},\Lambda _{\ast },\widehat{%
\theta }_{1n}\right) >\underset{n\rightarrow \infty }{\text{ }\lim \inf }%
L_{n}\left( \Lambda _{0},\Lambda _{\ast },\widehat{\theta }_{1n}\right) 
a.s.
\end{equation*}
\label{a.9}
\end{lemma}
\begin{proof}\textbf{(Lemma \ref{a.9})}\\
For all $\overline{\Lambda }\in \Pi $ and $k\in 
\mathbb{N}
^{\ast }$, let $\mathcal{V}_{k}(\overline{\Lambda })$ be the open ball with
center $\overline{\Lambda }$ and radius $1/k$. Since $\inf_{\Lambda \in 
\mathcal{V}_{k}(\overline{\Lambda })\cap \Pi }l_{t}\left( \Lambda ,\Lambda
_{\ast },\theta _{0}\right) $ is a measurable function of the terms of $%
\left\{ Y_{t},t\in 
\mathbb{Z}
\right\} $, which is strictly stationary and ergodic under \textbf{A0}, then 
$\left\{ \inf_{\Lambda \in \mathcal{V}_{k}(\overline{\Lambda })\cap \Pi
}l_{t}\left( \Lambda ,\Lambda _{\ast },\theta _{0}\right) ,\text{ }t\in 
\mathbb{Z}
\right\} $ is also strictly stationary and ergodic where, by Lemma \ref{a.8},
\begin{equation*}
\mathbb{E}\left( \inf_{\Lambda \in \mathcal{V}_{k}(\overline{\Lambda })\cap \Pi
}l_{t}\left( \Lambda ,\Lambda _{\ast },\theta _{0}\right) \right) \in
\lbrack -\infty ,+\infty \lbrack .
\end{equation*}%
Therefore, in view of Lemma \ref{a.7} and the ergodic theorem
(Billingsley, $2008$), it follows that%
\begin{eqnarray*}
\underset{n\rightarrow \infty }{\lim \inf }\inf_{\Lambda \in \mathcal{V}_{k}(%
\overline{\Lambda })\cap \Pi }L_{n}\left( \Lambda ,\Lambda _{\ast },\widehat{%
\theta }_{1n}\right)  &=&\text{ }\underset{n\rightarrow \infty }{\lim \inf }%
\inf_{\Lambda \in \mathcal{V}_{k}(\overline{\Lambda })\cap \Pi }L_{n}\left(
\Lambda _{0},\Lambda _{\ast },\theta _{0}\right)  \\
&\geq &\mathbb{E}\left( \inf_{\Lambda \in \mathcal{V}_{k}(\overline{\Lambda })\cap
\Pi }l_{t}\left( \Lambda ,\Lambda _{\ast },\theta _{0}\right) \right) .
\end{eqnarray*}%
By the Beppo-Levi theorem, $\mathbb{E}\left( \inf_{\Lambda \in \mathcal{V}_{k}(%
\overline{\Lambda })\cap \Pi }l_{t}\left( \Lambda ,\Lambda _{\ast },\theta
_{0}\right) \right) $ converges while increasing to $\mathbb{E}\left( l_{t}\left( 
\overline{\Lambda },\Lambda _{\ast },\theta _{0}\right) \right) $ as $%
k\rightarrow \infty $. Hence, Lemma \ref{a.9}  follows from Lemma \ref{a.8} ii). 
\end{proof}\\

Next, it follows the proof of the Theorem \ref{teo43}, based on the previous lemmas.

\begin{proof}{\textbf{(Theorem \ref{teo43})}}

a) The proof of (\ref{4.12b})-i  is based on the results in Lemma (\ref{a.7}) and (\ref{a.9}). It is shown that for all $%
\overline{\Lambda }\neq \Lambda _{0}$, there exists a neighborhood $\mathcal{%
V}\left( \overline{\Lambda }\right) $ such that 
\begin{eqnarray*}
\underset{n\rightarrow \infty }{\lim \inf }\inf_{\Lambda \in \mathcal{V}_{k}(%
\overline{\Lambda })\cap \Pi }L_{n}\left( \Lambda ,\Lambda _{\ast },\widehat{%
\theta }_{1n}\right)  &\geq &\text{ }\underset{n\rightarrow \infty }{\lim
\inf }L_{n}\left( \Lambda _{0},\Lambda _{\ast },\widehat{\theta }%
_{1n}\right)  \\
&=&\underset{n\rightarrow \infty }{\text{ }\lim \inf }L_{n}\left( \Lambda
_{0},\Lambda _{\ast },\theta _{0}\right) =\mathbb{E}\left( l_{t}\left( \Lambda
_{0},\Lambda _{\ast },\theta _{0}\right) \right),
\end{eqnarray*}%
which concludes the proof.

b) The proof of (\ref{4.12b})-ii is similar to that of (\ref{4.12b})-i by replacing $\widehat{\theta }_{1n}$ by $\widehat{\theta }_{2n}$. Similarly, the proof of (\ref{4.12b})-iii and (\ref{4.12b})-iv is similar to that
of (\ref{4.12b})-i by replacing $\Lambda _{\ast }$ by $\widehat{\Lambda }_{1n}$ and $\Lambda _{\ast }$ by $\widehat{\Lambda }_{2n}$, respectively.

c) The proof of the results in (\ref{4.13}) is similar to all estimators and thus only the consistency for $\widehat{\Lambda }_{1n}$ is shown.

d) From standard arguments, the proof of the consistency of $\widehat{\Lambda }_{1n}$ is completed while using the compactness assumption \textbf{A1} of $\Pi $. Using again a Taylor expansion of 
\begin{eqnarray*}
\tfrac{\left(
\left( Y_{t}-\mathcal{Y}_{t-1}^{\prime }\widehat{\theta }_{1n}\right)
^{2}-V_{t}\left( \widehat{\theta }_{1n},\Lambda \right) \right) ^{2}}{\left(
V_{t}\left( \widehat{\theta }_{1n},\Lambda _{\ast }\right) \right) ^{2}},
\end{eqnarray*}
around $\theta _{0}$, it holds that
 \begin{eqnarray*}
&&\sqrt{n}\sup_{\theta \in \Theta }\left\Vert L_{n}\left( \Lambda ,\Lambda
_{\ast },\widehat{\theta }_{1n}\right) -L_{n}\left( \Lambda ,\Lambda _{\ast
},\theta _{0}\right) \right\Vert  \\
&\leq &\tfrac{1}{\sqrt{n}}\sum_{t=1}^{n}\left\vert \tfrac{\left( \left(
Y_{t}-\mathcal{Y}_{t-1}^{\prime }\widehat{\theta }_{1n}\right)
^{2}-V_{t}\left( \widehat{\theta }_{1n},\Lambda \right) \right) ^{2}}{\left(
V_{t}\left( \widehat{\theta }_{1n},\Lambda _{\ast }\right) \right) ^{2}}%
-\tfrac{\left( \left( Y_{t}-\mathcal{Y}_{t-1}^{\prime }\theta _{0}\right)
^{2}-V_{t}\left( \theta _{0},\Lambda \right) \right) ^{2}}{\left(
V_{t}\left( \theta
_{0},\Lambda _{\ast }\right) \right) ^{2}}\right\vert = o_{a.s.}\left( 1\right) \text{.}
\end{eqnarray*}

Now from \textbf{A2} and (\ref{4.12b})-i, the estimate $\widehat{\Lambda }_{1n}$
cannot be on the boundary of $\Pi $\ for a sufficiently large $n$.
Therefore, a Taylor expansion of $\sqrt{n}\tfrac{\partial L_{n}\left( 
\widehat{\Lambda }_{1n},\Lambda _{\ast },\widehat{\theta }_{1n}\right) }{%
\partial \Lambda }$ around $\Lambda _{0}$ implies that for some $\Lambda
_{u}\in \Pi $ between $\widehat{\Lambda }_{1n}$ and $\Lambda _{0}$,%
\begin{eqnarray}
0 &=&\sqrt{n}\tfrac{\partial L_{n}\left( \widehat{\Lambda }_{1n},\Lambda
_{\ast },\widehat{\theta }_{1n}\right) }{\partial \Lambda }=\sqrt{n}\tfrac{%
\partial L_{n}\left( \widehat{\Lambda }_{1n},\Lambda _{\ast },\theta
_{0}\right) }{\partial \Lambda }+o_{a.s.}\left( 1\right)   \notag \\
&=&\sqrt{n}\tfrac{L_{n}\left( \Lambda _{0},\Lambda _{\ast },\theta
_{0}\right) }{\partial \Lambda }+\sqrt{n}\tfrac{L_{n}\left( \Lambda
_{u},\Lambda _{\ast },\theta _{0}\right) }{\partial \Lambda }\left( \widehat{%
\Lambda }_{1n}-\Lambda _{0}\right) +o_{a.s.}\left( 1\right) .  \label{A.20}
\end{eqnarray}
Finally, given (\ref{A.20}) the result is established while the following
two lemmas are shown.
\end{proof}

\begin{lemma}
 Under \textbf{A0**}, \textit{\textbf{A1}}, \textit{\textbf{A2}} and the finiteness of the fourth moments of the random coefficients, then
\begin{equation*}
\sqrt{n}\tfrac{L_{n}\left( \Lambda _{0},\Lambda _{\ast },\theta _{0}\right) 
}{\partial \Lambda }\underset{n\rightarrow \infty }{\overset{\mathcal{L}}{%
\rightarrow }}\mathcal{N}\left( 0,D\left( \Lambda _{\ast }\right) \right) .
\end{equation*}
\label{S.4}
\end{lemma}
\begin{proof}{(\textbf{Lemma \ref{S.4}})}\\
Clearly $\left\{ \tfrac{\partial l_{t}\left( \Lambda
_{0},\Lambda _{\ast },\theta _{0}\right) }{\partial \theta },t\in 
\mathbb{Z}
\right\} $ is a martingale difference with respect to $\left\{ \tciFourier
_{t}^{Y},t\in 
\mathbb{Z}
\right\} $ where 
\begin{equation*}
\tfrac{\partial l_{t}\left( \Lambda ,\Lambda _{\ast },\theta _{0}\right) }{%
\partial \theta }=-2\tfrac{ \partial V_{t}\left( \theta _{0},\Lambda
\right) }{\partial \Lambda }\tfrac{e_{t}^{2}-V_{t}\left( \theta
_{0},\Lambda \right) }{u_{t}^{2}\left( \theta _{0},\Lambda _{\ast }\right) }
\mbox{ and } 
\sqrt{n}\tfrac{L_{n}\left( \Lambda _{0},\Lambda _{\ast },\theta
_{0}\right) }{\partial \Lambda }=\sum_{t=1}^{n}\tfrac{1}{\sqrt{n}}\tfrac{%
\partial l_{t}\left( \Lambda _{0},\Lambda _{\ast },\theta _{0}\right) }{%
\partial \Lambda }. 
\end{equation*}
By \textbf{A1} and \textbf{A2}, it holds that $\mathbb{E}\left( \tfrac{\partial
l_{t}\left( \Lambda _{0},\Lambda _{\ast },\theta _{0}\right) }{\partial
\Lambda }\tfrac{\partial l_{t}\left( \Lambda _{0},\Lambda _{\ast },\theta
_{0}\right) }{\partial \Lambda ^{\prime }}\right) =D\left( \Lambda _{\ast
}\right)$, and the lemma follows from the martingale CLT.    
\end{proof}
\begin{lemma}
Under \textbf{A0}**, \textbf{A1} and \textbf{A2} it follows that 
$
\tfrac{\partial ^{2}L_{n}\left( \Lambda _{u},\Lambda _{\ast },\theta
_{0}\right) }{\partial \Lambda \partial \Lambda ^{\prime }}\underset{%
n\rightarrow \infty }{\overset{a.s.}{\rightarrow }}C\left( \Lambda _{\ast
}\right).
$
\label{S.5}
\end{lemma}
\begin{proof}{(\textbf{Lemma \ref{S.5}})}\\
Let $\mathcal{V}_{k}(\Lambda _{0})$ ($k\in 
\mathbb{N}
^{\ast }$) be the open ball of center $\Lambda _{0}$ and radius $1/k$.
Assume that $n$ is large enough so that $\Lambda _{u}$ belongs to $\mathcal{V%
}_{k}(\Lambda _{0})$. From the stationarity and ergodicity of 
\begin{eqnarray*}
\left\{
\inf_{\Lambda \in \mathcal{V}_{k}(\Lambda _{0})}\left\vert \tfrac{\partial
^{2}l_{t}\left( \Lambda _{0},\Lambda _{\ast },\theta _{0}\right) }{\partial
\Lambda _{i}\partial \Lambda _{j}}-\mathbb{E}\left( \tfrac{\partial
^{2}l_{t}\left( \Lambda _{0},\Lambda _{\ast },\theta _{0}\right) }{\partial
\Lambda _{i}\partial \Lambda _{j}}\right) \right\vert ,t\in 
\mathbb{Z}
\right\},    
\end{eqnarray*}
if follows that%
\begin{eqnarray*}
\left\vert \tfrac{\partial ^{2}L_{n}\left( \Lambda _{u},\Lambda _{\ast
},\theta _{0}\right) }{\partial \Lambda \partial \Lambda ^{\prime }}-C\left(
\Lambda _{\ast }\right) \right\vert  &\leq &\tfrac{1}{n}\sum_{t=1}^{n}\inf_{%
\Lambda \in \mathcal{V}(\Lambda _{0})}\left\vert \tfrac{\partial
^{2}l_{t}\left( \Lambda ,\Lambda _{\ast },\theta _{0}\right) }{\partial
\Lambda _{i}\partial \Lambda _{j}}-\mathbb{E}\left( \tfrac{\partial ^{2}l_{t}\left(
\Lambda _{0},\Lambda _{\ast },\theta _{0}\right) }{\partial \Lambda
_{i}\partial \Lambda _{j}}\right) \right\vert  \\
&&\underset{n\rightarrow \infty }{\overset{a.s.}{\rightarrow }}\mathbb{E}\left(
\inf_{\Lambda \in \mathcal{V}_{k}(\Lambda _{0})}\left\vert \tfrac{\partial
^{2}l_{t}\left( \Lambda ,\Lambda _{\ast },\theta _{0}\right) }{\partial
\Lambda _{i}\partial \Lambda _{j}}-\mathbb{E}\left( \tfrac{\partial ^{2}l_{t}\left(
\Lambda _{0},\Lambda _{\ast },\theta _{0}\right) }{\partial \Lambda
_{i}\partial \Lambda _{j}}\right) \right\vert \right) \text{.}
\end{eqnarray*}

Note that the Lebesgue-dominated convergence theorem entails 
\begin{eqnarray*}
&&\lim_{k\rightarrow \infty }\mathbb{E}\left( \inf_{\Lambda \in \mathcal{V}%
_{k}(\Lambda _{0})}\left\vert \tfrac{\partial ^{2}l_{t}\left( \Lambda
,\Lambda _{\ast },\theta _{0}\right) }{\partial \Lambda _{i}\partial \Lambda
_{j}}-\mathbb{E}\left( \tfrac{\partial ^{2}l_{t}\left( \Lambda _{0},\Lambda _{\ast
},\theta _{0}\right) }{\partial \Lambda _{i}\partial \Lambda _{j}}\right)
\right\vert \right)  \\
&=&\mathbb{E}\left( \lim_{k\rightarrow \infty }\sup_{\Lambda \in \mathcal{V}%
_{k}(\Lambda _{0})}\left\vert \tfrac{\partial ^{2}l_{t}\left( \Lambda
,\Lambda _{\ast },\theta _{0}\right) }{\partial \Lambda _{i}\partial \Lambda
_{j}}-\mathbb{E}\left( \tfrac{\partial ^{2}l_{t}\left( \Lambda _{0},\Lambda _{\ast
},\theta _{0}\right) }{\partial \Lambda _{i}\partial \Lambda _{j}}\right)
\right\vert \right) =0,
\end{eqnarray*}%
which completes the proof of the result. 
\end{proof} 

\end{appendices}
\newpage
\begin{center}
\Large{Supplementary Material}    
\end{center}
\begin{center}
(Table \ref{table11} here)
\end{center}
\begin{table}[ht]
\begin{center}
\begin{tabular}{lc|ccccc}
\hline\hline
& $n_{c}$ & $300$ & $420$ & $520$ & $680$ & $700$ \\ \hline
A RMINAR(1) & $%
\begin{array}{c}
\text{MSFE} \\ 
\text{MAFE} \\ 
\text{MSPFE}%
\end{array}%
$ & \multicolumn{1}{|l}{$%
\begin{array}{c}
246.9545 \\ 
9.4164 \\ 
1.360%
\end{array}%
$} & \multicolumn{1}{l}{$%
\begin{array}{c}
318.5958 \\ 
10.4993 \\ 
1.9864%
\end{array}%
$} & \multicolumn{1}{l}{$%
\begin{array}{c}
413.4456 \\ 
11.8362 \\ 
2.7857%
\end{array}%
$} & \multicolumn{1}{l}{$%
\begin{array}{c}
322.4681 \\ 
13.5432 \\ 
1.5346%
\end{array}%
$} & \multicolumn{1}{l}{$%
\begin{array}{c}
377.6246 \\ 
13.4922 \\ 
1.8009%
\end{array}%
$} \\ \hline
A RMINAR(2) & $%
\begin{array}{c}
\text{MSFE} \\ 
\text{MAFE} \\ 
\text{MSPFE}%
\end{array}%
$ & \multicolumn{1}{|l}{$%
\begin{array}{c}
243.8918 \\ 
9.4024 \\ 
1.4137%
\end{array}%
$} & \multicolumn{1}{l}{$%
\begin{array}{c}
314.6527 \\ 
10.4726 \\ 
2.0281%
\end{array}%
$} & \multicolumn{1}{l}{$%
\begin{array}{c}
406.6413 \\ 
11.7646 \\ 
2.7375%
\end{array}%
$} & \multicolumn{1}{l}{$%
\begin{array}{c}
316.6054 \\ 
13.2168 \\ 
1.7056%
\end{array}%
$} & \multicolumn{1}{l}{$%
\begin{array}{c}
381.6725 \\ 
13.1774 \\ 
2.2566%
\end{array}%
$} \\ \hline
A RMINAR(3) & $%
\begin{array}{c}
\text{MSFE} \\ 
\text{MAFE} \\ 
\text{MSPFE}%
\end{array}%
$ & \multicolumn{1}{|l}{$%
\begin{array}{c}
241.7308 \\ 
9.4727 \\ 
1.3535%
\end{array}%
$} & \multicolumn{1}{l}{$%
\begin{array}{c}
295.1503 \\ 
10.2364 \\ 
1.7079%
\end{array}%
$} & \multicolumn{1}{l}{$%
\begin{array}{c}
398.7492 \\ 
11.7017 \\ 
2.3473%
\end{array}%
$} & \multicolumn{1}{l}{$%
\begin{array}{c}
309.7075 \\ 
13.3228 \\ 
1.2327%
\end{array}%
$} & \multicolumn{1}{l}{$%
\begin{array}{c}
366.5074 \\ 
13.1845 \\ 
1.4226%
\end{array}%
$} \\ \hline \hline
\end{tabular}
\caption{Out-of-sample forecasting ability of the additive (A) RMINAR$\left( p\right)$ models on the ExRate time series.}
\label{table11}
\end{center}
\end{table}

\end{document}